\documentclass[twoside,11pt]{article}

\newif\ifdraft
\draftfalse

\usepackage{jair,rawfonts}
\usepackage[round,sort]{natbib}
\usepackage{amssymb}
\usepackage{placeins}
\usepackage[lined,boxed,commentsnumbered]{algorithm2e}
\SetAlgorithmName{Pseudocode}{}{}
\usepackage{graphicx}
\usepackage{amsfonts}
\usepackage{amsmath}
\usepackage{amsthm}
\usepackage{tabularx,booktabs}
\usepackage{multicol}
\usepackage{multirow}
\usepackage{subcaption}
\usepackage{dcolumn}
\usepackage[title]{appendix}
\usepackage{xurl}
\usepackage{wasysym}
\usepackage{comment}
\usepackage{cleveref}
\crefname{section}{§}{§§}
\Crefname{section}{§}{§§}

\newtheorem{prop}{Proposition}
\newtheorem{remark}{Remark}
\newtheorem{definition}{Definition}
\newtheorem{example}{Example}

\newcommand{\mysize}{0.84}

\usepackage{rotating} 
\newcommand{\side}[1]{\begin{sideways}{#1}\end{sideways}}
\usepackage{tikz}
\usetikzlibrary{fit,calc}
\newcommand*{\tikzmk}[1]{\tikz[remember picture,overlay,] \node (#1) {};\ignorespaces}
\newcommand{\boxit}[1]{\tikz[remember picture,overlay]{\node[yshift=3pt,fill=#1,opacity=.25,fit={(A)($(B)+(.75\linewidth,.8\baselineskip)$)}] {};}\ignorespaces}
\colorlet{pink}{red!40}

\usepackage[normalem]{ulem} 
 \newcommand{\strikeout}[1]{\ifdraft{\sout{#1}}\else{\vspace{0ex}}\fi}
\newcommand{\blue}[1]{\ifdraft{\leavevmode\color{blue}{#1}}\else{\leavevmode\color{black}{#1}}\fi}
\newcommand{\wasblue}[1]{\ifdraft{\leavevmode\color{black}{#1}}\else{\leavevmode\color{black}{#1}}\fi}
\newcommand{\replace}[2]{\ifdraft{\leavevmode\color{blue}\strikeout{#1} \leavevmode\color{blue}{#2}}\else{\leavevmode\color{black}{#2}}\fi}

\newcommand{\TP}{\mathrm{TP}}
\newcommand{\TN}{\mathrm{TN}}
\newcommand{\FP}{\mathrm{FP}}
\newcommand{\FN}{\mathrm{FN}}

\jairheading{76}{2023}{1-65}{06/2022}{03/2023}
\ShortHeadings{Measuring Fairness Under Unawareness Using Quantification}
{Fabris, Esuli, Moreo, \& Sebastiani}
\firstpageno{25}

\begin{document}

\title{Measuring Fairness \\ Under Unawareness of Sensitive
Attributes: \\ A Quantification-Based Approach}

\author{\name Alessandro Fabris \email fabrisal@dei.unipd.it \\
\addr Max Planck Institute for Security and Privacy \\
Universitätsstraße 140, 44799 Bochum, Germany \\
Department of Information Engineering, University of Padova\\
Via Giovanni Gradenigo 6B, Padua, 35131, Italy \AND
\name Andrea Esuli \email andrea.esuli@isti.cnr.it \\
\name Alejandro Moreo \email alejandro.moreo@isti.cnr.it \\
\name Fabrizio Sebastiani \email fabrizio.sebastiani@isti.cnr.it \\
\addr Istituto di Scienza e Tecnologie dell'Informazione \\ Consiglio
Nazionale delle Ricerche \\ Via Giuseppe Moruzzi 1, Pisa, 56124,
Italy}

% For research notes, remove the comment character in the line below.
% \researchnote

\maketitle

\begin{abstract}
  \noindent Algorithms and models are increasingly deployed to inform
  decisions about people, inevitably affecting their lives. As a
  consequence, those in charge of developing these models must
  carefully evaluate their impact on different groups of people and
  favour \emph{group fairness}, that is, ensure that groups determined
  by sensitive demographic attributes, such as race or sex, are not
  treated unjustly. To achieve this goal, the availability
  (\emph{awareness}) of these demographic attributes to those
  evaluating the impact of these models is fundamental. Unfortunately,
  collecting and storing these attributes is often in conflict with
  industry practices and legislation on data minimisation and
  privacy. For this reason, it can be hard to measure the group
  fairness of trained models, even from within the companies
  developing them. In this work, we tackle the problem of measuring
  group \emph{fairness under unawareness} of sensitive attributes, by
  using techniques from \textit{quantification}, a supervised learning
  task concerned with directly providing group-level prevalence
  estimates (rather than individual-level class labels). We show that
  quantification approaches are particularly suited to tackle the
  fairness-under-unawareness problem, as they are robust to inevitable
  distribution shifts while at the same time decoupling the
  (desirable) objective of measuring group fairness from the
  (undesirable) side effect of allowing the inference of sensitive
  attributes of individuals. More in detail, we show that fairness under unawareness can be cast as a
  quantification problem and solved with proven methods from the
  quantification literature. We show that these methods outperform
  previous approaches to measure demographic parity in five
  experimental protocols, corresponding to important challenges that
  complicate the estimation of classifier fairness under unawareness.
\end{abstract}

\setcounter{page}{1}
\section{Introduction}
\label{sec:intro}

\noindent The widespread adoption of algorithmic decision-making in
high-stakes domains has determined an increased attention to the
underlying algorithms and their impact on people, with attention to
sensitive (or ``protected'') groups. Typically, sensitive groups are
subpopulations determined by salient social and demographic factors,
such as race or sex. The unfair treatment of such groups is not only
unethical, but also ruled out by anti-discrimination laws, and is thus
studied by a growing community of algorithmic fairness
researchers. Important works in this area have addressed the unfair
treatment of subpopulations that may arise in the judicial system
\citep{angwin2016machine,larson2016how,berk2021fairness}, healthcare
\citep{obermeyer2019dissecting,gervasi2022potential,ricci2022addressing}, search engines
\citep{geyik2019fairnessaware,fabris2020gender,ekstrand2022fairness}, insurance
\citep{angwin2017:mn,donahue2021better,fabris2021algorithmic}, and computer vision
\citep{buolamwini2018:gender,raji2019actionable,goyal2022fairness}, just to name a few
domains that may be affected. One common trait of these research works
is their attention to a careful definition (and subsequent
measurement) of what it means for a model to be fair to the subgroups
involved (\textit{group fairness}), which is typically viewed in terms
of differences, across the salient subpopulations, in quantities of
interest such as accuracy, recall, or acceptance rate. According to
popular definitions of fairness, sizeable such differences (e.g.,
between women and men) correspond to low fairness on the part of the
algorithm
\citep{dwork2012:fairness,barocas2019fair,pedreschi2008discrimination}.
% \fabsebcomment{Noi usiamo spesso la parola \emph{model}, che è più
% generale di \emph{classifier}, perché e.g., un \emph{recommender} è
% anch'esso un \emph{model}. Secondo me dovremmo specificare già
% dall'introduzione che noi ci occupiamo di \emph{classifiers}, e non
% di \emph{recommenders}, e io lo farei usando sempre, fin
% dall'introduzione, la parola \emph{classifier} al posto della parola
% \emph{model}.}

Unfortunately, sensitive demographic data, such as the race or sex of
subjects, are often not available, since practitioners find several
barriers to obtaining these data, both during model development and
after deployment. Among these barriers, legislation plays a major
role, prohibiting the collection of sensitive attributes in some
domains \citep{bogen2021:ap}. Even in the absence of explicit
prohibition, privacy-by-design standards and a data minimization ethos
often push companies in the direction of avoiding the collection of
sensitive data from their customers. Similarly, the prospect of
negative media coverage is a clear concern, so companies often err on
the side of caution and inaction \citep{andrus2021:wm}. The
unavailability of these data thus makes the measurement of model
fairness nontrivial, even for the company that is developing and/or
deploying the model. For these reasons, in a recent survey of industry
professionals, most of the respondents stated that the availability of
tools that support fairness auditing in the absence of
individual-level demographics would be very useful
\citep{holstein2019:if}.
% Moreover, noisy estimates of sensitive attributes naïvely plugged
% into fairness-enhancing schemes may actually make a model
% \emph{less} fair \citep{ghosh2021when,mehrotra2021:mitigating}.
In other words, the problem of measuring group fairness when the
values of the sensitive attributes are unknown (\textit{fairness under
unawareness}) is pressing and requires ad hoc solutions.

In the literature on algorithmic fairness, much work has been done to
propose techniques directly aimed at improving the fairness of a model
\citep{zafar2017fairness,donini2018empirical,hashimoto2018fairness,he2020geometric,hardt2016:equality,sankar2021matching}. However,
relatively little attention has been paid to the problem of reliably
measuring fairness. This represents an important, but rather
overlooked, preliminary step to enforce fairness and make algorithms
more equitable across groups. More recent works have studied non-ideal
conditions, such as missing data \citep{goel2021importance}, noisy or
missing group labels \citep{chen2019:fu,awasthi2020equalized}, and
non-iid samples \citep{singh2021fairness,rezaei2021robust}, and showed
that naïve fairness-enhancing algorithms may actually make a model
\emph{less} fair under noisy demographic information
\citep{ghosh2021when,mehrotra2021:mitigating}.

% The problem of measuring differences in outcomes across sensitive
% demographics has a long tradition in health research
% \citep{andjelkovich1990:ar,fiscella2006:ug}. One notable example is
% the Bayesian Improved Surname Geocoding (BISG) method
% \citep{elliott2009:uc}, which was developed for the health sector and
% later adopted for an external fairness audit of interest rates set
% by auto lending companies \citep{cfpb2013cfpb,cfpb2014using}. Only
% recently this line of work has converged with algorithmic fairness
% research, where a handful of works have studied the problem of
% reliably measuring group-wise disparities of classifiers without
% access to sensitive
% attributes\citep{chen2019:fu,kallus2020:aa,aswasthi2021:ef}.

% The problem of reliably measuring the group fairness of classifiers
% without access to sensitive attributes is considered in three recent
% works \citep{chen2019:fu, kallus2020:aa, aswasthi2021:ef}, following
% the deployment of related techniques by US consumer protection
% agencies \citep{cfpb2013cfpb, cfpb2014using}. \fabsebcomment{Questo
% paragrafetto qui sopra non viene approfondito, e quindi io lo
% toglierei del tutto; tanto poi questi lavori li approfondiamo nel
% related work.}

In this work, we propose a novel solution to the problem of measuring
classifier fairness under unawareness by using techniques from
\textit{quantification} \citep{Gonzalez:2017it}, a supervised learning
task concerned with estimating, rather than the class labels of
individual data points, the class prevalence values for samples of
such data points, i.e., group-level quantities, such as the percentage
of women in a given sample. Quantification methods address two
pressing facets of the fairness under unawareness problem: (1) their
estimates are robust to \emph{distribution shift} (i.e., to the fact
that the distribution of the labels in the unlabeled data may
significantly differ from the analogous distribution in the training
data), which is often inevitable since populations evolve, and
demographic data are unlikely to be representative of every condition
encountered at deployment time; (2) they allow the estimation of
group-level quantities but do not allow the inference of sensitive
attributes at the individual level, which is beneficial since the
latter might lead to the inappropriate and nonconsensual utilization
of this sensitive information, reducing individuals' agency over data
\citep{andrus2022demographic}. Quantification methods achieve these
goals by \emph{directly} targeting group-level prevalence
estimates. They do so through a variety of approaches, including,
e.g., dedicated loss functions, task-specific adjustments,
and \textit{ad hoc} model selection procedures.
% , or minimization of the divergence between the cumulative
% distributions of the posterior probabilities of the validation set
% and the unlabelled set.  \fabsebcomment{Questa ultima frase si
% potrebbe fose anche tagliare?}  \strikeout{; including aggregation
% of posterior probabilities and error compensation via the estimation
% of false positive and false negative rates.}\fabsebcomment{Questa
% ultima frase ho proposto di tagliarla perché fa sembrare la
% quantification una cosa triviale.}

% This is precisely the goal of practitioners looking to measure
% fairness under unawareness of sensitive attributes. When auditing an
% algorithm for group fairness, the aim is not the development of a
% model that is accurate for individual predictions (i.e.,
% classification), which may be misused to infer people's
% demographics, such as a user's race, and may thus lead to the
% inappropriate and non-consensual utilization of this
% information. Rather, the central interest of fairness audits is the
% reliable estimation of group-level quantities (i.e.,
% quantification), such as the prevalence of women among the instances
% to which a certain class has been assigned by the model.

%\wasreplace{We concentrate on a notion of fairness called
%\emph{demographic disparity}, defined as the difference in the
%acceptance rates of the classifier (that is, the rate at which a
%classifier assigns the positive class) among the relevant
%sub-populations. In order to measure the demographic disparity of a
%classifier, we leverage different solutions proposed in the
%quantification literature.}{}

Overall, we make the following contributions:

\begin{itemize}
\item \textbf{Quantifying fairness under unawareness}. We show that measuring fairness under
  unawareness can be cast as a quantification problem and solved with
  approaches of proven consistency established in the quantification
  literature (Section \ref{sec:method}). We propose and demonstrate several high-accuracy fairness estimators for both vanilla and fairness-aware classifiers.
  % We prove that our approach is a generalised \blue{(and much
  % improved)} version of the Weighted Estimator \citep{chen2019:fu}
  % . \fabsebcomment{E' il caso di metterla, quest'ultima frase? Fa
  % sembrare, right from the start. il nostro approccio una banale
  % generalizzazione di qualcosa di già noto. Per me si può anche
  % togliere.}
  % We prove the importance of each component through an ablation
  % study (Section~\ref{sec:ablation}).
\item \textbf{Experimental protocols for five major
  challenges}. Drawing from the algorithmic fairness literature, we
  identify five important challenges that arise in estimating fairness
  under unawareness. These challenges are encountered in real-world
  applications, and include the nonstationarity of the processes
  generating the data and the variable cardinality of the available
  samples. For each such challenge, we define and formalise a precise
  experimental protocol, through which we compare the performance of
  quantifiers (i.e., group-level prevalence estimators) generated by
  six different quantification methods
  (Sections~\ref{sec:sample_prev_d3}--\ref{sec:flip_prev_d1}).
\item %\textbf{Decoupling.}
  \textbf{Decoupling group-level and individual-level inferences.} We
  consider the problem of potential model misuse to maliciously infer
  demographic characteristics at an individual level, which represents
  a concern for \emph{proxy methods}, i.e., methods that measure model fairness based on proxy
  attributes. Proxy methods are estimators of sensitive attributes
  which exploit the correlation between available attributes (e.g.,
  ZIP code) and the sensitive attributes (e.g., race) in order to
  infer the values of the latter. Through a set of experiments, we
  demonstrate two methods that yield precise estimates of demographic
  disparity but poor classification performance, thus decoupling the
  (desirable) objective of group-level prevalence estimation from the
  (undesirable) objective of individual-level class label prediction
  (Section~\ref{sec:q_not_c}).
\end{itemize}

% \noindent \fabsebcomment{Al di là del fatto che, notoriamente, io
% trovo il paragrafo qui sotto fuori posto (lo vedrei bene nelle
% conclusioni), non capisco proprio cosa voglia dire e cosa c'entri
% col nostro approccio; magari me lo spiegherete a voce ...}

\noindent 
It is worth noting from the outset some intrinsic limitations of proxy
methods and measures of group fairness. In essence, proxy methods
exploit the co-occurrence of membership in a group and display of a
given trait, potentially learning, encoding and reinforcing
stereotypical associations \citep{lipton2018does}. More in general,
even when labels for sensitive attributes are available, these are not
all equivalent. Self-reported labels are preferable to avoid external
assignment (i.e., inference of sensitive attributes), which can be
harmful in itself \citep{keyes2018:misgendering}. In broader terms,
approaches that define sensitive attributes as rigid and fixed
categories are limited in that they impose a taxonomy onto people,
erasing the needs and experiences of those who do not fit the
envisioned prevalent categories \citep{namaste2000invisible}. Although
we acknowledge these limitations, we hope that our work will help
highlight, investigate, and mitigate unfavourable outcomes for
disadvantaged groups caused by different automated decision-making
systems.
% \fabsebcomment{Io di questo paragrafo (da ``It is worth noting ...")
% ci ho capito poco, forse perché non sono un ``iniziato" al mondo
% della fairness; e comunque lo vedrei meglio nelle conclusioni.}
% \afabcomment{Lo so, però c'è una parte della comunità fairness
% atttenta a queste limitazioni e il paragrafo è utile qui per
% appianare subito potenziali osservazioni/fastidi del lettore}
% \andcomment{Per me è abbastanza chiaro. Si notano due limitazioni
% intrinseche del problema, il possibile bias nelle label del dataset
% quando queste non sono self-assigned, e il possibile bias nella
% definizione del label set. Lo scopo è quello cautelativo scritto da
% Alessandro.}

The outline of this work is as follows. Section \ref{sec:prelimnaries}
summarizes the notation and background for this article. Section
\ref{sec:related} introduces related works. After giving a primer on
quantification, with emphasis on the approaches we consider in this
work, Section~\ref{sec:method} shows how these approaches can be
leveraged to measure fairness under unawareness of
sensitive attributes. Section~\ref{sec:experiments} presents our
experiments, in which we tackle, one by one, each of the five major
challenges mentioned above. We then summarize and discuss these
results (Section~\ref{sec:discussion}) and present concluding remarks
(Section~\ref{sec:conclusion}), describing key limitations and
directions for future work.

% \noindent \textbf{Reproducibility}.
Our code is available at
\url{https://github.com/alessandro-fabris/ql4facct}.

\section{Preliminaries}
\label{sec:prelimnaries}

\subsection{Notation}
\label{sec:notation}

\begin{table}[tb]
  \caption{Main notational conventions used in this work.}
  \label{tab:notation}
  \begin{center}
    % \resizebox{\textwidth}{!} {
    \begin{tabular}{|c|p{10cm}|}
      % \begin{tabular}{|r|l|}
      \hline
      $\mathbf{x} \in \mathcal{X}$ & a data point, i.e.,  a vector of non-sensitive attribute values \\
      $s \in \mathcal{S}$ & a value for the sensitive attribute 
                            , with $\mathcal{S}=\{0,1\}$ \\
      $y \in \mathcal{Y}$ & a class from the target domain
                            $\mathcal{Y}=\{\ominus,\oplus\}$ \\
      $X, S, Y, \hat{Y}$ & random variables for data points,
                           non-sensitive attributes,
                           classes, and class predictions \\
      $h(\mathbf{x})$ & 
                        a classifier $h:\mathcal{X} \rightarrow \mathcal{Y}$ 
                        issuing predictions in $\mathcal{Y}$ for data points in $\mathcal{X}$ \\ 
      $k(\mathbf{x})$ & 
                        a classifier $k:\mathcal{X} \rightarrow \mathcal{S} $ 
                        issuing predictions in $\mathcal{S}$ for data points in $\mathcal{X}$ \\ 
      $\sigma$ & a sample, i.e., a non-empty set of data 
                 points drawn from $\mathcal{X}$ \\
      $p_{\sigma}(s)$ & true prevalence of sensitive attribute value 
                        $s$ in sample $\sigma$ \\
      $\hat{p}_{\sigma}(s)$ & estimate of the prevalence 
                              of sensitive attribute value $s$ 
                              in sample $\sigma$ \\
      $\hat{p}_{\sigma}^q(s)$ & estimate $\hat{p}_{\sigma}(s)$ 
                                obtained via quantifier $q$ \\
      $q(\sigma)$ & a quantifier $q:2^\mathcal{X} \rightarrow [0,1]$ 
                    estimating the prevalence of the positive class of 
                    sensitive attribute $S$ in a sample \\
      $\mathcal{D}_1$ & set of pairs $(\mathbf{x}_i,y_i) \in (\mathcal{X}, 
                        \mathcal{Y})$ for training classifier $h(\mathbf{x})$ \\
      $\mathcal{D}_2$ & set of pairs $(\mathbf{x}_i,s_i) \in (\mathcal{X}, 
                        \mathcal{S})$ for training quantifier $q(\sigma)$ \\
      $\mathcal{D}_3$ & set of points $\mathbf{x}_i \in \mathcal{X}$ to which $h(\mathbf{x})$ and $q(\sigma)$ are to be applied \\
      $\mathcal{D}_{2}^y$ & short for 
                            $\mathcal{D}_{2}^{\hat{Y}=y}=\{ (\mathbf{x}_i,s_i) 
                            \in \mathcal{D}_{2} \;|\; h(\mathbf{x}_i)=y \}$ \\
      $\mathcal{D}_{3}^y$ & short for 
                            $\mathcal{D}_{3}^{\hat{Y}=y}=\{ \mathbf{x}_i
                            \in \mathcal{D}_{3} \;|\; h(\mathbf{x}_i)=y \}$ \\
      $\breve{\mathcal{D}} \;$ & a set
                                 derived from $\mathcal{D}$ according to an experimental 
                                 protocol among those detailed in Sections 
                                 \ref{sec:sample_prev_d3}--\ref{sec:flip_prev_d1}\\
 %
                                 % $\breve{\mathcal{D}}_1$ & a
                                 % modified training set, derived from
                                 % $\mathcal{D}_1$ according to an
                                 % experimental protocol. Protocols
                                 % are detailed in Sections
                                 % \ref{sec:sample_prev_d3}--\ref{sec:flip_prev_d1}\\
                                 % $\breve{\mathcal{D}}_{2},\breve{\mathcal{D}}_{3}$,
                                 % $\breve{\mathcal{D}}_{2}^\oplus,$ &
                                 % \multirow{2}{*}{analogous sets}\\
                                 % $\breve{\mathcal{D}}_{2}^\ominus,
                                 % \breve{\mathcal{D}}_{3}^\oplus,
                                 % \breve{\mathcal{D}}_{3}^\ominus$ &
                                 % \\
      \hline
    \end{tabular}
    % }
  \end{center}
\end{table}

\noindent In this paper, we use the following notation, summarized in
Table~\ref{tab:notation}. By $\mathbf{x}$ we indicate a data point
drawn from a domain $\mathcal{X}$, represented via a set $X$ of
nonsensitive attributes (i.e., features).
% \strikeout{; by $X$ we indicate a random variable that takes values
% in $\mathcal{X}$}. $\mathcal{S}$ to denote the (binary) domain of a
% sensitive attribute, that we indicate as $\mathcal{S}=\{0,1\}$ for
% ease of exposition, and by $s$ a value that $\mathcal{S}$ may take.
We use $S$ to denote a sensitive attribute that takes values in
$\mathcal{S}=\{0,1\}$, and by $s\in\mathcal{S}$ a value that $S$ may
take.\footnote{\blue{Note that, for ease of exposition, we consider only one binary sensitive attribute; our approach straighforwardly applies to more complex settings (see Remark \ref{rem:nonbinary}).}}
% \blue{Note also that we use the term ``attribute'' with some
% liberty, sometimes meaning a feature in a vector space (as in
% ``fairness with respect to a sensitive attribute'') and sometimes
% meaning the value that the feature takes for a specific data point
% (as in ``unawareness of sensitive attributes''); in these cases, the
% precise meaning is clear from the context.} \fabsebcomment{Che ne
% direste di usare $\mathcal{S}=\{\female,\male\}$ invece di
% $\mathcal{S}=\{0,1\}$? Un sacco di formule risulterebbero più
% immediate da capire. Potremmo dire ``We use $\mathcal{S}$ to denote
% the (binary) domain of a sensitive attribute; for ease of exposition
% we assume that the sensitive attribute is ``sex'', and we thus
% indicate its domain as $\mathcal{S}=\{\female,\male\}$. By $s$ we
% indicate a value that $\mathcal{S}$ may take.''.}
By $Y$ we indicate a class (representing the target of a prediction
task) taking values in a binary domain
$\mathcal{Y}=\{\ominus,\oplus\}$, and by $y\in\mathcal{Y}$ a value
that $Y$ can take.
% Moreover, we focus on the case in which the classifier that we want
% to audit is a binary one, but the definitions and techniques we
% employ can be straightforwardly extended to a multiclass
% setting. See also Footnote~\ref{foot:multiclass} on this.} Given
% $\mathbf{x}\in\mathcal{X}$ and $y\in\mathcal{Y}$, a pair
% $(\mathbf{x},y)$ thus denotes a data point with its true class
% label.
%
The symbol $\sigma$ denotes a \emph{sample}, i.e., a non-empty set of
data points drawn from $\mathcal{X}$. By $p_{\sigma}(s)$ we indicate
the true prevalence of an attribute value $s$ in the sample $\sigma$,
% by $\hat{p}_{\sigma}(y)$ we indicate an estimate of this prevalence,
while by $\hat{p}_{\sigma}^{q}(s)$ we indicate the estimate
% \footnote{Consistently with most mathematical literature, we use the
% caret symbol (\^\/\/) to indicate estimation.}
of this prevalence obtained by means of a quantifier $q$, which we
define as a function $q:2^\mathcal{X} \rightarrow [0,1]$.
% quantification method $M$.
Since $0\leq p_{\sigma}(s) \leq 1$ and
$0\leq \hat{p}_{\sigma}^{q}(s) \leq 1$ for all $s\in\mathcal{S}$, and
since
$\sum_{s\in\mathcal{S}}p_{\sigma}(s) =
\sum_{s\in\mathcal{S}}\hat{p}_{\sigma}^{q}(s)=1$, the
$p_{\sigma}(s)$'s and the $\hat{p}_{\sigma}^{q}(s)$'s form two
probability distributions in $\mathcal{S}$. We also introduce the
random variable $\hat{Y}$, which denotes a predicted label. By
$\Pr(V=v)$ we indicate, as usual, the probability that a random
variable $V$ takes value $v$, which we shorten as $\Pr(v)$ when $V$ is
clear from the context, since $X,S,Y$ can also be seen as random
variables. By $h:\mathcal{X}\rightarrow \mathcal{Y}$ we indicate a
binary classifier that assigns classes in $\mathcal{Y}$ to data points
in $\mathcal{X}$; by $k:\mathcal{X}\rightarrow \mathcal{S}$ we instead
indicate a binary classifier that assigns sensitive attribute values
in $\mathcal{S}$ to data points (e.g., that predicts the sensitive
attribute value of a certain data item $\mathbf{x}$). It is worth
re-emphasizing that both $h$ and $k$ only use nonsensitive attributes
$X$ as input variables,
% \blue{since the values of the sensitive attribute $S$ are assumed
% not available}.
For ease of use, we will interchangeably write $h(\mathbf{x})=y$ or
$h_{y}(\mathbf{x})=1$, and $k(\mathbf{x})=s$ or $k_{s}(\mathbf{x})=1$.

\subsection{Background}

\noindent Several criteria for group fairness have been proposed in
the machine learning literature, typically requiring equalization of
some conditional or marginal property of the distribution of sensitive
variable $S$, ground truth $Y$, and classifier estimate $\hat{Y}$
\citep{dwork2012:fairness,hardt2016:equality,narayanan2018:21}.
% \wasreplace{\emph{Demographic parity}, also called \textit{statistical
% parity} or \textit{independence} \citep{dwork2012:fairness,
% barocas2019fair, emelianov2022fair}, concentrates on the acceptance
% rate of a classifier conditional on the protected attribute. We
% adopt the following definition from \citet{chen2019:fu}:}
\wasblue{The main criteria of observational group fairness
\citep{barocas2019fair}, i.e., the ones computed directly from
groupwise confusion matrices, are defined as follows:}
\begin{definition}
  \label{def:dd}
  \wasblue{Given a classifier $h:\mathcal{X}\rightarrow \mathcal{Y}$
  issuing predictions $\hat{y}=h(\mathbf{x})$, and given the
  respective ground truth labels $y$, the following groupwise
  disparities with respect to attribute $S$ can be defined.}  \wasblue{
  \small{
  \begin{align*}
    \text{\emph{Demographic Disparity: }} \delta_{h}^{S\mathrm{, DD}} &= \Pr(\hat{Y}=\oplus|S=1) - \Pr(\hat{Y}=\oplus|S=0) \\
    \text{\emph{True Positive Rate Disparity: }} \delta_{h}^{S\mathrm{, TPRD}} &= \Pr(\hat{Y}=\oplus|S=1, Y=\oplus) - \Pr(\hat{Y}=\oplus|S=0, Y=\oplus) \\
    \text{\emph{True Negative Rate Disparity: }} \delta_{h}^{S\mathrm{, TNRD}} &= \Pr(\hat{Y}=\ominus|S=1, Y=\ominus) - 
                                                                                 \Pr(\hat{Y}=\ominus|S=0, Y=\ominus) \\
    \text{\emph{Positive Predicted Value Disparity: }} \delta_{h}^{S\mathrm{, PPVD}} &= \Pr(Y=\oplus|S=1, \hat{Y}=\oplus) - \Pr(Y=\oplus|S=0, \hat{Y}=\oplus) \\
    \text{\emph{Negative Predicted Value Disparity: }} \delta_{h}^{S\mathrm{, NPVD}} &= \Pr(Y=\ominus|S=1, \hat{Y}=\ominus) - 
                                                                                       \Pr(Y=\ominus|S=0, \hat{Y}=\ominus)\\
  \end{align*}
  } }
  \hfill\qed
\end{definition}
\noindent
% These are the main criteria of observational group fairness in the
% literature \citep{barocas2019fair}. Demographic Disparity falls under
% the \emph{independence} criterion, True Positive Rate Disparity and
% True Negative Rate Disparity under \emph{separation}, and Positive
% Predicted Value Disparity and Negative Predicted Value Disparity
% under \emph{sufficiency}. In this work we show how to estimate
% measures of group fairness under unawareness via quantification
% methods, and provide an extensive evaluation It is worth
% reemphasising that the sensitive attribute $S$ does \emph{not}
% typically belong to the set of attributes $\mathcal{X}$ which
% generate the feature space on which classifier $h$ operates (in
% other words, when training $h$ we are \emph{unaware} of $S$).
Demographic disparity, for example, measures whether the prevalence of
the positive class is the same across subpopulations identified by the
sensitive attribute $S$; a value
$\delta_{h}^{S\wasblue{\mathrm{, DD}}}=0$ indicates maximum fairness,
while values of $\delta_{h}^{S\wasblue{\mathrm{, DD}}}=-1$ or
$\delta_{h}^{S\wasblue{\mathrm{, DD}}}=+1$ indicate minimum fairness,
i.e., maximum advantage for $S=0$ over $S=1$ or vice versa. We
illustrate the problem of measuring fairness under unawareness using
an example focused on demographic disparity.

\begin{example}
  \label{ex:bank}
  Assume that $S$ stands for ``race'', $S=1$ for ``African-American''
  and $S=0$ for ``White'',\footnote{While acknowledging its
  limitations \citep{hephzibah2018race}, we follow the race
  categorization adopted by the US Census Bureau wherever possible.}
  and that the classifier, deployed by a bank, is responsible for
  recommending loan applicants for acceptance, classifying them as
  ``grant'' ($\oplus$) or ``deny'' ($\ominus$). For simplicity, let us
  assume that the outcome of the classifier will be translated
  directly into a decision without human supervision. The bank might
  want to check that the fraction of loan recipients out of the total
  number of applicants is approximately the same in the
  African-American and White subpopulations. In other words, the bank
  might want $\delta_{h}^{S\wasblue{\mathrm{, DD}}}$ to be close to 0. Of
  course, if the bank is aware of the race of each applicant, this
  constraint is very easy to check and, potentially, enforce. If the
  bank is unaware of the applicants' race, the problem is not trivial,
  and can be addressed by the method we propose in this paper.
\end{example}
% -------------------------------------------------------------------

\section{\wasblue{Related Work}}
\label{sec:related}

\subsection{Fairness Under Unawareness}
\label{sec:fairnessunderunawareness}

% \afabcomment{May want to mention this work
% \url{https://arxiv.org/pdf/2105.09985.pdf} on conditional parity
% (emphasis on covariates)}

\noindent Unavailability of sensitive attribute values poses a major
challenge for internal and external fairness audits. When these values
are unknown, it is sometimes possible to seek expert advice to obtain
them \citep{buolamwini2018:gender}. Alternatively, disclosure
procedures have been proposed for subjects to provide their sensitive
attributes to a trusted third party \citep{veale2017:fm} or to share
them encrypted \citep{kilbertus2018:bj}. Another line of research
studies the problem of reliably estimating measures of group fairness,
in classification \citep{chen2019:fu,kallus2020:aa,aswasthi2021:ef} and
ranking \citep{kirnap2021estimation,ghazimatin2022measuring}, without
access to sensitive attributes, via proxy
variables. %\afabcomment{\url{https://arxiv.org/abs/2108.05152} this may also be relevant (in the context of ranking) proposing a sampling strategy for annotation.}

\citep{chen2019:fu} is the work most closely related to ours. The
authors study the problem of estimating the demographic disparity of a
classifier, exploiting the values of non-sensitive attributes $X$ as
proxies to infer the value of the sensitive variable $S$. Starting
from a naïve approach, dubbed \emph{threshold estimator}
(\textbf{TE}), which estimates $\mu(s)=\Pr(\hat{Y}=\oplus|S=s)$ as
\begin{align}
  \hat{\mu}^{\mathrm{TE} (s)
  } = \frac{\sum_{\mathbf{x}_i} k_s(\mathbf{x}_i) 
  h_{\oplus}(\mathbf{x}_i)}{\sum_{\mathbf{x}_i} k_s(\mathbf{x}_i)} 
  \label{eq:mu_cc}
\end{align}
\noindent i.e., by using a hard classifier
$k_s:\mathcal{X}\rightarrow \{0,1\}$ (which outputs Boolean decisions
regarding membership in a sensitive group $S=s$), they propose a
\emph{weighted estimator} (\textbf{WE}) with better convergence
properties.
% They characterise the limitations of a naïve approach which
% estimates $\mu(s)$ in Equation~\ref{eq:mu} via a hard classifier
% $k_s:\mathcal{X}\rightarrow \{0,1\}$ that outputs Boolean decisions
% regarding membership to a sensitive group $S=s$ as follows:
%%
% \begin{align}
% \hat{\mu}(s) = \frac{\sum_{\mathbf{x}_i} k_s(\mathbf{x}_i) h_{\oplus}(\mathbf{x}_i)}{\sum_{\mathbf{x}_i} k_s(\mathbf{x}_i)} \label{eq:mu_cc}
%\end{align}
%%
%\noindent They propose a \emph{weighted estimator} (\textbf{WE})
%
\begin{align}
  \hat{\mu}^{\mathrm{WE}}(s) = \frac{\sum_{\mathbf{x}_i} 
  \pi_s(\mathbf{x}_i) h_{\oplus}(\mathbf{x}_i)}{\sum_{\mathbf{x}_i} 
  \pi_s(\mathbf{x}_i)} \label{eq:mu_we}
\end{align}
\noindent WE exploits a soft classifier
$\pi_s:\mathcal{X}\rightarrow [0,1]$ that outputs posterior
probabilities $\Pr(s|\mathbf{x}_i)$. The posteriors represent the
probability that the classifier attributes to the fact that
$\mathbf{x}_i$ belongs to the subpopulation with sensitive attribute
$S=s$. The authors argue that the naïve estimator of
Equation~\eqref{eq:mu_cc} has a tendency to exaggerate disparities,
and show that WE mitigates this problem under the hypothesis that
$\pi_s(\mathbf{x}_i)$ outputs well-calibrated posterior
probabilities. \wasblue{A contribution of our paper is to show that TE
and WE are just instances of a broad family of estimators (Proposition
\ref{prop:we_pcc}). Moreover, we consider alternative methods from the
same family, and show them to outperform both TE and WE on an
extensive suite of experiments (Section \ref{sec:experiments}).}

% In this work we show extensive empirical evidence of this phenomenon
% and adapt different approaches from the quantification

\citet{kallus2020:aa} study the problem of measuring a classifier's
demographic disparity, true positive rate disparity, and true negative
rate disparity in a setting with access to a primary dataset involving
$(\hat{Y},Z)$ and an auxiliary dataset involving $(S,Z)$, where $Z$ is
a generic set of proxy variables, potentially disjoint from $X$. They
show that reliably estimating the demographic disparity of a
classifier issuing predictions $\hat{Y}$ when $Z$ is not highly
informative with respect to $\hat{Y}$ or $S$ is infeasible. Moreover,
they provide upper and lower bounds for the true value of the estimand
in a setting where the primary and auxiliary datasets are drawn from
marginalisations of a common joint distribution. \wasblue{Our work
departs from this setting in two important ways, to focus on realistic
conditions for internal fairness audits. Firstly, we take into account
the nonstationarity of the processes generating the data and do not
assume the primary and auxiliary dataset to be marginalisations of the
same joint distribution. Rather, we identify different sources of
distribution shift, and formalize them into protocols to test the
performance of different estimators in a more realistic setting
(Sections \ref{sec:sample_prev_d3}--\ref{sec:flip_prev_d1}). Secondly,
we hypothesize that, from within the company deploying a classifier
$h(\mathbf{x})$, the available proxy variables $Z$ comprise $X$, and
are thus highly informative with respect to $\hat{Y}$.}

\citet{aswasthi2021:ef} characterize the structure of the best
estimator for sensitive attributes when the final estimand is a
classifier's disparity in true positive rates across protected
groups. They show that the test accuracy of the attribute classifier
and its performance as an estimator of the true positive rate
disparity are not necessarily correlated. \wasblue{We contribute to this
line of research, demonstrating the possibility to decouple the
\emph{classification} performance of a model when deployed for
sensitive attribute inference at the individual level, which
constitutes a privacy infringement, from its \emph{quantification}
performance in applications where it is used for group-level estimates
(Section \ref{sec:q_not_c}). This line of work opens the possibility
of developing estimators that reliably measure group fairness under
unawareness of sensitive attributes, while guaranteeing privacy at the
individual level.}

\subsection{Quantification and Fairness}
\label{sec:quantificationlearning}

\noindent The application of quantification methods in algorithmic
fairness research is not entirely new. \citet{biswas2021ensuring}
study the problem of enforcing fair classification under distribution
shift, which potentially affects different demographic groups at
different rates. They define a notion of fairness based on the
proportionality between the prevalence of positives in a protected
group $S=s$ and the group-specific acceptance rate of a classifier
issuing predictions $\hat{Y}$. This notion, called \emph{proportional
equality}, is defined by the quantity
$$\text{PE}= \left | \frac{\Pr(Y=\oplus|S=1)}{\Pr(Y=\oplus|S=0)} -
  \frac{\Pr(\hat{Y}=\oplus|S=1)}{\Pr(\hat{Y}=\oplus|S=0)} \right |$$
calculated on a test set $\mathcal{D}$, where low values of
$\text{PE}$ correspond to fairer predictions $\hat{Y}$. In the
presence of distribution shift between training and testing
conditions, the true group-specific prevalences
% , i.e., the numerator and denominator of the first fraction,
$\Pr(Y=\oplus|S=1)$ and $\Pr(Y=\oplus|S=0)$ are unknown. The authors
use an approach from the quantification literature to estimate these
prevalence values, integrating it in a wider system aimed at
optimizing PE.

% PCC-based quantifier to estimate these prevalence values. PCC is
% embedded within a wider pipeline called CAPE, combining sampling,
% ensemble and quantification techniques to provide fair predictions.

In other words, prior work applying quantification to problems of
algorithmic fairness concentrates on \emph{enforcing} classifier
fairness under unawareness of \emph{target labels}.
% \fabsebcomment{Che \citet{biswas2021ensuring} si occupi di
% ``enforcing'' non mi sembra per nulla evidente dalla descrizione di
% cui sopra, che mi sembra invece indicare il contrario.}
Our work, on the other hand, aims at \emph{measuring} classifier
fairness under unawareness of \emph{sensitive attributes}.
        
% -------------------------------------------------------------------

\section{Measuring Fairness Under Unawareness: A Quantification-based
Method}
\label{sec:method}

\noindent In this section, we first present a primer on quantification
(Section~\ref{sec:learningtoquantify}), and then show how
to measure fairness under unawareness with quantification
(Section~\ref{sec:adaptation}), discussing the properties of the
resulting estimators.

\subsection{Learning to Quantify}
\label{sec:learningtoquantify}

\noindent \emph{Quantification} (also known as \emph{supervised
prevalence estimation}, or \emph{learning to quantify}) is the task of
training, by means of supervised learning, a predictor that estimates
the relative frequency (also known as \emph{prevalence}, or
\emph{prior probability}) of the classes of interest in a sample of
unlabelled data points, where the data used to train the predictor are
a set of labelled data points; see \citet{Gonzalez:2017it} for a
survey of quantification research.

\wasblue{
\begin{definition}
  \label{def:ql}
  Given a sample $\sigma$ of data points $\mathbf{x} \in \mathcal{X}$,
  with unknown target labels in domain $\mathcal{S}$, a
  \emph{quantifier} $q(\sigma)$ is an estimator
  $q:2^\mathcal{X} \rightarrow [0,1]$ that predicts the prevalence of
  class $s$ in the sample $\sigma$ as
  $\hat{p}_{\sigma}^{q}(s) = q(\sigma)$.
\end{definition}

\begin{remark}
  The above definition is deliberately broad to include the trivial
  \emph{classify and count} baseline introduced below. In practice, a
  method is \emph{truly} quantification-based when explicitly
  targeting prevalence estimates, rather than simply treating them as
  a by-product of classification. This includes methods that make use
  of dedicated loss functions, task-specific adjustments, and ad hoc
  model selection procedures. Typically, the prevalence estimates
  issued by these methods display desirable properties of unbiasedness
  and convergence.
\end{remark}
}

\noindent Quantification can be trivially solved via classification,
i.e., by classifying all the unlabelled data points by means of a
standard classifier, counting, for each class, the data points that
have been assigned to the class, and normalizing. However, it has
unequivocally been shown (see, among many others,
\citet{Forman:2008kx,Gonzalez:2017it,Gonzalez-Castro:2013fk,Moreo:2022bf,Vaz:2019eu})
that solving quantification by means of this \emph{classify and count}
(CC) method is suboptimal, and that more accurate quantification
methods exist.
The key reason behind this is the fact that many applicative scenarios
suffer from \emph{distribution shift}, therefore the class prevalence
values in the training set may substantially differ from the class
prevalence values in the unlabelled data that the classifier issues
predictions for \citep{morenotorres2012:unifying}. The presence of
distribution shift means that the well-known IID assumption, on which
most learning algorithms for training classifiers are based, does not
hold; in turn, this means that CC will perform suboptimally on
scenarios that exhibit distribution shift, and that the higher the
amount of shift, the worse we can expect CC to perform.

A wide variety of quantification methods have been defined in the
literature. In the experiments presented in this paper, we compare six
such methods, which we briefly present in this section. One of them is
the trivial CC baseline; we have chosen the other five methods over
other contenders because they are simple and proven, and because some
of them (especially the ACC, PACC, SLD and HDy methods; see below)
have shown top-notch performance in recent comparative tests run in
other domains~\citep{Moreo:2021sp,Moreo:2022bf}. We briefly describe
them here, with direct reference to the application we are interested
in, i.e., estimating the prevalence of a protected subgroup.

As mentioned above, an obvious way to solve quantification (used,
among others, in Equation~\ref{eq:mu_cc}) is by aggregating the
predictions of a ``hard'' classifier, i.e., a classifier
$k_s:\mathcal{X}\rightarrow \{0,1\}$ that outputs Boolean decisions
regarding membership in a sensitive group (defined by constraint
$S=s$).
% encoding $\mathcal{Y}=\{\ominus,\oplus\}$, with $\ominus\equiv 0$
% and $\oplus\equiv 1$) \footnote{While in this work quantifiers are
% used to estimate the prevalence of sensitive attribute $S$, in this
% section, we present quantification approaches targeting a generic
% variable $\mathcal{Y}$, as typical in the literature. Similarly, the
% hard classifier $h$ in this section should not be confused with the
% one introduced in Section~\ref{sec:notation}. The notational overlap
% is limited to the current section, where we follow the typical
% machine learning conventions in presenting the supervised task of
% quantification.} Of course,
% $\Pr(S=1|\mathbf{x})=(1-\Pr(S=0|\mathbf{x}))$.
The (trivial) \emph{classify and count} (\textbf{CC}) quantifier then
comes down to computing
\begin{align}
  \label{eq:cc}
  \hat{p}_{\sigma}^{\mathrm{CC}}(s) & = \frac{\sum_{\mathbf{x}_i\in \sigma}
                                      k_s(\mathbf{x}_i)}{|\sigma|}.
\end{align}
\noindent Alternatively, quantification methods can use a ``soft''
classifier $\pi_s:\mathcal{X}\rightarrow [0,1]$ that produces
posterior probabilities $\Pr(s|\mathbf{x}_i)$.
% , representing the probability that the classifier attributes to the
% fact that $\mathbf{x}_i$ belongs to the subpopulation with sensitive
% attribute $S=s$.
The resulting \emph{probabilistic classify and count} quantifier
(\textbf{PCC})~\citep{Bella:2010kx} is defined by the equation
\begin{align}
  \label{eq:pcc}
  \hat{p}_{\sigma}^{\mathrm{PCC}}(s) & = \frac{\sum_{\mathbf{x}_i\in \sigma}
                                       \pi_{s}(\mathbf{x}_i)}{|\sigma|}.
\end{align}
\noindent It should be noted that PCC and CC are clearly related to WE
and TE, summarized by Equations \eqref{eq:mu_cc} and \eqref{eq:mu_we},
as shown later in Proposition \ref{prop:we_pcc}.

% \noindent Of course, for any quantification method $q$ we have
% $\hat{p}_{\sigma}^{\mathit{q}}(\ominus)=(1-\hat{p}_{\sigma}^{\mathit{q}}(\oplus))$.

A different and popular quantification method consists of applying an
\emph{adjustment} to the prevalence
$\hat{p}^{\mathrm{CC}}_{\sigma}(s)$ estimated through ``classify and
count''. It is easy to check that, in the binary case, the true
prevalence $p_{\sigma}(s)$ and the estimated prevalence
$\hat{p}^{\mathrm{CC}}_{\sigma}(s)$ are such that
\begin{equation}
  \label{eq:exactacc} 
  p_{\sigma}(s) = \frac{\hat{p}_{\sigma}^{\mathrm{CC}}(s) - 
  \mathrm{fpr}_{k_{s}}}{\mathrm{tpr}_{k_{s}} - \mathrm{fpr}_{k_{s}}}
\end{equation}
\noindent where $\mathrm{tpr}_{k_{s}}$ and $\mathrm{fpr}_{k_{s}}$
stand for \emph{true positive rate} and \emph{false positive rate} of
the classifier $k_{s}$ used to obtain
$\hat{p}_{\sigma}^{\mathrm{CC}}(s)$. The values of
$\mathrm{tpr}_{k_{s}}$ and $\mathrm{fpr}_{k_{s}}$ are unknown, but can
be estimated via $k$-fold cross-validation on the training data. This
boils down to using the results $k_{s}(\mathbf{x}_{i})$ obtained in
the $k$-fold cross-validation (i.e., $\mathbf{x}_{i}$ ranges on the
training items) in Equations
\begin{align}
  \begin{split}
    \label{eq:tprandfpr}
    \hat{\mathrm{tpr}}_{k_{s}} = \frac{\sum_{\{(\mathbf{x}_i,
    s_i)|s_i=s\}}k_{s}(\mathbf{x}_i)}{|\{(\mathbf{x}_i, s_i)|s_i=s\}|}
    \hspace{3em} \hat{\mathrm{fpr}}_{k_{s}}
    =\frac{\sum_{\{(\mathbf{x}_i, s_i)|s_i\neq
    s\}}k_s(\mathbf{x}_i)}{|\{(\mathbf{x}_i, s_i)|s_i\neq s\}|}.
  \end{split}
\end{align}
\noindent We obtain estimates of $p_{\sigma}^{\mathrm{ACC}}(s)$, which
define the \emph{adjusted classify and count} method
(\textbf{ACC})~\citep{Forman:2008kx}, by replacing
$\mathrm{tpr}_{k_{s}}$ and $\mathrm{fpr}_{k_{s}}$ in Equation
\ref{eq:exactacc} with the estimates of Equation~\ref{eq:tprandfpr},
i.e.,
\begin{equation}
  \label{eq:acc} 
  \hat{p}_{\sigma}^{\mathrm{ACC}}(s) = 
  \frac{\hat{p}_{\sigma}^{\mathrm{CC}}(s) 
  - \hat{\mathrm{fpr}}_{k_{s}}}{\hat{\mathrm{tpr}}_{k_{s}} 
  - \hat{\mathrm{fpr}}_{k_{s}}}.
\end{equation}
\noindent If the soft classifier $\pi_{s}(\mathbf{x}_i)$ is used in
place of $k_{s}(\mathbf{x}_i)$, analogues of
$\hat{\mathrm{tpr}}_{k_{s}}$ and $\hat{\mathrm{fpr}}_{k_{s}}$ from
Equation~\ref{eq:tprandfpr} can be defined as
\begin{align}
  \begin{split}
    \label{eq:tprandfpr2}
    \hat{\mathrm{tpr}}_{\pi} = \frac{\sum_{\{(\mathbf{x}_i,
    s_i)|s_i=s\}}\pi_{s}(\mathbf{x}_i)}{|\{(\mathbf{x}_i,
    s_i)|s_i=s\}|} \hspace{3em} \hat{\mathrm{fpr}_{\pi}}
    =\frac{\sum_{\{(\mathbf{x}_i, s_i)|s_i\neq
    s\}}\pi_s(\mathbf{x}_i)}{|\{(\mathbf{x}_i, s_i)|s_i\neq s\}|}.
  \end{split}
\end{align}
\noindent We obtain $p_{\sigma}^{\mathrm{PACC}}(s)$ estimates, which
define the \emph{probabilistic adjusted classify and count} method
(\textbf{PACC})~\citep{Bella:2010kx}, by replacing all factors on the
right-hand side of Equation~\ref{eq:acc} with their ``soft''
counterparts from Equations \ref{eq:pcc} and \ref{eq:tprandfpr2},
i.e.,
\begin{equation}
  \label{eq:pacc} 
  \hat{p}_{\sigma}^{\mathrm{PACC}}(s) = 
  \frac{\hat{p}_{\sigma}^{\mathrm{PCC}}(s) 
  - \hat{\mathrm{fpr}}_{\pi}}{\hat{\mathrm{tpr}}_{\pi} 
  - \hat{\mathrm{fpr}}_{\pi}}.
\end{equation}
\noindent A further method is the one proposed in
\citep{Saerens:2002uq} (which we here call \textbf{SLD}, from the names
of its proposers), which consists of training a probabilistic
classifier and then using the Expectation–Maximization (EM) algorithm
(i) to update (in an iterative, mutually recursive way) the posterior
probabilities that the classifier returns, and (ii) to re-estimate the
class prevalence values of the test set until convergence. This makes
the method robust to distribution shift, since the iterative process
allows the estimates of the prevalence values to become increasingly
attuned to the changed conditions found in the unlabelled set.
% \fabsebcomment{Qui potrebbe valer la pena accennare come mai SLD è
% robusto al distribution shift; di fatto, il lettore di questo papero
% potrebbe chiedersi come facciano questi algoritmi di quantificazione
% a essere robusti al distribution shift.}
Pseudocode describing the SLD algorithm can be found in Appendix
\ref{app:sld}.

%Lastly, we 
We consider \textbf{HDy} \citep{Gonzalez-Castro:2013fk}, a
probabilistic binary quantification method that views quantification
as the problem of minimizing the divergence (measured in terms of the
Hellinger Distance) between two cumulative distributions of posterior
probabilities returned by the classifier, one coming from the
unlabelled examples and the other coming from a validation set. HDy
looks for the mixture parameter $\alpha$ that best fits the validation
distribution (consisting of a mixture of a ``positive'' and a
``negative'' distribution) to the unlabelled distribution, and returns
$\alpha$ as the estimated prevalence of the positive class. Here,
robustness to distribution shift is achieved by the analysis of the
distribution of the posterior probabilities in the unlabelled set,
which reveals how conditions have changed with respect to the training
data. A more detailed description of HDy can be found in Appendix
\ref{app:hdy}.

\blue{Lastly, we consider Maximum Likelihood Prevalence Estimator (\textbf{MLPE}),
a dummy method that assumes there is no shift and always returns the 
class prevalence value as found in the training data, as the estimate of any future test sample.
This method is not a serious contender, since MLPE makes no real attempt to
address the problem. Notwithstanding this, MLPE is going to generate very low error values
in all protocols in which the test prevalence is kept fixed.}

% -------------------------------------------------------------------

\subsection{Using Quantification to Measure Fairness Under
Unawareness}
\label{sec:adaptation}

\noindent We assume the existence, in
the operational setup, of three separate sets of data points:
\begin{itemize}
\item A \emph{training set} $\mathcal{D}_1$ for $h$,
  $\mathcal{D}_1 = \{(\mathbf{x}_i,y_i) \;|\;
  \mathbf{x}_i\in\mathcal{X}, y_i\in\mathcal{Y}\}$, typically of large
  size, where $h$ is the classifier whose fairness we want to
  measure. Given the difficulties inherent in demographic data
  procurement mentioned in the introduction, we assume that the
  sensitive attribute $S$ is not part of the vectorial representation
  $X$.
\item A small \emph{auxiliary set}
  $\mathcal{D}_2 = \{(\mathbf{x}_i,s_i) \;|\;
  \mathbf{x}_i\in\mathcal{X}, s_i\in\mathcal{S}\}$, containing
  demographic data, employed to train quantifiers for the sensitive
  attribute.
  % \fabsebcomment{Questa frase qui sotto è ripetuta identica verso la
  % fine della Sezione 5; secondo me sta meglio là.} This dataset may
  % originate from a targeted effort, such as interviews
  % \citep{baker2005patients}, surveys sent to customers asking for
  % voluntary disclosure of sensitive attributes
  % \citep{andrus2021:wm}, or other optional means to share
  % demographic information
  % \citep{beutel2019putting,beutel2019fairness}. %Alternatively it could derive from data acquisitions carried out for other purposes \citep{gladonclavell2020auditing}.
  % \fabsebcomment{Questa frase qui sotto non si capisce cosa voglia
  % dire.}
  % Both $\mathcal{D}_1$ and $\mathcal{D}_2$ are in the development
  % domain of the operational setup, i.e.\ they are available to
  % practitioners prior to model deployment.
\item A set
  $\mathcal{D}_3 = \{\mathbf{x}_i \;|\; \mathbf{x}_i\in\mathcal{X}\}$
  of unlabelled data points, which are the data to which classifier
  $h$ is to be applied, representing the deployment
  conditions. Alternatively, $\mathcal{D}_3$ could also be a labelled
  held-out test set available at a company, if it has acted
  proactively rather than reactively, for pre-deployment audits
  \citep{raji2020closing}. In our experiments we will use labelled data
  and call $\mathcal{D}_3$ the \emph{test set}, on which
  the fairness of the classifier $h$ should be measured.
  % \fabsebcomment{Alessandro, adesso ho capito perché insistevi per
  % chiamarlo ``deployment set''. Perché per te la pipeline che
  % descriviamo qui sopra era la pipeline ``in production'', mentre
  % per me era la pipeline dei nostri esperimenti. Ho provato a
  % spiegare, come vedi sopra.}
  % \fabsebcomment{Continuo a pensare che sia assurdo chiamare questo
  % insieme il ``deployment'' set; tutto il mondo del machine learning
  % lo chiamerebbe il \emph{test} set, dato che un test set è, ovunque
  % in applicazioni ML, l'insieme che, in una valutazione sperimentale
  % come la nostra, fa le veci del test set, e cioè dei dati
  % unlabelled a cui il sistema oramai in produzione viene
  % applicato. Tra l'altro, ``deploy`` in inglese vuol dire un'altra
  % cosa (``schierare'', ``mettere in campo''; si può dire ``to deploy
  % a classifier`` ma non ``to deploy a classifier on a dataset''.}
\end{itemize}

\noindent It is worth re-emphasizing that, from the perspective of the
estimation task at hand, i.e., estimating the fairness of the classifier $h$, $\mathcal{D}_2$
represents the quantifier's training set, while $\mathcal{D}_3$ is its
test set.
\wasblue{
\begin{prop} \label{prop:ql4facct} Observational measures of
  algorithmic fairness, such as the ones introduced in Definition
  \ref{def:dd}, can be computed, under unawareness of sensitive
  attributes, by estimating the prevalence of the sensitive attribute
  in specific subsets of the test set.
\end{prop}
\begin{proof}
  We prove this statement for TPRD in Definition \ref{def:dd}, which
  we recall below: \small{
  \begin{align*}
    \text{True Positive Rate Disparity: } \delta_{h}^{S\mathrm{, TPRD}} = \Pr(\hat{Y}=\oplus|S=1, Y=\oplus) - \Pr(\hat{Y}=\oplus|S=0, Y=\oplus) 
  \end{align*}
  }
  \noindent Both terms in the above equation can be written as
  \begin{align*}
    \Pr(\hat{Y}=\oplus|S=s, Y=\oplus) &= \frac{\Pr(Y=\oplus, \hat{Y}=\oplus,S=s)}{\Pr(Y=\oplus, S=s)} \\
                                      &= \underbrace{ \frac{\Pr(S=s| Y=\oplus,\hat{Y}=\oplus)}{\Pr(S=s| Y=\oplus)}}_{\text{obtained from prevalence estimator}} \cdot \underbrace{\frac{\Pr( Y=\oplus,\hat{Y}=\oplus)}{\Pr(Y=\oplus)}}_{\text{known quantity}}
  \end{align*}
 
  % \Pr(\hat{Y}=\oplus|S=s)=\overbrace{\Pr(S=s|\hat{Y}=\oplus)}^{\text{EST.}}\frac{\overbrace{\Pr(\hat{Y}=\oplus)}^{\text{KNOWN}}}{\underbrace{\Pr(S=s|\hat{Y}=\oplus)}_{\text{EST.}}
  % \underbrace{\Pr(\hat{Y}=\oplus)}_{\text{KNOWN}} +
  % \underbrace{\Pr(S=s|\hat{Y}=\ominus)}_{\text{EST.}}
  % \underbrace{\Pr(\hat{Y}=\ominus)}_{\text{KNOWN}}}
 
  \noindent In other words, TPRD can be calculated by estimating the
  prevalence of the sensitive attribute among the positives and the
  true positives in $\mathcal{D}_3$. Analogous results can be proven
  for other measures of observational fairness, under the assumption
  that $Y$ and $\hat{Y}$ are known.
\end{proof}

\begin{remark}
  This proposition is important for two reasons. First, it shows that
  inference of sensitive attributes at the individual level is not
  necessary to measure fairness under unawareness; rather, prevalence
  estimates in given subsets are sufficient. Second, it suggests that
  methods directly targeting prevalence estimates (i.e.,
  \emph{quantifiers}) are especially suited in this setting.
\end{remark}

\noindent Notice that, for the purposes of a fairness audit, it is
common to assume that the ground truth variable $Y$ is available in
$\mathcal{D}_3$. In the banking scenario of Example \ref{ex:bank},
this is only partially realistic, as the outcomes for the accepted
applicants are eventually observed, but the outcomes for the rejected
applicants remain unknown, leaving us with a problem of sample
selection bias \citep{banasik2003sample}. This is an instance of a
general estimation problem, common to all fairness criteria that
require knowledge of the ground truth variable $Y$, such as TPRD,
TNRD, PPVD, and NPVD in Definition \ref{def:dd}. This represents an
open research problem \citep{wang2021fair,sabato2020bounding} which is
beyond the scope of this work \blue{and demands additional caution in the estimation and interpretation of these fairness measures.} 

\replace{In the remainder of this article, we avoid this issue by focusing on a detailed study of demographic disparity (DD), which does not require information on the ground truth $Y$, leaving additional measures of observational fairness for future work.}{In the remainder of this article, we focus on a
detailed study of demographic disparity (DD). This allows us to thoroughly characterize and discuss DD estimators while avoiding the pitfalls and complexity of uncertain ground truth information. We leave additional measures of observational fairness for future work.} 

Following \citep{chen2019:fu},
we write DD as
\begin{align}
  \delta_{h}^{S} = \Pr(\hat{Y}=\oplus|S=1) - 
  \Pr(\hat{Y}=\oplus|S=0) = \mu(1) - \mu(0), \label{eq:dd}
\end{align}
\noindent where
\begin{align}
  \mu(s) = \Pr(\hat{Y}=\oplus|S=s) \label{eq:mu}
\end{align} 
\noindent is the acceptance rate of individuals in the group $S=s$.}
To estimate the demographic disparity of a classifier $h(\mathbf{x})$
in the test set $\mathcal{D}_3$, we
can use any quantification approach from
Section~\ref{sec:learningtoquantify}.
% \fabsebcomment{Questo deployment set $\mathcal{D}_3$ è stato
% introdotto 5 pagine fa, e qui il lettore se l'è già bell'e scordato;
% questo è un altro motivo secondo me per cui quella parte andrebbe
% spostata da pagina 5 a qua.}
Applying Bayes' theorem to
Equation~\eqref{eq:mu}, we obtain
\begin{align}
  \mu(s) = & \ p_{\mathcal{D}_{3}}(\oplus|s) \nonumber \\
  = & \ p_{\mathcal{D}_{3}^{\oplus}}(s) 
      \frac{p_{\mathcal{D}_{3}}(\oplus)}{p_{\mathcal{D}_{3}}(s)},
      \label{eq:mu2}
\end{align}
\noindent where we use $p_{\mathcal{D}_{3}}(\oplus)$ as a shorthand of
$p_{\mathcal{D}_{3}}(h(\mathbf{x})=\oplus)$, and where we have defined
\begin{align*}
  \begin{split}\label{eq:split}
    % \mathcal{D}_{3}^{\ominus} = \{(\mathbf{x},s) \in \mathcal{D}_3 |
    % h(\mathbf{x})=\hat{y} \}, \quad \hat{y} \in \left \lbrace
    %   \ominus,\oplus \right \rbrace.
    \mathcal{D}_{3}^{\oplus} = & \{\mathbf{x} \in \mathcal{D}_3 \;|\;
    h(\mathbf{x})=\oplus \} \\
    \mathcal{D}_{3}^{\ominus} = & \{\mathbf{x} \in \mathcal{D}_3 \;|\;
    h(\mathbf{x})=\ominus \}.
  \end{split}
\end{align*}
%
% \noindent and where we make explicit the fact that, if a value $s$
% that attribute $S$ can take is viewed as a class, the probabilities
% $\Pr(s|\ominus)$ and $\Pr(s|\oplus)$ of Equation~\ref{eq:dd1} may be
% seen as the prevalence values of class $s$ in the two samples
% $\mathcal{D}_{3}^{\oplus}$ and $\mathcal{D}_{3}^{\ominus}$. In other
% words, measuring demographic disparity is reduced to estimating the
% prevalence values of class $s$ in the two samples
% $\mathcal{D}_{3}^{\oplus}$ and $\mathcal{D}_{3}^{\ominus}$, i.e.,
% \emph{it can be framed as a task of quantification}.

\noindent Since $p_{\mathcal{D}_{3}}(\oplus)$ is known (it is the
fraction of items in $\mathcal{D}_{3}$ that have been assigned class
$\oplus$ by the classifier $h$), in order to compute $\mu(s)$ through
Equation~\eqref{eq:mu2}, for $s\in\{0,1\}$, we only need to estimate
the prevalence values $\hat{p}_{\mathcal{D}_{3}^{\oplus}}(s)$ and
$\hat{p}_{\mathcal{D}_{3}^{\ominus}}(s)$; the latter is needed to
estimate the denominator of Equation~\eqref{eq:mu2}, i.e., the
prevalence $p_{\mathcal{D}_{3}}(s)$ of the sensitive attribute value
$s$ in the entire test set $\mathcal{D}_3$, since
\begin{equation}
  p_{\mathcal{D}_{3}}(s) = p_{\mathcal{D}_{3}^{\oplus}}(s) 
  \cdot p_{\mathcal{D}_{3}}(\oplus) + p_{\mathcal{D}_{3}^{\ominus}}(s) 
  \cdot p_{\mathcal{D}_{3}}(\ominus). \label{eq:dd_den}
\end{equation}
%
% Since quantification is a supervised task, we need a set of labelled
% data points for training our quantifiers, i.e., a set of data points
% $\mathbf{x}$ which use the same set $\mathcal{A}$ of features as
% $\mathcal{D}_{1}$ and $\mathcal{D}_{3}$ but whose labels are the
% values of $S$ (which are here viewed as classes). We call this
% ``auxiliary'' data set $\mathcal{D}_{2}$.
%
\noindent In order to compute $p_{\mathcal{D}_{3}^{\oplus}}(s)$ and
$p_{\mathcal{D}_{3}^{\ominus}}(s)$ we can use a quantification-based
approach, which can be easily integrated into existing machine
learning workflows, as summarized by
the method below.

\medskip

\noindent \wasblue{\textbf{Method}. Quantification-Based Estimate of
Demographic Disparity.}

\begin{enumerate}

\item The classifier $h:\mathcal{X}\rightarrow\mathcal{Y}$ is trained
  on $\mathcal{D}_1$ and ready for deployment, e.g., to estimate the
  creditworthiness of individuals. The assumption that, at this
  training stage, we are unaware of the sensitive attribute $S$ is due
  to the inherent difficulties in demographic data procurement already
  mentioned in Section~\ref{sec:intro}.
  % A \emph{training set} \afabcomment{development set may more
  % appropriate, to stress that this is all the data used to choose an
  % appropriate classifier for production, including training,
  % validation and test.} $\mathcal{D}_1$ for $h$ involving $(X,Y)$,
  % typically of large size. Given the inherent difficulties in
  % demographic data procurement, we expect this dataset to contain no
  % explicit information on the sensitive attribute $S$.
 
\item We use the classifier $h$ to classify the auxiliary set
  $\mathcal{D}_2$, thus inducing a partition of $\mathcal{D}_2$ into
  $\mathcal{D}_{2}^{\oplus}=\{(\mathbf{x}_i, s_i) \in \mathcal{D}_2
  \;|\; h(\mathbf{x})=\oplus \}$ and
  $\mathcal{D}_{2}^{\ominus}=\{(\mathbf{x}_i, s_i) \in \mathcal{D}_2
  \;|\; h(\mathbf{x})=\ominus \}$.
 
\item \label{item:trainquantifiers} We use $\mathcal{D}_{2}^{\oplus}$
  as the training set for the quantifier $q_{\oplus}(s)$, whose task
  will be to estimate the prevalence of value $s$ (e.g.,
  African-American applicants) on sets of data points labelled with
  class $\oplus$ (e.g., creditworthy applicants). Likewise, we use
  $\mathcal{D}_{2}^{\ominus}$ as the training set for a quantifier
  $q_{\ominus}(s)$ whose task will be to estimate the prevalence of
  $s$ on sets of data points labelled with
  $\ominus$. \wasblue{Intuitively, separate quantifiers specialized on
  different subpopulations (of positively and negatively classified
  individuals) should perform better than a single quantifier. The
  ablation study in Section \ref{sec:ablation} supports this
  hypothesis.}
  % These two disjoint datasets act as the training sets for the two
  % quantifiers $q_{\ominus}$ and $q_{\oplus}$,
  % respectively. Quantifier $q_{\ominus}$ (or its dual $q_{\oplus}$)
  % is trained on $\mathcal{D}_{2}^{\ominus}$ (resp.,
  % $\mathcal{D}_{2}^{\oplus}$) to estimate the prevalence of data
  % points where $S=s$ among the data points labelled with $\ominus$
  % (resp., $\oplus$).
 
\item The classifier $h$ is deployed, classifying the test set
  $\mathcal{D}_3$, thus inducing a partition of $\mathcal{D}_3$ into
  positive
  $\mathcal{D}_{3}^{\oplus}=\{\mathbf{x} \in \mathcal{D}_3 \;|\;
  h(\mathbf{x})=\oplus \}$ and negative
  $\mathcal{D}_{3}^{\ominus}=\{\mathbf{x} \in \mathcal{D}_3 \;|\;
  h(\mathbf{x})=\ominus \}$.

\item We apply the quantifier $q_{\oplus}$ to
  $\mathcal{D}_{3}^{\oplus}$ to obtain an estimate
  $\hat{p}_{\mathcal{D}_{3}^{\oplus}}^{q_{\oplus}}(s)$ of the
  prevalence of $s$ in $\mathcal{D}_{3}^{\oplus}$, and we apply
  $q_{\ominus}$ to $\mathcal{D}_{3}^{\ominus}$ to obtain an estimate
  $\hat{p}_{\mathcal{D}_{3}^{\ominus}}^{q_{\ominus}}(s)$ of the
  prevalence of $s$ in $\mathcal{D}_{3}^{\ominus}$. Recall from
  Section~\ref{sec:notation} that $\hat{p}_{\sigma}^q(s)$ denotes the
  prevalence of an attribute value $s$ in a set $\sigma$ as estimated
  via quantification method $q$.
 
\item To avoid numerical instability in the denominator of
  Equation~\eqref{eq:mu_ql} below, we apply Laplace smoothing to the
  estimated prevalence values
  $\hat{p}^{q_{\oplus}}_{\mathcal{D}_3^{\oplus}}(s)$ and
  $\hat{p}^{q_{\ominus}}_{\mathcal{D}_3^\ominus}(s)$. We use the
  variant that uses known incidence rates, using
  $\mathcal{D}_2^\ominus$ and $\mathcal{D}_2^\oplus$ as the control
  populations, and assume a pseudocount $\alpha=1/2$. We thus compute
  the smoothed estimator
  \begin{align*}
    \tilde{p}_{\mathcal{D}_3^\oplus}^{q_{\oplus}}(s) = & \ \frac{\hat{p}_{\mathcal{D}_3^\oplus}^{q_{\oplus}}(s)
                                                         \cdot
                                                         |\mathcal{D}_3^\oplus|+p_{\mathcal{D}_2^\oplus}(s)\cdot \alpha \cdot |\mathcal{Y}|}{|\mathcal{D}_3^\oplus|+\alpha \cdot |\mathcal{Y}|} \\ 
    = & \ \frac{\hat{p}_{\mathcal{D}_3^\oplus}^{q_{\oplus}}(s)
        \cdot
        |\mathcal{D}_3^\oplus|+p_{\mathcal{D}_2^\oplus}(s)}{|\mathcal{D}_3^\oplus|+1}
  \end{align*}
  and analogously for
  $\tilde{p}_{\mathcal{D}_3^\ominus}^{q_{\ominus}}(s)$.
 
\item Finally, we estimate the demographic disparity of $h$, defined
  in Equation~\eqref{eq:dd}, as
  \begin{align}
    \hat{\delta}_{h}^{S} = \hat{\mu}(1) - \hat{\mu}(0)
    \label{eq:dd_estim}
  \end{align}
  \noindent where, as from Equations~\eqref{eq:mu2}
  and~\eqref{eq:dd_den},
  \begin{align}
    \hat{\mu}(s) &= \tilde{p}_{\mathcal{D}_{3}^{\oplus}}^{q_{\oplus}}(s) \cdot \frac{p_{\mathcal{D}_{3}}(\oplus)}{\tilde{p}_{\mathcal{D}_{3}^{\oplus}}^{q_{\oplus}}(s) \cdot p_{\mathcal{D}_{3}}(\oplus) + \tilde{p}_{\mathcal{D}_{3}^{\ominus}}^{q_{\ominus}}(s) \cdot p_{\mathcal{D}_{3}}(\ominus)} \label{eq:mu_ql}
  \end{align}
\end{enumerate}

\wasblue{
\begin{remark}
  \label{rem:ql_fisher}
  Therefore, prevalence estimates
  $\hat{p}_{\mathcal{D}_{3}^{\oplus}}^{q_{\oplus}}(s)$ and
  $\hat{p}_{\mathcal{D}_{3}^{\ominus}}^{q_{\ominus}}(s)$, obtained
  with a quantification method of the type introduced in Section
  \ref{sec:learningtoquantify}, can be translated into estimates of a
  classifier's demographic disparity using Equations
  \eqref{eq:dd_estim} and \eqref{eq:mu_ql}. Importantly, the bias and
  variance of said estimate depend on the properties of the underlying
  quantification method, which have been characterized in the
  quantification literature. For example, SLD, ACC, and PACC have been
  shown to be \emph{Fisher-consistent}, that is, unbiased, under prior
  probability shift \citep{tasche2017:fisher,Vaz:2019eu}. In other
  words, we expect Equation \ref{eq:dd_estim} instantiated with SLD,
  PCC, and PACC to provide unbiased estimates when $\mathcal{D}_2$ and
  $\mathcal{D}_3$ are linked by prior probability shift. We verify
  this property in Sections \ref{sec:sample_prev_d3} and
  \ref{sec:sample_prev_d2}.

\end{remark}

}

\noindent It is worth noting that the weighted estimator (WE)
introduced in \citep{chen2019:fu}, summarized by
Equation~\eqref{eq:mu_we}, can be viewed as a special case of this
approach, as shown by the proposition below.
\begin{prop} \label{prop:we_pcc} The weighted estimator of
  Equation~\eqref{eq:mu_we} is a special case of quantification-based
  estimation of demographic disparity, instantiated with the PCC
  quantification method. \wasblue{Moreover, the threshold estimator of
  Equation~\eqref{eq:mu_cc} corresponds to CC.}
\end{prop}
\begin{proof}
  See Appendix \ref{sec:proof}.
\end{proof}
\wasblue{
\begin{remark}
  The above proposition shows that PCC and WE are equivalent, and that
  the trivial CC quantifier is equivalent to TE. We treat these
  methods as prior art and refer to them as CC and PCC for consistency
  of exposition.
\end{remark}
}
\noindent This quantification-based method of addressing demographic
disparity is suitable for internal fairness audits, since it allows unawareness of
the sensitive attribute $S$ (i) in the set $\mathcal{D}_{1}$ used for
training the classifier $h$ to be audited, and (ii) in the set
$\mathcal{D}_{3}$ on which this classifier is going to be deployed; it
only requires the availability of an auxiliary data set
$\mathcal{D}_{2}$ where the attribute $S$ is labelled.
% \fabsebcomment{Questa frase qui sotto è ripetuta identica verso la
% fine della Sezione 2; secondo me sta meglio qui.}
Dataset $\mathcal{D}_{2}$ may originate from a targeted effort, such
as interviews \citep{baker2005patients}, surveys sent to customers
asking for voluntary disclosure of sensitive attributes
\citep{andrus2021:wm}, or other optional means of sharing demographic
information
\citep{beutel2019putting,beutel2019fairness}. Alternatively, it could
derive from data acquisitions carried out for other purposes
\citep{gladonclavell2020auditing}.
% Both $\mathcal{D}_1$ and $\mathcal{D}_2$ are in the development
% domain of our machine learning pipeline.

%\wasreplace{Additionally, we note that this approach is extremely
%suitable to situations in which the prevalence of attribute value $s$
%in $\mathcal{D}_{2}$ is possibly very different from the prevalence of
%$s$ in the test set $\mathcal{D}_{3}$ (a situation that certainly
%characterizes many operational environments), since many
%quantification approaches are robust by construction to distribution
%shift, \wasreplace{as we will show in the next section.}{as mentioned in
%Remark \ref{rem:ql_fisher}.}}{}
% i.e., to changes in the distribution of the class labels when moving
% from the training to the test set.
%
% Training and maintaining two separate quantifiers $q_{\ominus}$,
% $q_{\oplus}$ may not be necessary. Section~\ref{sec:ablation}
% contains an ablation study where we test the performance of methods
% based on a single quantifier.

% \fabsebcomment{What about a single non-binary (e.g., categorical)
% attribute? We could say that in this case we can assess, in a
% one-vs-all fashion, the fairness of the classifier for each value
% that the attribute can take.}
Finally, note that, in this paper, we assume the existence of a single
binary sensitive attribute $S$ only for ease of exposition; our approach can straightforwardly used in more complex scenarios.

\blue{
\begin{remark} \label{rem:nonbinary}
Our method can deal with multiple, non-binary sensitive attributes. 
\end{remark}
}

If \emph{multiple} sensitive attributes are present at the same time, one can simply measure fairness with respect
to each sensitive attribute separately, if interested in independent
assessments, or jointly, if emphasizing intersectionality
\citep{ghosh2021characterizing}. Our approach can also be extended to
deal with \emph{categorical, non-binary} attributes. In this case, one
needs (1) to extend the notion of demographic disparity to the case of
non-binary attributes. This can be done, e.g., by considering, instead
of the simple difference between two acceptance rates $\mu(s)$ as in
Equation~\eqref{eq:dd}, the variance of the acceptance rates across
the possible values of $S$, or the difference between the highest and
lowest acceptance rate
$\max_{s \in \mathcal{S}} \mu(s) - \min_{s \in \mathcal{S}} \mu(s)$;
and (2) to use a single-label multiclass (rather than a binary)
quantification system. Concerning this, note that all the methods
discussed in Section~\ref{sec:learningtoquantify} except HDy admit
straightforward extensions from the binary case to the single-label
multiclass case (see~\citep{Moreo:2022bf} for details). HDy is a method
for binary quantification only, but it can be adapted to the
single-label multiclass scenario by training a binary quantifier for
each class in one-vs-all fashion, estimating the prevalence of each
class independently of the others, and normalising the obtained
prevalence values so that they sum to 1.
% In this paper we are not concerned with the single-class
% multi-label case since in our experiments we consider sensitive
% attributes $s$ which range on two values only; however, the above
% considerations would be useful for applying our approach to
% non-binary attributes.

\section{Experiments}
\label{sec:experiments}

\subsection{General Setup}
\label{sec:setup}

\begin{table}
  \caption{Summary of experimental protocols.}
  \label{tab:prot} % Give a unique label
  \centering
  \begin{tabular}{lp{4.5cm}p{5.2cm}l}
    \hline\noalign{\smallskip}
    Protocol name & Variable & \blue{Motivation} & Section \\
    \noalign{\smallskip}\hline\noalign{\smallskip}
    \texttt{sample-prev-$\mathcal{D}_3$} & joint distribution of $(S,\hat{Y})$ in $\mathcal{D}_3$, via sampling & \blue{post-deployment drift, ripple effect, domain adaptation} & \cref{sec:sample_prev_d3} \\
    \texttt{sample-prev-$\mathcal{D}_2$} & joint distribution of $(S,\hat{Y})$ in $\mathcal{D}_2$, via sampling & \blue{skewed auxiliary data, non-response bias} & \cref{sec:sample_prev_d2} \\
    \texttt{sample-size-$\mathcal{D}_2$} & size of $\mathcal{D}_2$, via sampling & \blue{variable response rates, issues with sensitive data procurement}  & \cref{sec:sample_size_d2} \\
    \texttt{sample-prev-$\mathcal{D}_1$} & joint distribution of $(S,Y)$ in $\mathcal{D}_1$, via sampling & \blue{censored data, sampling bias} & \cref{sec:sample_prev_d1} \\
    \texttt{flip-prev-$\mathcal{D}_1$} & joint distribution of $(S,Y)$ in $\mathcal{D}_1$, via label flipping & \blue{ground truth distortion, group-dependent annotation inaccuracy} & \cref{sec:flip_prev_d1} \\
    \noalign{\smallskip}\hline
  \end{tabular}
\end{table}

\noindent 
In this section, we carry out an evaluation of different estimators of
demographic disparity. We propose five experimental protocols
(Sections \ref{sec:sample_prev_d3}--\ref{sec:flip_prev_d1}) summarized
in
Table~\ref{tab:prot}. %\fabsebcomment{Forse non sarebbe stato male fare un sesto blocco di esperimenti, i.e., uno dove prendiamo i dataset come sono, senza modificarli mediante sampling o label flipping; farebbe vedere come, sperabilmente, i metodi di quantificazione vanno meglio di ``classify and count'' anche quando esistono quantità di shift ``naturali'', i.e., moderate, i.e., derivanti dalle moderate fluttuazioni che ci sono di già all'interno di un dataset.}
Each protocol addresses a major challenge that may arise in estimating
fairness under unawareness, and does so by varying the size and the
mutual distribution shift of the training, auxiliary, and test
sets. Protocol names are in the form
\texttt{action-characteristic-dataset}, as they act on datasets
($\mathcal{D}_1$, $\mathcal{D}_2$ or $\mathcal{D}_3$), modifying their
characteristics (size or class prevalence) through one of two actions
(sampling or flipping of labels). We investigate the performance of
six estimators of demographic disparity in each of the five
challenges/protocols, keeping the remaining factors
constant. For every protocol, we perform an extensive
empirical evaluation as follows:

% \fabsebcomment{Questa parte va un po' riordinata.} We perform
% extensive empirical evaluation of different estimation
% techniques. Under each experimental protocol, the size or the
% prevalence of a given dataset is carefully varied based on the
% protocol definition (Sections
% \ref{sec:sample_prev_d1}--\ref{sec:sample_prev_d3}). To obtain
% reliable results we perform multiple repetitions as follows:
%
\begin{itemize}
\item We compare the performance of each estimation technique on three
  datasets (Adult, COMPAS, and CreditCard). The datasets and
  respective preprocessing are described in detail in
  Section~\ref{sec:datasets}. We focus our discussion (and we present
  plots -- see Figures~\ref{fig:sample_prev_d3}--\ref{fig:q_wo_c_d3})
  on the experiments carried out on the Adult dataset, while we
  summarise numerically the results on COMPAS and CreditCard (Tables
  \ref{tab:sample_prev_d3}--\ref{tab:flip_prev_d1}), discussing them
  only when significant differences from Adult arise.
\item We divide a given data set into three %stratified
  subsets $\mathcal{D}_A, \mathcal{D}_B, \mathcal{D}_C$ of identical
  sizes and identical joint distribution over $(S,Y)$. We perform five
  random such splits; in order to test each estimator under the same
  conditions, these splits are the same for every method. For each
  split, we permute the role of the stratified subsets
  $\mathcal{D}_A, \mathcal{D}_B, \mathcal{D}_C$, so that each subset
  alternatively serves as the training set ($\mathcal{D}_1$), or
  auxiliary set ($\mathcal{D}_2$), or test set ($\mathcal{D}_3$). We
  test all (six) such permutations.
\item Whenever an experimental protocol requires sampling from a set,
  for instance when artificially altering a class prevalence value, we
  perform 10 different samplings. To perform extensive experiments at
  a reasonable computational cost, every time an experimental protocol
  requires changing a dataset $\mathcal{D}$ into a version
  $\breve{\mathcal{D}}$ characterized by distribution shift, we also
  reduce its cardinality to $|\breve{\mathcal{D}}|=500$. Further
  details and implications of this choice on each experimental
  protocol are provided in the context of the protocol's setup (e.g.,
  Section~\ref{sec:sample_prev_d1:setup}).
\item Different learning approaches can be used to train the sensitive
  attribute classifier $k_{s}$ underlying the quantification
  methods. We test Logistic Regression (LR) and Support Vector
  Machines (SVMs).\footnote{Some among the quantification methods we
  test in this study require the classifier to output posterior
  probabilities (as is the case for classifiers trained via LR). If a
  classifier natively outputs classification scores that are not
  probabilities (as is the case for classifiers trained via SVM), we
  convert the former into the latter via \citet{Platt:2000fk}'s
  probability calibration method.} Sections
  \ref{sec:sample_prev_d3}--\ref{sec:flip_prev_d1} report results of
  quantification algorithms wrapped around a classifier trained via
  LR. Analogous results obtained with SVMs are reported in Appendix
  \ref{app:svm}.
\item We train the classifier $h$, whose demographic disparity we aim
  to estimate, using LR with balanced class weights (i.e., loss
  weights inversely proportional to class frequencies).
\item To measure %the effect of a given factor on
  the performance of different quantifiers, we report the signed
  estimation error, derived from Equations~\eqref{eq:dd}
  and~\eqref{eq:dd_estim} as
  \begin{align}
    e = \hat{\delta}_{h}^{S} - \delta_{h}^{S} =\left [ \hat{\mu}(1) - \hat{\mu}(0) \right ] -
    \left [ \mu(1) - \mu(0) \right ] \label{eq:estim_err1}
  \end{align}
  \noindent We refer to $|e|$ as the Absolute Error (AE), and evaluate
  the results of our experiments by Mean Absolute Error (MAE) and Mean
  Squared Error (MSE), defined as
  \begin{align}
    \label{eq:mae}
    \mathrm{MAE}(E) & = \frac{1}{|E|}\sum_{e_i\in E}|e_i| \\
    \label{eq:mse}
    \mathrm{MSE}(E) & = \frac{1}{|E|}\sum_{e_i\in E}e_i^2 
  \end{align}

  \noindent where the mean of the signed estimation errors $e_i$ is
  computed over multiple experiments $E$. Overall, our experiments
  consist of over 700,000 separate estimations of demographic
  disparity.
  % We focus on binary outcomes and sensitive attributes, but the
  % definitions and techniques we employ can be straightforwardly
  % extended to a multi-class setting.
\end{itemize}

\noindent The remainder of this section is organized as
follows. Section~\ref{sec:datasets} presents the datasets that we have
chosen and the pre-processing steps we apply.
% in order to map the data points into vectors.
Sections~\ref{sec:sample_prev_d3}--\ref{sec:flip_prev_d1} motivate and
detail each of the five experimental protocols, reporting the
performance of different demographic disparity
estimators. \blue{Section~\ref{sec:fair_classifier} presents an experiment on fairness-aware methods,  where the classifier whose fairness we aim to estimate has been trained to optimize that measure.} Section~\ref{sec:q_not_c} shows that
% good estimators of demographic disparity are not necessarily good at
% classifying the sensitive attribute at an individual level, so that
reliable fairness auditing may be decoupled from undesirable misuse
aimed at inferring the values of the sensitive attribute at an
individual level. Finally, Section~\ref{sec:ablation} describes an
ablation study, aimed at investigating the benefits of training and
maintaining multiple class-specific quantifiers.

% -------------------------------------------------------------------

\subsection{Datasets}
\label{sec:datasets}

\noindent We perform our experiments on three datasets. We choose
Adult and COMPAS, the two most popular datasets in algorithmic fairness
research \citep{fabris2022algorithmic}, and Credit Card Default (hereafter: CreditCard), which
serves as a representative use case for a bank performing a fairness
audit of a prediction tool used internally. For each dataset, we
standardize the selected features by subtracting the mean and scaling
to unit variance.

\textbf{Adult}.\footnote{\url{https://archive.ics.uci.edu/ml/datasets/adult}}
One of the most popular resources in the UCI Machine Learning
Repository, the Adult dataset was curated to benchmark the performance
of machine learning algorithms. It was extracted from the March 1994
US Current Population Survey and represents respondents along
demographic and socioeconomic dimensions, reporting, e.g., their sex,
race, educational attainment, and occupation. Each instance comes with
a binary label, encoding whether their income exceeds \$50,000, which
is the target of the associated classification task. We consider
``sex'' the sensitive attribute $S$, with a binary categorization of
respondents as ``Female'' or ``Male''. From the non-sensitive
attributes $X$, we remove ``education-num'' (a redundant feature),
``relationship'' (where the values ``husband'' and ``wife'' are
near-perfect predictors of ``sex''), and ``fnlwgt'' (a variable
released by the US Census Bureau to encode how representative each
instance is of the overall population). Categorical variables are
dummy-encoded and instances with missing values (7\%) are removed.

\textbf{COMPAS}.\footnote{\url{https://github.com/propublica/compas-analysis}}
This dataset was curated to audit racial biases in the Correctional
Offender Management Profiling for Alternative Sanctions (COMPAS) risk
assessment tool, which estimates the likelihood of a defendant
becoming a recidivist \citep{angwin2016machine,larson2016how}. The
dataset represents defendants who were scored for risk of recidivism
by COMPAS in Broward County, Florida between 2013 and 2014,
summarizing their demographics, criminal record, custody, and COMPAS
scores. We consider the \texttt{compas-scores-two-years} subset
published by ProPublica on github, consisting of defendants who were
observed for two years after screening, for whom a binary recidivism
ground truth is available. We follow standard pre-processing to remove
noisy instances \citep{propublica2016compas}. We focus on ``race'' as a
protected attribute $S$, restricting the data to defendants labelled
``African-American'' or ``Caucasian''. Our attributes $X$ are the age
of the defendant (``age'', an integer), the number of juvenile
felonies, misdemeanours, and other convictions (``juv\_fel\_count'',
``juv\_misd\_count'', ``juv\_other\_count'', all integers), the number
of prior crimes (``priors\_count'', an integer) and the degree of
current charge (``c\_charge\_degree'', felony or misdemeanour,
dummy-encoded).

\textbf{CreditCard}.\footnote{\url{https://archive.ics.uci.edu/ml/datasets/default+of+credit+card+clients}. Note
that we discuss variables with the names they are given in the tabular
data (.xls file), which do not match those in the documentation.}
This resource was curated to study automated credit card default
prediction, following a wave of defaults in Taiwan. The dataset
summarizes the payment history of customers of an important Taiwanese
bank, from April to October 2005. Demographics, marital status, and
education of customers are also provided, along with the amount of
credit given and a binary variable encoding the default on payment
within the next month, which is the associated prediction task. We
consider ``sex'' (binarily encoded) as the sensitive attribute $S$ and
keep every other variable in $X$, preprocessing categorical ones via
dummy-encoding (``education'', ``marriage'', ``pay\_0'', ``pay\_2'',
``pay\_3'', ``pay\_4'', ``pay\_5'', ``pay\_6''). Differently from
Adult, we keep marital status as its values are not trivial predictors
of the sensitive attribute.
 
A summary of these datasets and related statistics is reported in
Table~\ref{tab:datastats}.

\begin{table}[ht!]
  \caption{Dataset statistics after preprocessing.}
  \label{tab:datastats}
  \centering
  \begin{tabular}{lrrr}
    \toprule
    Dataset & Adult  & COMPAS & CreditCard \\
    \midrule
    \# data points & 45,222 & 5,278 & 30,000  \\
    \# non-sensitive features & 84 & 6   &  81 \\
    sensitive attribute & sex & race  & sex \\
    $S=1$ & Male   & Caucasian & Male   \\
    $\Pr(S=1)$ & 0.675 &  0.398   &  0.396 \\
    target variable & income & recidivist   & default \\
    $Y=\oplus$ & $>$\$50,000   &  no   &  no \\
    $\Pr(Y=\oplus)$ & 0.248 &  0.498   &  0.779   \\
    \bottomrule
  \end{tabular}
\end{table}

\subsection{Distribution Shift Affecting the Test Set: Protocol
\texttt{sample-prev-$\mathcal{D}_3$}}
\label{sec:sample_prev_d3}

% -------------------------------------------------------------------

\subsubsection{Motivation and Setup}

\noindent The first experimental protocol models a setting in which
the test set $\mathcal{D}_3$ shows a significant distribution shift
with respect to the sets $\mathcal{D}_1$ and $\mathcal{D}_2$ available
during training of $h$ and $k$. In other words, in this protocol,
$\mathcal{D}_1$ and $\mathcal{D}_2$ are marginalisations of the same
joint distribution, while $\mathcal{D}_3$ (more precisely
$\breve{\mathcal{D}}_3$) is drawn from a different joint
distribution. %\fabsebcomment{Qui forse bisognerebbe caratterizzare la cosa meglio, perché ha perfettamente senso parlare di distribution shift fra $\mathcal{D}_1$ e $\mathcal{D}_3$, molto meno fra $\mathcal{D}_2$ e $\mathcal{D}_3$, dato che usano codeframes differenti.} \afabcomment{Vedi ``In other words''}
We consider two sub-protocols
(\texttt{sample-prev-$\mathcal{D}_{3}^{\ominus}$} and
\texttt{sample-prev-$\mathcal{D}_{3}^{\oplus}$}) that model changes in
the distribution of a sensitive variable $S$ in
$\mathcal{D}_{3}^{\ominus}$ and $\mathcal{D}_{3}^{\oplus}$, the test
subsets of either negatively or positively predicted instances. More
in detail, we let $\Pr(s|\ominus)$ (or its dual $\Pr(s|\oplus)$) in
$\breve{\mathcal{D}}_3$ range on eleven evenly spaced values between 0
and 1. For example, under sub-protocol
\texttt{sample-prev-$\mathcal{D}_{3}^{\ominus}$}, we vary the
distribution of sensitive attribute $S$ in
$\breve{\mathcal{D}}_{3}^{\ominus}$, so that
$\Pr(s|\ominus) \in \left \lbrace 0.0, 0.1 \dots, 0.9, 1.0 \right
\rbrace$, while keeping the distribution in
$\breve{\mathcal{D}}_{3}^{\oplus}$ fixed. For both sub-protocols, in
each repetition we sample subsets of the test set $\mathcal{D}_3$ such
that
$|\breve{\mathcal{D}}_{3}^{\ominus}|=|\breve{\mathcal{D}}_{3}^{\oplus}|=500$.
Pseudocode~\ref{pseudo:sample-prev-d3} describes the protocol when
acting on $\mathcal{D}_{3}^{\ominus}$; the case for
$\mathcal{D}_{3}^{\oplus}$ is analogous and consists of swapping the
roles of $\mathcal{D}_{3}^{\ominus}$ and $\mathcal{D}_{3}^{\oplus}$ in
Lines~\ref{line:randomsample:d30} and~\ref{line:randomsample:d31}. The
pale red region highlights the part of the experimental protocol that
is specific to Protocol \texttt{sample-prev-$\mathcal{D}_3$}; the rest
is common to all the experimental protocols mentioned in this paper.

% Note that the pale red region (delimiting the protocol-specific
% computations) comprises the computation of the error
% (line~\ref{line:error:p4}) since, in this case, the error is
% computed with respect to the samples
% $\breve{\mathcal{D}}_{3}^{\ominus}$ and
% $\breve{\mathcal{D}}_{3}^{\oplus}$, and not with respect to
% $\mathcal{D}_{3}^{\ominus}$ and $\mathcal{D}_{3}^{\oplus}$.
% \alex{potremmo riportare il numero totale di esperimenti per ciascun
% boxplot}

This protocol accounts for the inevitable evolution of phenomena,
especially those related to human behaviour. Indeed, it is common in
real-world scenarios for data generation processes to be nonstationary
and change across development and deployment, due, e.g., to
seasonality, changes in the spatiotemporal application context, or any
sort of unmodelled novelty and difference in populations
\citep{morenotorres2012:unifying,ditzler2015:learning,malinin2021shifts}. Given
that most work on algorithmic fairness focuses on decisions or
predictions about people, and given inevitable changes in human lives,
values, and behaviour, the above considerations about non-stationarity
seem particularly relevant. For example, data available from one
population is often repurposed to train algorithms that will be
deployed on a different population, requiring ad hoc fair learning
approaches \citep{coston2019:fair} and evoking the \emph{portability
trap} of fair machine learning \citep{selbst2019:fairness}. In
addition, agents can respond to novel technology in their social
context and adapt their behaviour accordingly
\citep{hu2019:disparate,tsirtsis2019optimal}, causing \emph{ripple
effects} \citep{selbst2019:fairness} and \emph{feedback loops}
\citep{mansoury2020:feedback}. \replace{Furthermore, as a concrete (although spurious) example of shift in a popular fairness-related dataset, the repeated offence rate for defendants in the COMPAS dataset
%\citep{larson2016how} 
increases sharply between 2013 and 2014 %\citep{barenstein2019propublica,biswas2021ensuring}
.}{} Finally,
personalized pricing constitutes an increasingly possible practice
with nontrivial fairness concerns \citep{kallus2021:fairness} and
inevitable shifts due to changing habits and environments
\citep{sindreu2021covid19}.

% In the quantification literature, this is the most common evaluation
% protocol. Similarly to \texttt{sample-prev-$\mathcal{D}_2$}, it
% imposes shifts in the estimand between the training and testing
% conditions of a quantifier, represented by the auxiliary set
% $\mathcal{D}_2$ and the test set $\mathcal{D}_3$, respectively.

In this protocol, quantifiers are tested on subsets
$\breve{\mathcal{D}}_{3}^{\ominus}$,
$\breve{\mathcal{D}}_{3}^{\oplus}$ that exhibit a different prevalence
of sensitive attribute $s$ with respect to their counterparts
$\mathcal{D}_{2}^{\ominus}$, $\mathcal{D}_{2}^{\oplus}$ in the
auxiliary set. More specifically, with this protocol we vary the joint
distribution of $(S,\hat{Y})$ to directly influence the demographic
disparity of the classifier $h$ in the test set $\mathcal{D}_{3}$, and
move it away from the value $\delta_{h}^{S}$ of the same measure that
we would obtain on the set $\mathcal{D}_{2}$. This is a fundamental
evaluation protocol, as it makes our estimand different across
$\mathcal{D}_2$ and $\mathcal{D}_3$ (or, more precisely, its modified
version $\breve{\mathcal{D}}_3$), which is typically expected in
practice. If this was not the case, a practitioner could simply resort
to an explicit calculation of the demographic disparity in the
auxiliary set $\mathcal{D}_2$ and consider it representative of any
deployment condition, \blue{as in the MLPE trivial baseline}. Given this reasoning, this protocol imposes
sizeable variations in the demographic disparity of $h$ between
$\mathcal{D}_2$ and $\mathcal{D}_3$, which act as the training set and
the test set, respectively, for our quantifiers. For example, on
Adult, $\delta_{h}^{S}$ is approximately equal to $0.3$ in
$\mathcal{D}_2$, while in $\mathcal{D}_3$ we let it vary in the range
$[-0.7, 0.9]$. Despite these sizeable
variations, we expect that methods such as SLD, ACC, and PACC perform
well, due to their proven unbiasedness in this setting (Remark
\ref{rem:ql_fisher}).

\begin{algorithm}[t]
  \LinesNumbered \SetNoFillComment
 
 \begin{footnotesize}
   \SetKwInOut{Input}{Input} \SetKwInOut{Output}{Output} \Input{
   \textbullet\ Dataset $\mathcal{D}$ ; \\
   \hspace{0.35em}\textbullet\ Classifier learner CLS; \\
   \hspace{0.35em}\textbullet\ Quantification method Q; \\
   } \Output{
   \textbullet\ MAE of the demographic disparity estimates ; \\
   \hspace{0.35em}\textbullet\ MSE of the demographic disparity
   estimates ;}
 
   \BlankLine

   $E\leftarrow\emptyset$ ;
 
   \For{5 random splits}{
   $\mathcal{D}_A,\mathcal{D}_B,\mathcal{D}_C \leftarrow
   \mathrm{split\_stratify}(\mathcal{D})$ ;
 
   \For{$\mathcal{D}_1,\mathcal{D}_2,\mathcal{D}_3 \in
   \mathrm{permutations}(\mathcal{D}_A,\mathcal{D}_B,\mathcal{D}_C)$}{
 
   \tcc{Learn a classifier $h : \mathcal{X} \rightarrow \mathcal{Y}$}
   $h\leftarrow$ CLS.fit($\mathcal{D}_1$) ;
 
   $\mathcal{D}_{2}^{\ominus}\leftarrow\{(\mathbf{x}_i,s_i)\in\mathcal{D}_2
   \;|\; h(\mathbf{x}_i)=\ominus\}$ ;
 
   $\mathcal{D}_{2}^{\oplus}\leftarrow\{(\mathbf{x}_i,s_i)\in\mathcal{D}_2
   \;|\; h(\mathbf{x}_i)=\oplus\}$ ;
 
   \tcc{Learn quantifiers $q_y : 2^\mathcal{X} \rightarrow [0,1]$}
 
   $q_{\ominus} \leftarrow $ Q.fit($\mathcal{D}_{2}^{\ominus})$ ;
 
   $q_{\oplus} \leftarrow $ Q.fit($\mathcal{D}_{2}^{\oplus})$ ;
 
   \tcc{Split instances in $\mathcal{D}_3$ based on predicted labels
   from $h$}
 
   $\mathcal{D}_{3}^{\ominus}\leftarrow\{\mathbf{x}_i\in\mathcal{D}_3
   \;|\; h(\mathbf{x}_i)=\ominus\}$ ;
 
   $\mathcal{D}_{3}^{\oplus}\leftarrow\{\mathbf{x}_i\in\mathcal{D}_3
   \;|\; h(\mathbf{x}_i)=\oplus\}$ ;
 
   \For{10 repeats}{

   \tikzmk{A} \For{$p \in \{0.1,0.2,\ldots,0.9\}$}{
 
   \tcc{Generate samples from $\mathcal{D}_{3}^{\ominus}$ at desired
   prevalence and size, and uniform samples from
   $\mathcal{D}_{3}^{\oplus}$ at desired size}
 
   $\breve{\mathcal{D}}_{3}^{\ominus}\sim\mathcal{D}_{3}^{\ominus}$
   with $p_{\breve{\mathcal{D}}_{3}^{\ominus}}(s)=p$ and
   $|\breve{\mathcal{D}}_{3}^{\ominus}|=500$ \label{line:randomsample:d30}
   ;
 
   $\breve{\mathcal{D}}_{3}^{\oplus}\sim\mathcal{D}_{3}^{\oplus}$ with
   $|\breve{\mathcal{D}}_{3}^{\oplus}|=500$ \label{line:randomsample:d31}
   ;
 
   \tcc{Use quantifiers to estimate demographic prevalence}

   $\hat{p}^{q_{\ominus}}_{\breve{\mathcal{D}}_{3}^{\ominus}}(s)
   \leftarrow q_{\ominus}(\breve{\mathcal{D}}_{3}^{\ominus})$ ;
 
   $\hat{p}^{q_{\oplus}}_{\breve{\mathcal{D}}_{3}^{\oplus}}(s)
   \leftarrow q_{\oplus}(\breve{\mathcal{D}}_{3}^{\oplus})$ ;

   \tcc{Compute the signed error of the demographic disparity
   estimate} $e \leftarrow $ compute error using
   $\hat{p}^{q_{\ominus}}_{\breve{\mathcal{D}}_{3}^{\ominus}}(s)$,
   $\hat{p}^{q_{\oplus}}_{\breve{\mathcal{D}}_{3}^{\oplus}}(s)$ and
   Equation~\eqref{eq:estim_err1}
   % $e=\left [
   %   \hat{p}_{\breve{\mathcal{D}}_{3}^{\ominus}}^{q_{\ominus}}(s) -
   %   \hat{p}_{\breve{\mathcal{D}}_{3}^{\oplus}}^{q_{\oplus}}(s)
   % \right ] - \left [ p_{\breve{\mathcal{D}}_{3}^{\ominus}}(s) -
   %   p_{\breve{\mathcal{D}}_{3}^{\oplus}}(s) \right
   % ]$ \label{line:error:p4};
 
   \tikzmk{B} \boxit{pink} $E \leftarrow E \cup \{e\}$}}}}
 
   $\mathrm{mae}\leftarrow\mathrm{MAE}(E)$ ;
 
   $\mathrm{mse}\leftarrow\mathrm{MSE}(E)$ ;

   \Return{$\mathrm{mae}$, $\mathrm{mse}$} \BlankLine

   \caption{Protocol \texttt{sample-prev-$\mathcal{D}_3$}, shown for
   variations of prevalence values in class $y=\ominus$.}
   \label{pseudo:sample-prev-d3}
 \end{footnotesize}
\end{algorithm}

% -------------------------------------------------------------------

\subsubsection{Results}

\begin{figure}[t]%[h!]
  \centering
  \begin{subfigure}{\mysize\textwidth}
    \centering
    \includegraphics[width=\mysize\textwidth]{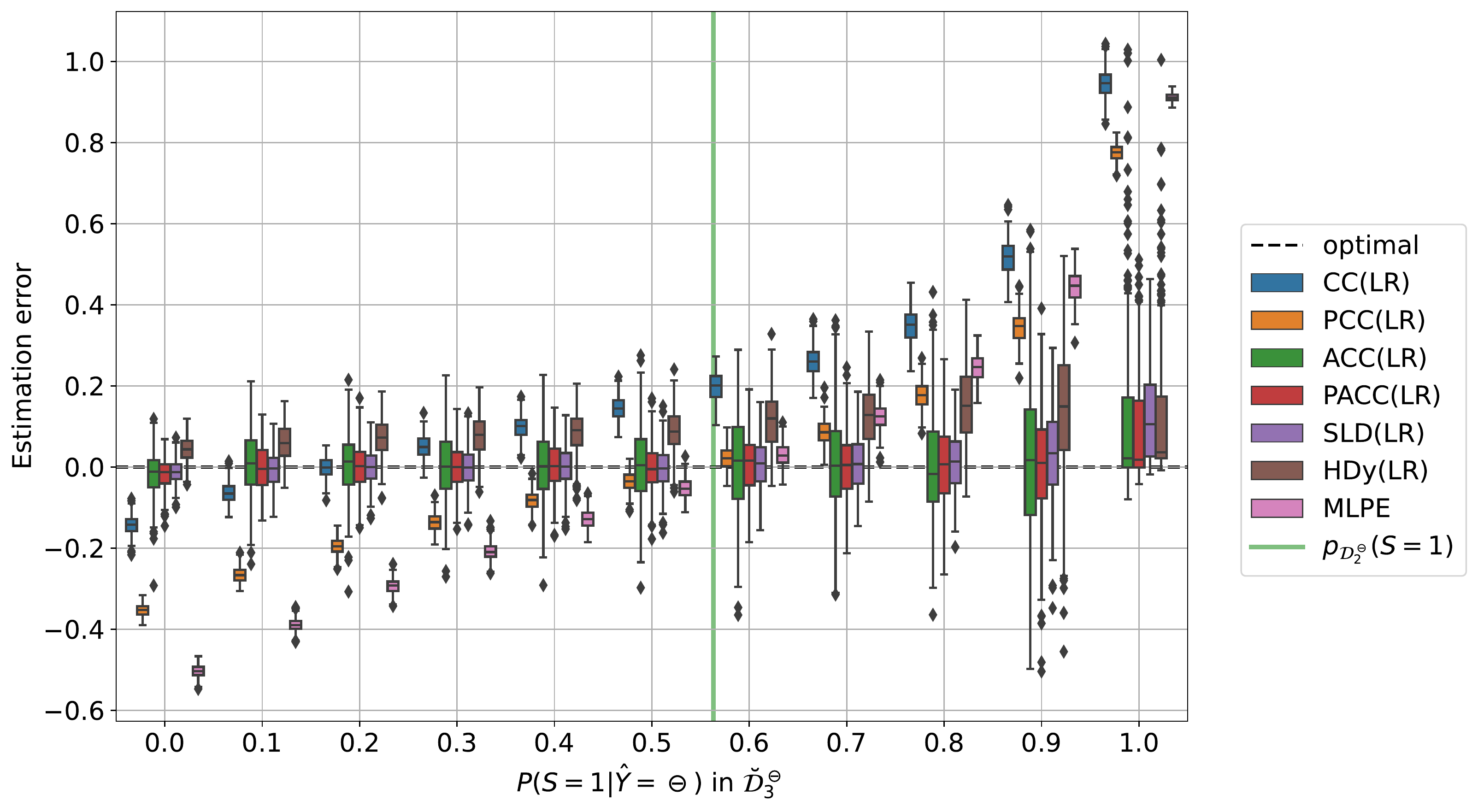}
    \caption{Protocol
    \texttt{sample-prev-$\mathcal{D}_{3}^{\ominus}$}}
    \label{fig:sample_prev_d30}
  \end{subfigure} \\
  \begin{subfigure}{\mysize\textwidth}
    \centering
    \includegraphics[width=\mysize\textwidth]{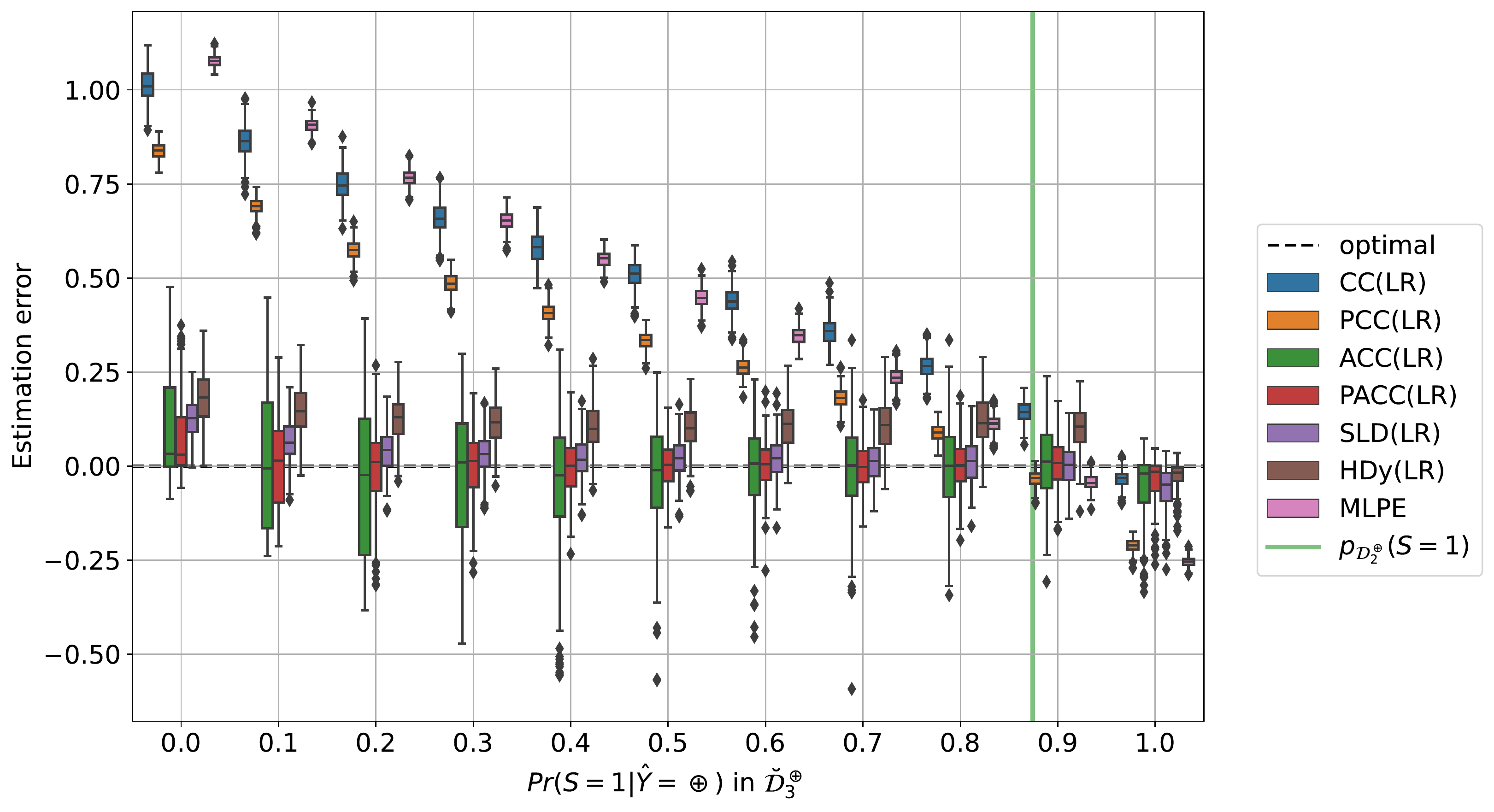}
    \caption{Protocol \texttt{sample-prev-$\mathcal{D}_{3}^{\oplus}$}}
    \label{fig:sample_prev_d31}
  \end{subfigure}
  \caption{Experiments conducted according to protocol
  \texttt{sample-prev-$\mathcal{D}_3$} on the Adult dataset. The
  figure shows the distribution of the estimation error (on the $y$
  axis) as $\breve{\mathcal{D}}_3$ is sampled with a given
  $\Pr(S=1|Y=\ominus)$ value (a) or with a given $\Pr(S=1|Y=\oplus)$
  value (b), which are shown on the $x$ axis. The green line indicates
  the value of $\Pr(S=1)$ as observed in $\mathcal{D}_{2}^{\ominus}$
  (a) or in $\mathcal{D}_{2}^{\oplus}$ (b).}
  \label{fig:sample_prev_d3}
\end{figure}

\noindent In Figure~\ref{fig:sample_prev_d3} we report the performance
of CC, PCC, ACC, PACC, SLD, \replace{and HDy}{HDy, and MLPE} on the Adult dataset under the
\texttt{sample-prev-$\mathcal{D}_3$} experimental protocol. The
estimation error (Equation~\ref{eq:estim_err1}) is reported on the $y$
axis, as we vary the prevalence of the protected group in the test
set, which is displayed on the $x$
axis. Figure~\ref{fig:sample_prev_d30} concentrates on prevalence
variations in $\mathcal{D}_{3}^{\ominus}$, while Figure
\ref{fig:sample_prev_d31} considers variations of the prevalence of
the protected group in $\mathcal{D}_{3}^{\oplus}$. Each boxplot
summarizes the results of 5 random splits, 6 role permutations, and 10
samplings of $\breve{\mathcal{D}}_3$, for a total of 300 repetitions
for each combination of 6 methods and 11 values which vary on the $x$
axis. Boxes enclose the two central quartiles (separated by a median
horizontal line), while whiskers surround points in the outer
quartiles, except for outliers marked with diamonds.

Similar trends emerge under both sub-protocols. \replace{CC and PCC}{CC, PCC, and MLPE} display a
clear trend along the $x$ axis, vastly over- or underestimating the
demographic disparity of $h$, and proving unreliable in settings where
the prevalence values in the unlabelled (test) set shift away from the
prevalence values of the training set. In sub-protocol
\texttt{sample-prev-$\mathcal{D}_{3}^{\oplus}$}, summarised in
Figure~\ref{fig:sample_prev_d31}, the prevalence of men ($S=1$) in
$\breve{\mathcal{D}}_{3}^{\oplus}$, used to test one of the
quantifiers, is almost always lower than the prevalence in the
respective training set $\mathcal{D}_{2}^{\oplus}$, reported with a
vertical green line. As a result, quantifiers trained on
$\mathcal{D}_{2}^{\oplus}$ tend to systematically overestimate the
prevalence of males in $\mathcal{D}_{3}^{\oplus}$, thus also
overestimating $\mu(1)$ and $\delta_{h}^{S}$, according to
Equations~\eqref{eq:dd_estim} and \eqref{eq:mu_ql}. Similar
considerations hold for sub-protocol
\texttt{sample-prev-$\mathcal{D}_{3}^{\ominus}$}, with a sign flip.
% due to $p_{\mathcal{D}_{3}^{\ominus}}(s)$ influencing the second
% term in Equation~\ref{eq:dd_estim}.

\begin{table}[tb]
  \small \centering
  \caption{Results obtained in the experiments run according to
  protocol \texttt{sample-prev-$\mathcal{D}_3$}.}
  \label{tab:sample_prev_d3}
  % \resizebox{\textwidth}{!}{%
  
\begin{tabular}{llcccc} \toprule
 & & $\downarrow$ MAE & $\downarrow$ MSE & $\uparrow$ $P(\mathrm{AE}<0.1)$ & $\uparrow$ $P(\mathrm{AE}<0.2)$ \\ \midrule
\multirow{7}{*}{Adult} &CC(LR) & 0.382$^{\phantom{\ddag}}\pm0.304$ & 0.239$^{\phantom{\ddag}}\pm0.305$ & 0.207 & 0.386 \\
 &PCC(LR) & 0.299$^{\phantom{\ddag}}\pm0.237$ & 0.146$^{\phantom{\ddag}}\pm0.199$ & 0.235 & 0.427 \\
 &ACC(LR) & 0.103$^{\phantom{\ddag}}\pm0.097$ & 0.020$^{\phantom{\ddag}}\pm0.047$ & 0.595 & 0.870 \\
 &PACC(LR) & 0.061$^{\phantom{\ddag}}\pm0.059$ & 0.007$^{\phantom{\ddag}}\pm0.016$ & 0.813 & 0.970 \\
 &SLD(LR) & \textbf{0.055}$^{\phantom{\ddag}}\pm0.052$ & \textbf{0.006}$^{\phantom{\ddag}}\pm0.012$ & \textbf{0.846} & \textbf{0.980} \\
 &HDy(LR) & 0.110$^{\phantom{\ddag}}\pm0.079$ & 0.018$^{\phantom{\ddag}}\pm0.032$ & 0.500 & 0.893 \\
 &MLPE & 0.397$^{\phantom{\ddag}}\pm0.298$ & 0.246$^{\phantom{\ddag}}\pm0.316$ & 0.162 & 0.294 \\

 \vspace{0.01cm} \\ 
\multirow{7}{*}{COMPAS} &CC(LR) & 0.541$^{\phantom{\ddag}}\pm0.369$ & 0.429$^{\phantom{\ddag}}\pm0.472$ & 0.118 & 0.237 \\
 &PCC(LR) & 0.337$^{\phantom{\ddag}}\pm0.242$ & 0.172$^{\phantom{\ddag}}\pm0.214$ & 0.181 & 0.344 \\
 &ACC(LR) & 0.495$^{\phantom{\ddag}}\pm0.363$ & 0.377$^{\phantom{\ddag}}\pm0.471$ & 0.143 & 0.252 \\
 &PACC(LR) & 0.252$^{\phantom{\ddag}}\pm0.213$ & 0.109$^{\phantom{\ddag}}\pm0.184$ & 0.287 & 0.492 \\
 &SLD(LR) & \textbf{0.169}$^{\phantom{\ddag}}\pm0.139$ & \textbf{0.048}$^{\phantom{\ddag}}\pm0.077$ & \textbf{0.385} & \textbf{0.669} \\
 &HDy(LR) & 0.267$^{\phantom{\ddag}}\pm0.213$ & 0.116$^{\phantom{\ddag}}\pm0.176$ & 0.250 & 0.472 \\
 &MLPE & 0.349$^{\phantom{\ddag}}\pm0.249$ & 0.184$^{\phantom{\ddag}}\pm0.227$ & 0.175 & 0.332 \\

 \vspace{0.01cm} \\ 
\multirow{7}{*}{CreditCard} &CC(LR) & 0.345$^{\phantom{\ddag}}\pm0.241$ & 0.177$^{\phantom{\ddag}}\pm0.212$ & 0.172 & 0.339 \\
 &PCC(LR) & 0.325$^{\phantom{\ddag}}\pm0.213$ & 0.151$^{\phantom{\ddag}}\pm0.157$ & 0.176 & 0.340 \\
 &ACC(LR) & 0.341$^{\phantom{\ddag}}\pm0.259$ & 0.183$^{\phantom{\ddag}}\pm0.256$ & 0.189 & 0.367 \\
 &PACC(LR) & 0.259$^{\phantom{\ddag}}\pm0.211$ & 0.111$^{\phantom{\ddag}}\pm0.173$ & 0.269 & 0.480 \\
 &SLD(LR) & \textbf{0.190}$^{\phantom{\ddag}}\pm0.148$ & \textbf{0.058}$^{\phantom{\ddag}}\pm0.086$ & \textbf{0.334} & \textbf{0.609} \\
 &HDy(LR) & 0.251$^{\phantom{\ddag}}\pm0.190$ & 0.099$^{\phantom{\ddag}}\pm0.142$ & 0.248 & 0.478 \\
 &MLPE & 0.334$^{\phantom{\ddag}}\pm0.218$ & 0.159$^{\phantom{\ddag}}\pm0.165$ & 0.172 & 0.330 \\
\bottomrule
\end{tabular}%

  % }%
\end{table}

ACC, PACC, SLD and HDy, on the other hand, display low bias, even
under sizeable prevalence shift. Their variance is higher than CC and
PCC, but their estimation error is moderate overall. The condition
$\Pr(S=1|\hat{Y}=\ominus)=1$ (right-most point in Figure
\ref{fig:sample_prev_d30}) is particularly critical for every method
due to $p_{\mathcal{D}_{3}}(s=0)$ dropping below 0.1, thus making
small estimation errors for the denominator of Equation \ref{eq:mu_ql}
especially impactful on $\hat{\mu}(0)$.

% Analogous results obtain on COMPAS and CreditCard, as summarised in
% Table~\ref{tab:sample_prev_d3}. Under large shifts between the
% auxiliary and test set, the estimation of demographic disparity is
% more difficult on COMPAS and CreditCard than on the Adult dataset.

The results of the COMPAS and CreditCard datasets are reported in
Table \ref{tab:sample_prev_d3}, along with a summary of the results of
the Adult dataset we have just discussed. The first and second columns
indicate the MAE and MSE values (lower is better), while the third and
fourth columns indicate the probability that the Absolute Error (AE)
falls below 0.1 and 0.2 across the entire experimental protocol
(higher is better). \textbf{Boldface} indicates the best method for a
given dataset and metric. The superscripts $\dag$ and $\ddag$ denote
the methods (if any) whose error scores (MAE, MSE) are \emph{not}
statistically significantly different from the best according to a
paired sample, two-tailed t-test at different confidence
levels. Symbol $\dag$ indicates $0.001<p$-value $<0.05$ while symbol
$\ddag$ indicates $0.05\leq p$-value; the absence of any such symbol
indicates $p$-value $\leq 0.001$ (i.e., that the performance of the
method is statistically significantly different from that of the best
method).  Overall, SLD strikes the best balance between bias and
variance. PACC is the second-best approach, outperforming ACC and PCC,
demonstrating the utility of combining posterior probabilities and
adjustments when the latter can reliably be estimated. The trends we
discussed also hold for COMPAS and CreditCard. Note that both datasets
appear to provide a setting harder than Adult for the inference of the
sensitive attribute $S$ from the non-sensitive attributes $X$.

% -------------------------------------------------------------------

% -------------------------------------------------------------------

\subsection{Distribution Shift Affecting the Auxiliary Set: Protocol
\texttt{sample-prev-$\mathcal{D}_2$}}
\label{sec:sample_prev_d2}

% -------------------------------------------------------------------

\subsubsection{Motivation and
Setup}\label{sec:sample_prev_d2_motivation}

\noindent This protocol is analogous to protocol
\texttt{sample-prev-$\mathcal{D}_3$} (Section
\ref{sec:sample_prev_d3}), but for the fact that it focuses on shifts
in the auxiliary set $\mathcal{D}_2$, while $\mathcal{D}_1$ and
$\mathcal{D}_3$ remain at their natural prevalence. Similarly to
Section \ref{sec:sample_prev_d3}, we assess the signed estimation
error under shifts that affect $\mathcal{D}_{2}^{\ominus}$ or
$\mathcal{D}_{2}^{\oplus}$, that is, the subsets of $\mathcal{D}_2$
labelled positively or negatively by the classifier $h$. Here too, we
consider two experimental sub-protocols, describing variations in the
prevalence of sensitive attribute $s$ in either subset. More
specifically, we let $\Pr(s|\ominus)$ (or its dual $\Pr(s|\oplus)$)
take 9 evenly spaced values between 0.1 and 0.9. \replace{We avoid extreme
values of 0 and 1, which would make either demographic group $S=0$ or
$S=1$ absent from the training set of one quantifier. For example, in
sub-protocol \texttt{sample-prev-$\mathcal{D}_{2}^{\ominus}$} we let
the prevalence $\Pr(s|\ominus)$ in $\breve{\mathcal{D}}_{2}^{\ominus}$
take values in
$\left \lbrace 0.1, 0.2 \dots, 0.8, 0.9 \right \rbrace$, while the
remaining subset $\breve{\mathcal{D}}_{2}^{\oplus}$ remains at its
natural prevalence $\Pr(s|\oplus)$.
% \footnote{The natural prevalence is matched allowing for small
% fluctuations due to subsampling.}
For each repetition, we set
$|\breve{\mathcal{D}}_{2}^{\ominus}|=|\breve{\mathcal{D}}_{2}^{\oplus}|=500$. This
makes for a challenging quantification setting and allows for fast
training of multiple quantifiers across many repetitions.}{}
Pseudocode~\ref{pseudo:sample-prev-d2} describes the protocol when
acting on $\mathcal{D}_{2}^{\ominus}$; the case for
$\mathcal{D}_{2}^{\oplus}$ is analogous, and comes down to swapping
the roles of $\mathcal{D}_{2}^{\ominus}$ and
$\mathcal{D}_{2}^{\oplus}$ in Lines~\ref{line:randomsample:d20}
and~\ref{line:randomsample:d21}.

This protocol captures issues of representativity in demographic data,
e.g., due to nonuniform response rates across subpopulations
\citep{schouten2009indicators,schouten2012evaluating}. Given the
importance of trust for the provision of one's sensitive attributes,
in some domains this provision is considered akin to a \emph{data
donation} \citep{andrus2021:wm}. Individuals from groups that were
historically served with worse quality or had lower acceptance rates
for a service can be reluctant to disclose their membership in those
groups, fearing that it may be used against them as grounds for
rejection or discrimination \citep{hasnain2006obtaining}. This may be
especially true for individuals who perceive to be at high risk of
rejection, and this can cause complex selection biases, jointly
dependent on $S$ and $Y$, or $S$ and $\hat{Y}$ if individuals have
some knowledge of the classification procedure. For example, health
care providers may be advised to collect information about the race of
patients to monitor the quality of services across subpopulations. In
a field study, 28\% of patients reported discomfort in revealing their
own race to a clerk, with African-American patients significantly less
comfortable than white patients on average \citep{baker2005patients}.

\replace{Through this protocol, we may expect to find patterns similar to those
highlighted in Section~\ref{sec:sample_prev_d3}, with the roles of the
auxiliary set $\mathcal{D}_2$ and the test set $\mathcal{D}_3$ now
switched. Under this protocol, $\mathcal{D}_2$ has a lower cardinality
and variable prevalence (and is noted by $\breve{\mathcal{D}}_2$ for
this reason), while $\mathcal{D}_3$ is left to its original
cardinality and prevalence of the sensitive attribute $s$.}{}

% -------------------------------------------------------------------

\subsubsection{Results}

\begin{figure}[t]%[h!]
  \centering
  \begin{subfigure}{\mysize\textwidth}
    \centering
    \includegraphics[width=\mysize\textwidth]{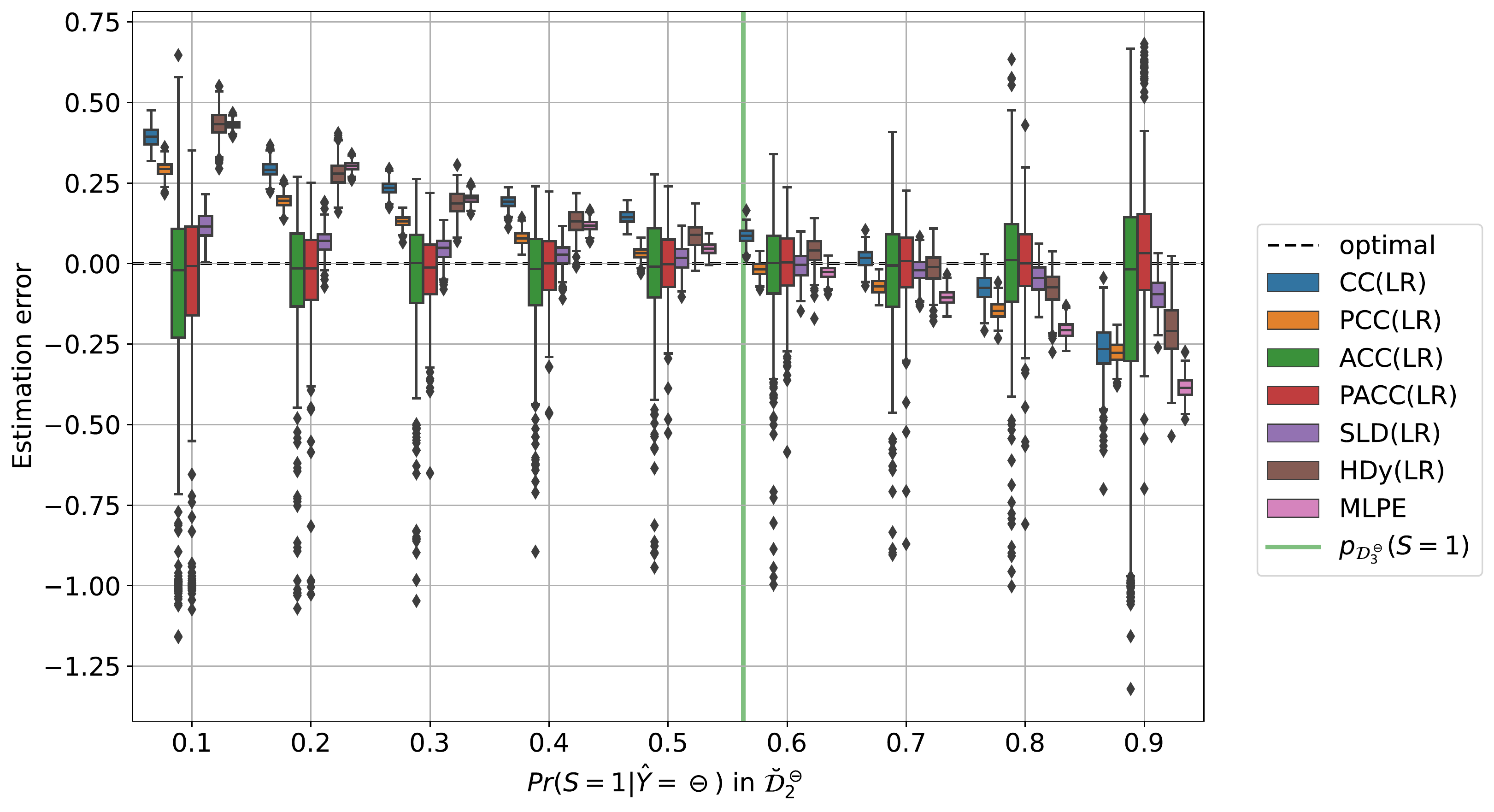}
    \caption{Protocol
    \texttt{sample-prev-$\mathcal{D}_{2}^{\ominus}$}}
    \label{fig:sample_prev_d20}
  \end{subfigure} \\
  \begin{subfigure}{\mysize\textwidth}
    \centering
    \includegraphics[width=\mysize\textwidth]{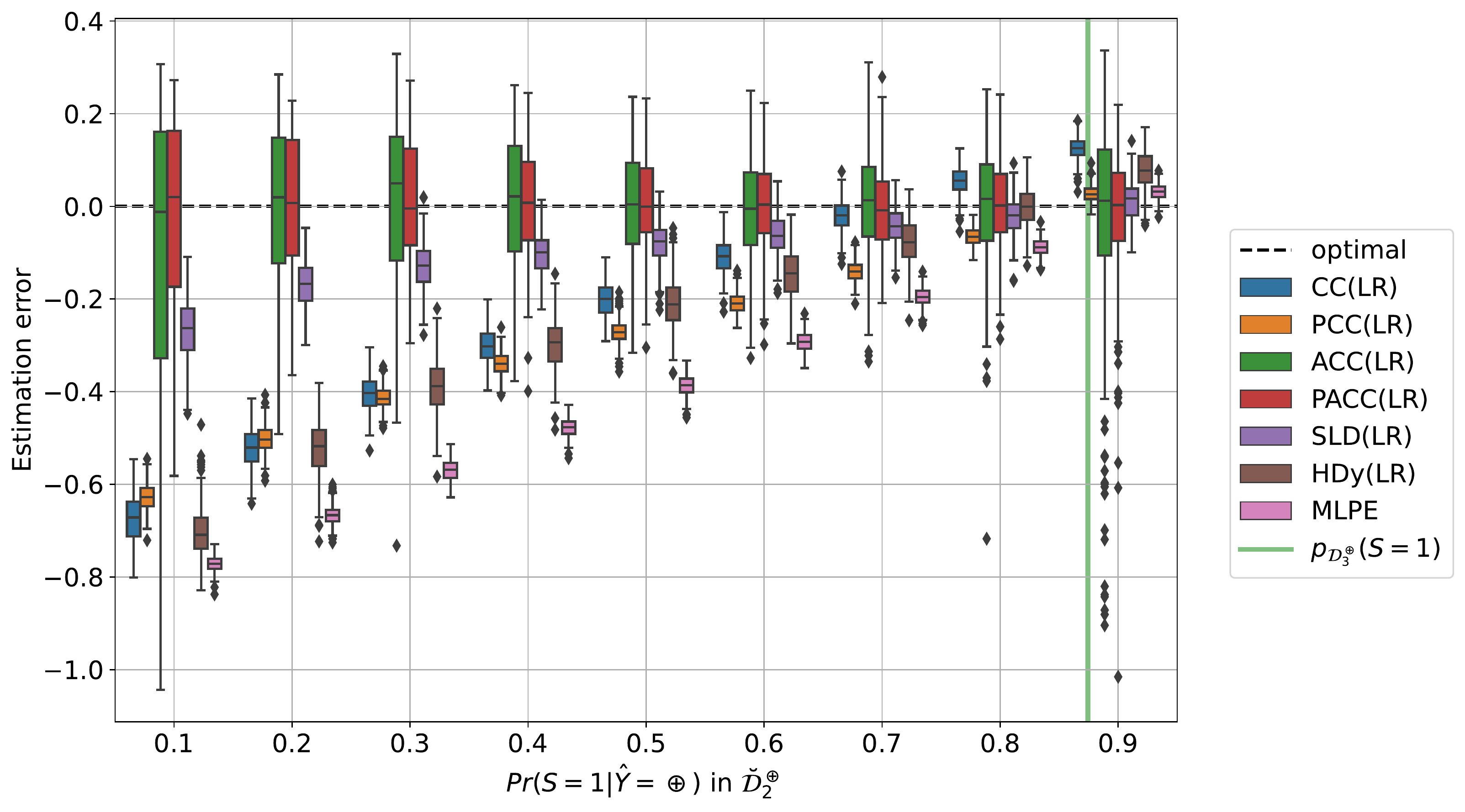}
    \caption{Protocol \texttt{sample-prev-$\mathcal{D}_{2}^{\oplus}$}}
    \label{fig:sample_prev_d21}
  \end{subfigure}
  \caption{Protocol \texttt{sample-prev-$\mathcal{D}_2$} on the Adult
  dataset. Distribution of the estimation error ($y$ axis) as
  $\breve{\mathcal{D}}_2$ is sampled with a given $\Pr(S=1|Y=\ominus)$
  value, plot (a), or $\Pr(S=1|Y=\oplus)$ value, plot (b) ($x$ axis).
  The green line indicates the value of $\Pr(S=1)$ as observed in
  $\mathcal{D}_{3}^{\ominus}$, plot (a), or
  $\mathcal{D}_{3}^{\oplus}$, plot (b).}
  \label{fig:sample_prev_d2}
\end{figure}

\noindent Figure~\ref{fig:sample_prev_d2} shows the signed estimation
error on the $y$ axis, as we vary, on the $x$ axis, the prevalence of
the sensitive attribute in $\mathcal{D}_{2}^{\ominus}$
(Figure~\ref{fig:sample_prev_d20}) and $\mathcal{D}_{2}^{\oplus}$
(Figure~\ref{fig:sample_prev_d21}). \blue{MLPE,} CC, PCC, and HDy prove to be
fairly sensitive to shifts in their training set. In sub-protocol
\texttt{sample-prev-$\mathcal{D}_{2}^{\oplus}$}, symmetrically to the
sub-protocol \texttt{sample-prev-$\mathcal{D}_{3}^{\oplus}$} discussed
in the previous section, the prevalence of males ($S=1$) in subset
$\mathcal{D}_{2}^{\oplus}$, used to train one of the quantifiers, is
almost always lower than the prevalence in the respective test subset
$\mathcal{D}_{3}^{\oplus}$, indicated with a vertical green line.
% \footnote{This is true for every data point on the $x$ axis in
% Figure~\ref{fig:sample_prev_d21}, except for the rightmost one,
% where
% $p_{\mathcal{D}_{2}^{\oplus}}(s)=0.9>p_{\mathcal{D}_{3}^{\oplus}}(s)=0.85$.}
As a result, quantifiers trained on $\mathcal{D}_{2}^{\oplus}$ tend to
systematically underestimate the prevalence of males in
$\mathcal{D}_{3}^{\oplus}$ and underestimate the (signed) demographic
disparity of the classifier
$h$. %Similar considerations hold for sub-protocol \texttt{sample-prev-$\mathcal{D}_{2}^{\ominus}$}, with a sign flip due to $p_{\mathcal{D}_{3}^{\ominus}}(s)$ being the first term in Equation~\ref{eq:dd}.
%\wasreplace{While it could be anticipated for CC and PCC, this behavior
%was unexpected for HDy and not documented in the quantification
%literature. Although less substantial, SLD shows a similar trend.}{}

ACC and PACC require splitting their training set to estimate the
respective adjustments
(Equations~\eqref{eq:tprandfpr}--\eqref{eq:pacc}), and suffer from a
reduced cardinality $|\breve{\mathcal{D}}_2|=1,000$. Their performance
worsens substantially with respect to protocol
\texttt{sample-prev-$\mathcal{D}_{3}$}, where
$|\mathcal{D}_2|>15,000$.
% As discussed above, HDy displays a large bias. ACC and PACC, on the
% other hand, have a good median performance.
Indeed, these methods have been shown to be \emph{Fisher-consistent}
under prior probability shift \citep{tasche2017:fisher,Vaz:2019eu},
that is, they are guaranteed to be accurate, thanks to the respective
adjustments, if $\mathcal{D}_2$ is large enough and linked to
$\mathcal{D}_3$ by prior probability shift. While the latter condition
holds, the former is violated under this protocol, hence ACC and PACC
are unbiased (in expectation), but display a large variance, due to
unstable adjustments. \wasblue{SLD, on the other hand, shows a moderate
variance and bias.}
% , as seen under protocol \texttt{sample-size-$\mathcal{D}_2$}
% (Section~\ref{sec:sample_size_d2}).
These effects are especially evident at the extremes of the $x$ axis,
which correspond to settings where few instances with either $S=0$ or
$S=1$ are available for quantifier training. In turn, the few
positives (negatives) make it particularly difficult to reliably
estimate $\mathrm{tpr}_{k_{s}}$ ($\mathrm{tnr}_{k_{s}}$), as required
by Equations \ref{eq:acc} and \ref{eq:pacc}. For example, in
Figure~\ref{fig:sample_prev_d20} we see that the error of ACC ranges
between $-1.3$ and $0.7$. Given that the true demographic disparity of
the classifier $h$ is $\delta_{h}^{S}=0.3$, these are the worst
possible errors, corresponding to extreme estimates
$\hat{\delta_{h}^{S}}=-1$ and $\hat{\delta_{h}^{S}}=1$,
respectively. Finally, it is worth noting that PACC outperforms ACC,
thanks to efficient use of posteriors $\pi_s(\mathbf{x}_i)$ in place
of binary decisions $k_s(\mathbf{x}_i)$.

\begin{table}[tb]
  \small \centering
  \caption{Results obtained in the experiments run according to
  protocol \texttt{sample-prev-$\mathcal{D}_2$}.}
  \label{tab:sample_prev_d2}
  % \resizebox{\textwidth}{!}{%
  
\begin{tabular}{llcccc} \toprule
 & & $\downarrow$ MAE & $\downarrow$ MSE & $\uparrow$ $P(\mathrm{AE}<0.1)$ & $\uparrow$ $P(\mathrm{AE}<0.2)$ \\ \midrule
\multirow{7}{*}{Adult} &CC(LR) & 0.230$^{\phantom{\ddag}}\pm0.177$ & 0.084$^{\phantom{\ddag}}\pm0.118$ & 0.274 & 0.523 \\
 &PCC(LR) & 0.213$^{\phantom{\ddag}}\pm0.169$ & 0.074$^{\phantom{\ddag}}\pm0.103$ & 0.323 & 0.551 \\
 &ACC(LR) & 0.159$^{\phantom{\ddag}}\pm0.178$ & 0.057$^{\phantom{\ddag}}\pm0.159$ & 0.439 & 0.789 \\
 &PACC(LR) & 0.112$^{\phantom{\ddag}}\pm0.118$ & 0.026$^{\phantom{\ddag}}\pm0.093$ & 0.559 & 0.889 \\
 &SLD(LR) & \textbf{0.081}$^{\phantom{\ddag}}\pm0.070$ & \textbf{0.011}$^{\phantom{\ddag}}\pm0.020$ & \textbf{0.705} & \textbf{0.929} \\
 &HDy(LR) & 0.219$^{\phantom{\ddag}}\pm0.188$ & 0.084$^{\phantom{\ddag}}\pm0.128$ & 0.345 & 0.573 \\
 &MLPE & 0.295$^{\phantom{\ddag}}\pm0.218$ & 0.134$^{\phantom{\ddag}}\pm0.165$ & 0.239 & 0.410 \\

 \vspace{0.01cm} \\ 
\multirow{7}{*}{COMPAS} &CC(LR) & 0.498$^{\phantom{\ddag}}\pm0.253$ & 0.312$^{\phantom{\ddag}}\pm0.260$ & 0.044 & 0.128 \\
 &PCC(LR) & 0.264$^{\phantom{\ddag}}\pm0.186$ & 0.104$^{\phantom{\ddag}}\pm0.126$ & 0.227 & 0.431 \\
 &ACC(LR) & 0.469$^{\phantom{\ddag}}\pm0.276$ & 0.296$^{\phantom{\ddag}}\pm0.303$ & 0.080 & 0.184 \\
 &PACC(LR) & 0.338$^{\phantom{\ddag}}\pm0.254$ & 0.179$^{\phantom{\ddag}}\pm0.250$ & 0.185 & 0.356 \\
 &SLD(LR) & \textbf{0.160}$^{\phantom{\ddag}}\pm0.123$ & \textbf{0.041}$^{\phantom{\ddag}}\pm0.060$ & \textbf{0.386} & \textbf{0.678} \\
 &HDy(LR) & 0.255$^{\phantom{\ddag}}\pm0.189$ & 0.101$^{\phantom{\ddag}}\pm0.135$ & 0.246 & 0.463 \\
 &MLPE & 0.275$^{\phantom{\ddag}}\pm0.193$ & 0.112$^{\phantom{\ddag}}\pm0.134$ & 0.219 & 0.417 \\

 \vspace{0.01cm} \\ 
\multirow{7}{*}{CreditCard} &CC(LR) & 0.429$^{\phantom{\ddag}}\pm0.252$ & 0.248$^{\phantom{\ddag}}\pm0.236$ & 0.103 & 0.225 \\
 &PCC(LR) & 0.204$^{\phantom{\ddag}}\pm0.140$ & 0.061$^{\phantom{\ddag}}\pm0.073$ & 0.287 & 0.551 \\
 &ACC(LR) & 0.535$^{\phantom{\ddag}}\pm0.316$ & 0.387$^{\phantom{\ddag}}\pm0.353$ & 0.085 & 0.165 \\
 &PACC(LR) & 0.512$^{\phantom{\ddag}}\pm0.311$ & 0.359$^{\phantom{\ddag}}\pm0.343$ & 0.094 & 0.171 \\
 &SLD(LR) & \textbf{0.171}$^{\phantom{\ddag}}\pm0.123$ & \textbf{0.044}$^{\phantom{\ddag}}\pm0.058$ & \textbf{0.348} & \textbf{0.645} \\
 &HDy(LR) & 0.222$^{\phantom{\ddag}}\pm0.159$ & 0.074$^{\phantom{\ddag}}\pm0.101$ & 0.260 & 0.508 \\
 &MLPE & 0.210$^{\phantom{\ddag}}\pm0.143$ & 0.065$^{\phantom{\ddag}}\pm0.077$ & 0.280 & 0.536 \\
\bottomrule
\end{tabular}%

  % }%
\end{table}

These trends also hold for COMPAS and CreditCard, as summarized in
Table~\ref{tab:sample_prev_d2}. Similarly to
Table~\ref{tab:sample_prev_d3}, we find that, under large shifts
between the auxiliary and the test set, the estimation of demographic
disparity is more difficult on COMPAS and CreditCard than on
Adult. Overall, these
experiments show that CC and PCC fare poorly under prior probability
shift, and are outperformed by estimators with better theoretical
guarantees.

% -------------------------------------------------------------------

\subsection{Reduced Cardinality of the Auxiliary
Set: Protocol \texttt{sample-size-$\mathcal{D}_2$}}
\label{sec:sample_size_d2}

% -------------------------------------------------------------------

\subsubsection{Motivation and Setup}

\noindent In this experimental protocol, we focus on the size of the
auxiliary set $\mathcal{D}_2$, studying its influence on the
estimation problem. Our goal is to understand how small this set can
be before degrading the performance of our estimation techniques. We
use subsets $\breve{\mathcal{D}}_2$ of the auxiliary set, obtained by
sampling instances uniformly without replacement from it. We let their
cardinality $|\breve{\mathcal{D}}_2|$ take five values evenly spaced
on a logarithmic scale, between a minimum size
$|\breve{\mathcal{D}}_2|$=1,000 and a maximum size
$|\breve{\mathcal{D}}_2|=|\mathcal{D}_2|$. In other words, we let the
cardinality of the auxiliary set take five different values between
1,000 and $|\mathcal{D}_2|$ in a geometric progression. \replace{As described
in Section~\ref{sec:setup}, for each cardinality of the auxiliary set
we wish to test, we perform ten samplings over five splits and six
permutations, for a total of 300 repetitions per approach per dataset.
Pseudocode~\ref{pseudo:sample-size-d2} describes Protocol
\texttt{sample-size-$\mathcal{D}_2$}, with the pale red region
highlighting the parts that are specific to this protocol.}{}

This protocol is justified by the well-documented difficulties in the
acquisition of demographic data for industry professionals, which vary
depending on the domain, the company and other factors of disparate
nature
\citep{beutel2019putting,holstein2019:if,gladonclavell2020auditing,bogen2021:ap,andrus2021:wm}. As
an example, \citet{gladonclavell2020auditing} perform an internal
fairness audit of a personalized wellness recommendation app, for
which sensitive features are not collected during production,
following the principles of data minimization. \replace{However, sensitive
features were available from an initial development phase in an
auxiliary set, whose size was determined by prior
considerations}{However, sensitive
features were available in a previously obtained
auxiliary set}. Furthermore, in the US, the collection of sensitive
attributes is highly industry dependent, ranging from mandatory to
forbidden, depending on the fragmented regulation applicable in each
domain \citep{bogen2021:ap}. High-quality auxiliary sets can be
obtained through optional surveys \citep{wilson2021building}, for which
response rates are highly dependent on trust, and can be improved by
making the intended use of the data clearer \citep{andrus2021:wm},
directly impacting the cardinality of $\mathcal{D}_2$.

Therefore, the cardinality of the auxiliary set $\mathcal{D}_2$ is an
interesting variable in the context of fairness audits. The estimation
methods that we consider have peculiar data requirements, such as the
need to estimate true/false positive rates. For this reason,
interesting patterns should emerge from this protocol.
% , as hinted by protocol \texttt{sample-prev-$\mathcal{D}_2$}.
We expect key trends for the estimation error to vary monotonically
with $|\breve{\mathcal{D}}_2|$, which is why we let it vary according
to a geometric progression.

% -------------------------------------------------------------------

\subsubsection{Results}

\noindent The signed estimation error on the Adult dataset under this
experimental protocol is illustrated in
Figure~\ref{fig:sample_size_d2}, as we vary the cardinality
$|\breve{\mathcal{D}}_2|$ along the $x$ axis. Clearly, the variance
for each approach decreases as the size of $\breve{\mathcal{D}}_2$
increases. Additionally, slight biases may improve,
% \fabsebcomment{Cosa vuol dire?} \afabcomment{aggiunto spiegazione da
% ``whose''}
as is the case with HDy, whose median error approaches zero as
$|\breve{\mathcal{D}}_2|$ increases. These trends are a direct
confirmation of hints already obtained from the protocols discussed
above. The most striking trend is the unreliability of ACC and PACC
(and especially the former) in the small data regime.
% \alex{this should hold true also for HDy, which similarly requires a
% validation split}. Indeed, these are adjusted methods, which require
% splitting the auxiliary set $\mathcal{D}_2$ into a first subset used
% for the actual training of a classifier and a second subset used to
% compute its $\mathrm{tpr}$ and $\mathrm{fpr}$ for the adjustment of
% prevalence estimates, as per Equations \ref{eq:acc} and
% \ref{eq:pacc}. \footnote{Once the estimations for $\mathrm{tpr}$ and
% $\mathrm{fpr}$ have been computed, it would be possible to
% \emph{retrain} the classifier, using the whole training set (as is
% customarily done in other areas of machine learning), thus
% generating a more robust classifier. In this particular case,
% however, retraining might cause the adjustment of the quantification
% to become less stable, since the main ingredients of the adjustemnt
% (the $\mathrm{tpr}$ and $\mathrm{fpr}$) actually characterise a
% different classifier. The extent to which this trade-off (more data
% vs. unstable adjustment) ends up impacting negatively or positively
% the final computation of the prevalence values is out of scope for
% this study; we thus avoid retraining in favour of having congruent
% estimates for the rates.}

\begin{figure}[tb]
  \centering
  \includegraphics[width=\mysize\textwidth]{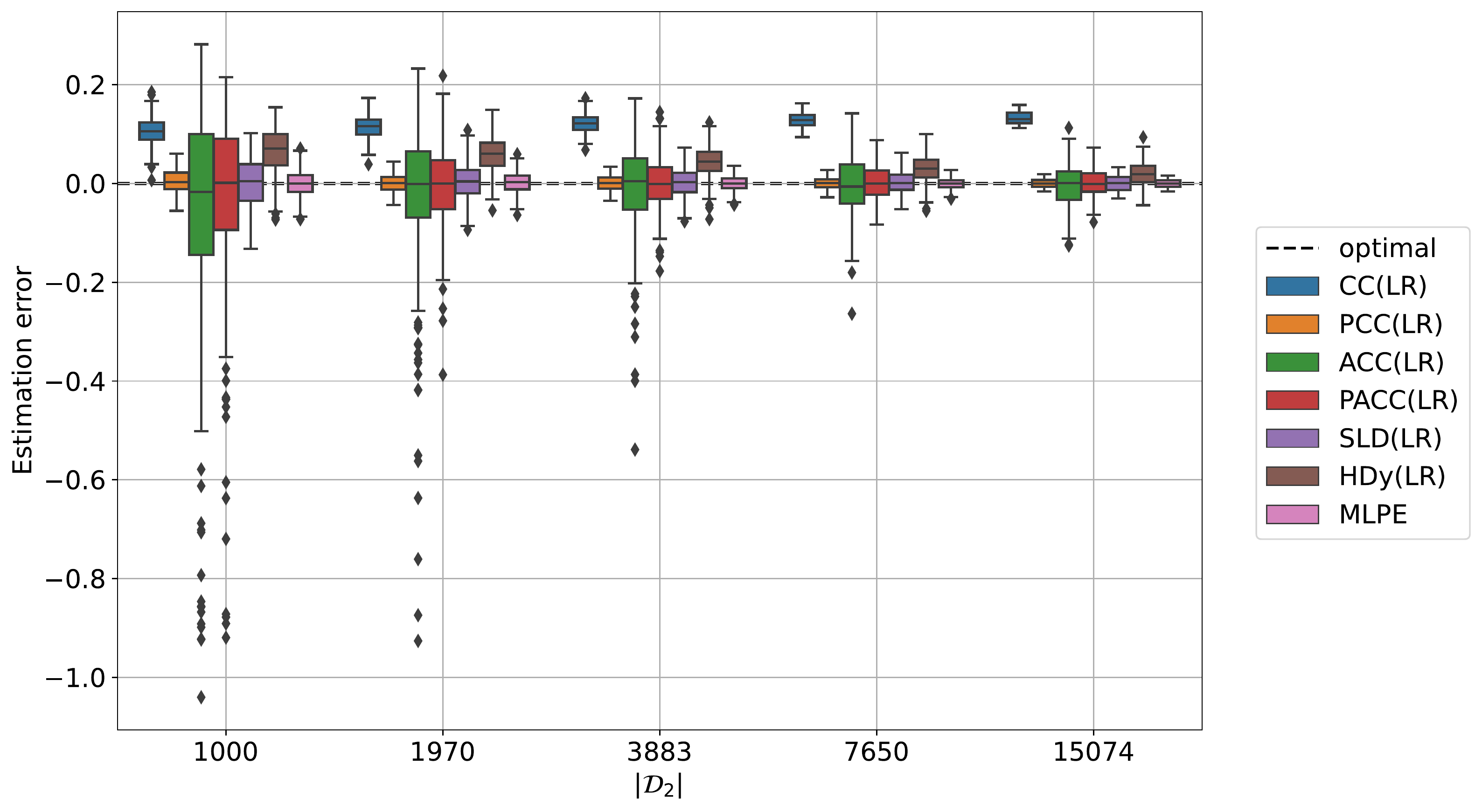}
  \caption{Protocol \texttt{sample-size-$\mathcal{D}_2$} on the Adult
  dataset. Distribution of the estimation error ($y$ axis) as the
  cardinality $|\breve{\mathcal{D}}_2|$ is varied ($x$ axis).}
  \label{fig:sample_size_d2}
\end{figure}

Similar results are obtained for COMPAS and CreditCard, as reported in
Table \ref{tab:sample_size_d2}. Across the three datasets, PACC and
ACC perform quite poorly due to the difficulty in estimating
$\mathrm{tpr}_{k_{s}}$ and $\mathrm{fpr}_{k_{s}}$ with the few
labelled data points available from $\breve{\mathcal{D}}_2$. On the
other hand, both SLD and HDy are fairly reliable. \replace{PCC stands out as a
strong performer}{PCC and MLPE stand out as
strong performers}, with low bias and low variance. This is due to the
fact that, under this experimental protocol, there is no shift between
the auxiliary set $\mathcal{D}_2$, on which the quantifiers are
trained, and the test set $\mathcal{D}_3$, on which they are
tested. Since the current protocol focuses on the cardinality of the
auxiliary set, $\mathcal{D}_2$ and $\mathcal{D}_3$ remain stratified
subsets of the Adult dataset, with identical distributions over
$(S,Y)$. In turn, this favours \blue{MLPE, which assumes no shift between $\mathcal{D}_2$ and $\mathcal{D}_3$, and } PCC, which relies on the fact that the
posterior probabilities of its underlying classifier $k$ are
well-calibrated on $\mathcal{D}_3$.\footnote{Posterior probabilities
$\Pr(s|\mathbf{x})$ are said to be \emph{well-calibrated} when, given
a sample $\sigma$ drawn from $\mathcal{X}$
$$\lim_{|\sigma|\rightarrow \infty}\frac{|\{\mathbf{x}\in s|
\Pr(s|\mathbf{x})=\alpha\}|}{|\{\mathbf{x}\in \sigma|
\Pr(s|\mathbf{x})=\alpha\}|}=\alpha.$$ i.e., when for big enough
samples, $\alpha$ approximates the true proportion of data points
belonging to class $s$ among all data points for which
$\Pr(s|\mathbf{x})=\alpha$.}

\begin{table}[tb]
  \small \centering
  \caption{Results obtained in the experiments run according to
  protocol \texttt{sample-size-$\mathcal{D}_2$}}
  \label{tab:sample_size_d2}
  % \resizebox{\textwidth}{!}{%
  
\begin{tabular}{llcccc} \toprule
 & & $\downarrow$ MAE & $\downarrow$ MSE & $\uparrow$ $P(\mathrm{AE}<0.1)$ & $\uparrow$ $P(\mathrm{AE}<0.2)$ \\ \midrule
\multirow{7}{*}{Adult} &CC(LR) & 0.120$^{\phantom{\ddag}}\pm0.022$ & 0.015$^{\phantom{\ddag}}\pm0.005$ & 0.159 & \textbf{1.000} \\
 &PCC(LR) & \textbf{0.012}$^{\phantom{\ddag}}\pm0.010$ & \textbf{0.000}$^{\phantom{\ddag}}\pm0.000$ & \textbf{1.000} & \textbf{1.000} \\
 &ACC(LR) & 0.083$^{\phantom{\ddag}}\pm0.113$ & 0.020$^{\phantom{\ddag}}\pm0.082$ & 0.747 & 0.928 \\
 &PACC(LR) & 0.055$^{\phantom{\ddag}}\pm0.079$ & 0.009$^{\phantom{\ddag}}\pm0.048$ & 0.856 & 0.969 \\
 &SLD(LR) & 0.025$^{\phantom{\ddag}}\pm0.020$ & 0.001$^{\phantom{\ddag}}\pm0.002$ & 0.996 & \textbf{1.000} \\
 &HDy(LR) & 0.047$^{\phantom{\ddag}}\pm0.033$ & 0.003$^{\phantom{\ddag}}\pm0.004$ & 0.922 & \textbf{1.000} \\
 &MLPE & 0.013$^{\ddag}\pm0.012$ & 0.000$^{\dag}\pm0.001$ & \textbf{1.000} & \textbf{1.000} \\

 \vspace{0.01cm} \\ 
\multirow{7}{*}{COMPAS} &CC(LR) & 0.353$^{\phantom{\ddag}}\pm0.047$ & 0.127$^{\phantom{\ddag}}\pm0.032$ & 0.000 & 0.005 \\
 &PCC(LR) & 0.030$^{\ddag}\pm0.020$ & \textbf{0.001}$^{\phantom{\ddag}}\pm0.001$ & 0.999 & \textbf{1.000} \\
 &ACC(LR) & 0.381$^{\phantom{\ddag}}\pm0.213$ & 0.190$^{\phantom{\ddag}}\pm0.214$ & 0.097 & 0.186 \\
 &PACC(LR) & 0.265$^{\phantom{\ddag}}\pm0.212$ & 0.115$^{\phantom{\ddag}}\pm0.183$ & 0.247 & 0.467 \\
 &SLD(LR) & 0.135$^{\phantom{\ddag}}\pm0.098$ & 0.028$^{\phantom{\ddag}}\pm0.038$ & 0.441 & 0.765 \\
 &HDy(LR) & 0.108$^{\phantom{\ddag}}\pm0.082$ & 0.018$^{\phantom{\ddag}}\pm0.027$ & 0.549 & 0.858 \\
 &MLPE & \textbf{0.029}$^{\phantom{\ddag}}\pm0.021$ & 0.001$^{\ddag}\pm0.002$ & \textbf{0.999} & \textbf{1.000} \\

 \vspace{0.01cm} \\ 
\multirow{7}{*}{CreditCard} &CC(LR) & 0.177$^{\phantom{\ddag}}\pm0.078$ & 0.037$^{\phantom{\ddag}}\pm0.030$ & 0.177 & 0.629 \\
 &PCC(LR) & 0.016$^{\ddag}\pm0.013$ & 0.000$^{\ddag}\pm0.001$ & \textbf{1.000} & \textbf{1.000} \\
 &ACC(LR) & 0.337$^{\phantom{\ddag}}\pm0.266$ & 0.184$^{\phantom{\ddag}}\pm0.259$ & 0.203 & 0.368 \\
 &PACC(LR) & 0.299$^{\phantom{\ddag}}\pm0.255$ & 0.154$^{\phantom{\ddag}}\pm0.240$ & 0.261 & 0.445 \\
 &SLD(LR) & 0.053$^{\phantom{\ddag}}\pm0.043$ & 0.005$^{\phantom{\ddag}}\pm0.008$ & 0.871 & 0.985 \\
 &HDy(LR) & 0.057$^{\phantom{\ddag}}\pm0.046$ & 0.005$^{\phantom{\ddag}}\pm0.009$ & 0.831 & 0.991 \\
 &MLPE & \textbf{0.016}$^{\phantom{\ddag}}\pm0.013$ & \textbf{0.000}$^{\phantom{\ddag}}\pm0.001$ & \textbf{1.000} & \textbf{1.000} \\
\bottomrule
\end{tabular}%

  % }%
\end{table}

% ------------------------------------------------------------------

\subsection{Distribution Shift Affecting $\mathcal{D}_1$: Protocol
\texttt{sample-prev-$\mathcal{D}_1$}}
\label{sec:sample_prev_d1}

% ------------------------------------------------------------------

\subsubsection{Motivation and Setup}
\label{sec:sample_prev_d1:setup}

\noindent With this protocol we evaluate the impact of shifts in the
training set $\mathcal{D}_1$, by drawing different subsets
$\breve{\mathcal{D}}_1$ as we vary $\Pr(Y=S)$.\footnote{While $Y$ and
$S$ take values from different domains, by $Y=S$ we mean
$(Y=\oplus \wedge S=1) \vee (Y=\ominus \wedge S=0)$, i.e. a situation
where positive outcomes are associated with group $S=1$ and negative
outcomes with group $S=0$.} More specifically, we vary $\Pr(Y=S)$
between 0 and 1 with a step of 0.1. In other words, we sample at
random from $\mathcal{D}_1$ a proportion $p$ of instances
$(\mathbf{x}_i,s_i,y_i)$ such that $Y=S$ and a proportion $(1-p)$ such
that $Y\neq S$, with
$p \in \left \lbrace 0.0, 0.1, \dots, 0.9, 1.0 \right \rbrace$.
% It is worth noting that we defined $\mathcal{D}_1$, in Section
% \ref{sec:notation}, as a training set involving
% $(\mathcal{X},\mathcal{Y})$. Here we exploit our knowledge of $S$ to
% control the dataset shift between training and test conditions,
% emulating a biased data collection procedure. Once a training set
% has been selected, the classifier $h$ is learnt exclusively from
% non-sensitive attributes $\mathcal{X}$, completely disregarding the
% sensitive attribute $S$.
We choose a limited cardinality $|\breve{\mathcal{D}}_1|=500$, which
allows us to perform multiple repetitions at reasonable computational
costs, as described in Section~\ref{sec:setup}. Although this may
impact the quality of the classifier $h$, this aspect is not the
central focus of the present work.

This experimental protocol aligns with biased data collection
procedures, sometimes referred to as \emph{censored data}
\citep{kallus2018residual}. Indeed, it is common for the ground-truth
variable to represent a mere proxy for the actual quantity of
interest, with nontrivial sampling effects between the two. For
example, the validity of arrest data as a proxy for offence has been
brought into question \citep{Fogliato:2021mb}.
% \afabcomment{Cite this work when proceedings become available
% \url{https://arxiv.org/pdf/2105.04953.pdf}}.
In fact, in this domain, different sources of sampling bias can be in
action, such as uneven allocation of police resources between
jurisdictions and neighbourhoods \citep{malcolm2008:minority} and lower
levels of cooperation in populations who feel oppressed by law
enforcement \citep{xie2012racial}.

By varying $\Pr(Y=S)$ we impose a spurious correlation between $Y$ and
$S$, which may be picked up by the classifier $h$. In extreme
situations, such as when $\Pr(Y=S) \simeq 1$, a classifier $h$ can
confound the concepts behind $S$ and $Y$. In turn, we expect this to
unevenly affect the acceptance rates for the two demographic groups,
effectively changing the demographic disparity of $h$, i.e., our
estimand $\delta_{h}^{S}$. Pseudocode~\ref{pseudo:sample-prev-d1}
describes the main steps to implement Protocol
\texttt{sample-prev-$\mathcal{D}_1$}.

% Datasets where $Y \perp S | X$ are an exception in this regard. If,
% given the observable attributes, the outcome is independent of the
% sensitive attribute, we expect that classifiers will not learn
% spurious correlations between $Y$ and $S$ that are due to sample
% selection.

% -------------------------------------------------------------------

\subsubsection{Results}

\begin{figure}[tb]
  \centering
  \includegraphics[width=\mysize\textwidth]{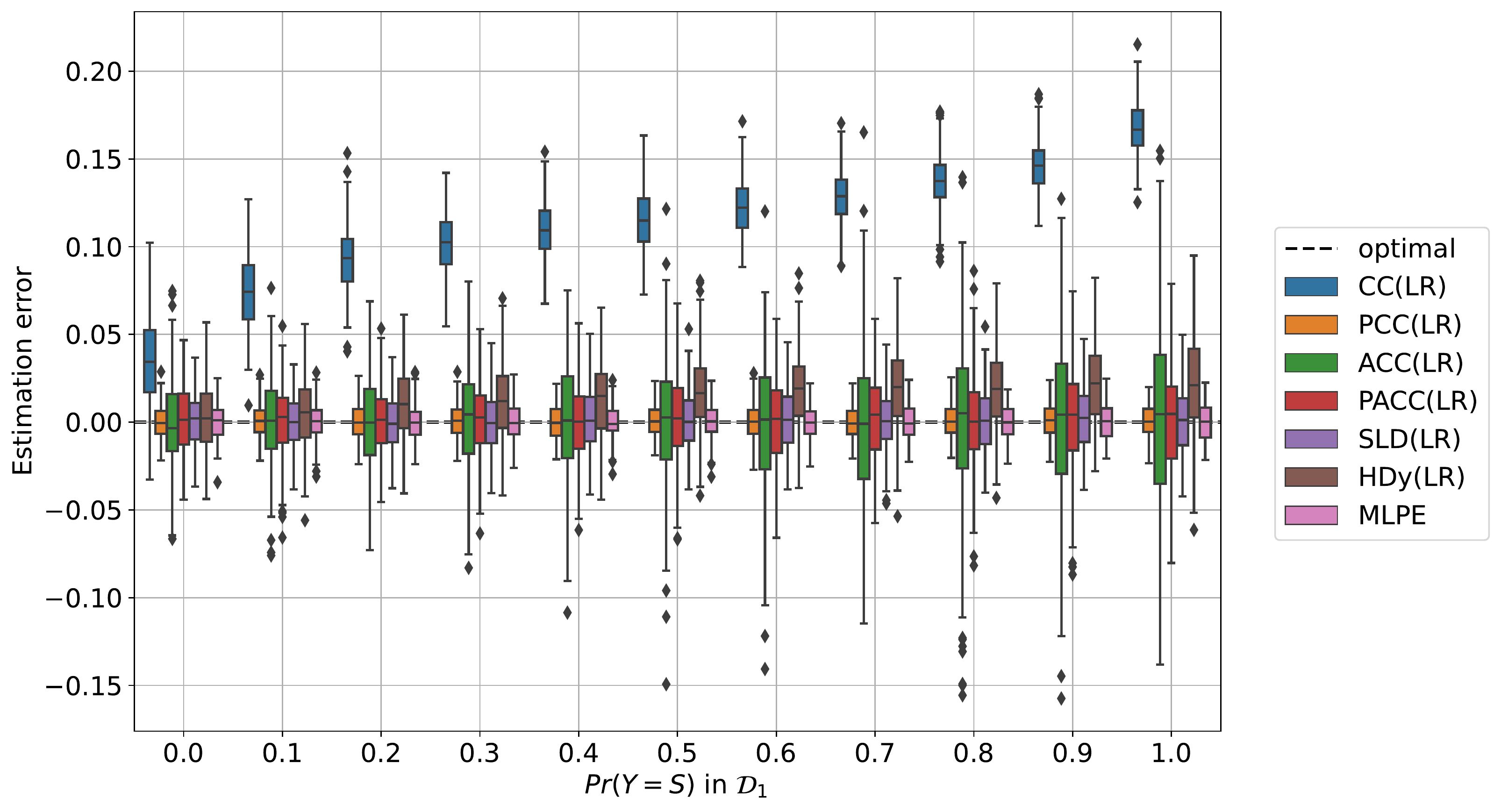}
  \caption{Protocol \texttt{sample-prev-$\mathcal{D}_1$} on the Adult
  dataset. Distribution of the estimation error ($y$ axis) as
  $\breve{\mathcal{D}}_1$ is sampled with a given $\Pr(Y=S)$ ($x$
  axis). Each boxplot summarizes the results of 5 random splits, 6
  role permutations and 10 samplings of $\breve{\mathcal{D}}_1$.}
  \label{fig:sample_prev_d1}
\end{figure}

\noindent In Figure~\ref{fig:sample_prev_d1}, the $y$ axis depicts the
estimation error (Equation~\ref{eq:estim_err1}), as we vary $\Pr(Y=S)$
along the $x$ axis. Each quantification approach outperforms vanilla
CC, which overestimates the demographic disparity of the classifier
$h$, i.e., its estimate is larger than the ground truth value, so
$\hat{\delta}_{h}^{S, \text{CC}} > \delta_{h}^{S}$. ACC, PCC, PACC, SLD,
\replace{and HDy}{HDy, and MLPE} display a negligible bias and a reliable estimate of
demographic disparity. The absolute error for these techniques is
% bounded close to zero so that the Mean Absolute Error is
always below 0.1, except for a few outliers.

\begin{table}[tb]
  \small \centering
  \caption{Results obtained in the experiments run according to
  protocol \texttt{sample-prev-$\mathcal{D}_1$}.}
  \label{tab:sample_prev_d1}
  % \resizebox{\textwidth}{!}{%
  
\begin{tabular}{llcccc} \toprule
 & & $\downarrow$ MAE & $\downarrow$ MSE & $\uparrow$ $P(\mathrm{AE}<0.1)$ & $\uparrow$ $P(\mathrm{AE}<0.2)$ \\ \midrule
\multirow{7}{*}{Adult} &CC(LR) & 0.112$^{\phantom{\ddag}}\pm0.038$ & 0.014$^{\phantom{\ddag}}\pm0.008$ & 0.321 & 0.998 \\
 &PCC(LR) & \textbf{0.008}$^{\phantom{\ddag}}\pm0.005$ & \textbf{0.000}$^{\phantom{\ddag}}\pm0.000$ & \textbf{1.000} & \textbf{1.000} \\
 &ACC(LR) & 0.029$^{\phantom{\ddag}}\pm0.024$ & 0.001$^{\phantom{\ddag}}\pm0.003$ & 0.983 & \textbf{1.000} \\
 &PACC(LR) & 0.019$^{\phantom{\ddag}}\pm0.014$ & 0.001$^{\phantom{\ddag}}\pm0.001$ & \textbf{1.000} & \textbf{1.000} \\
 &SLD(LR) & 0.013$^{\phantom{\ddag}}\pm0.010$ & 0.000$^{\phantom{\ddag}}\pm0.000$ & \textbf{1.000} & \textbf{1.000} \\
 &HDy(LR) & 0.022$^{\phantom{\ddag}}\pm0.016$ & 0.001$^{\phantom{\ddag}}\pm0.001$ & \textbf{1.000} & \textbf{1.000} \\
 &MLPE & 0.008$^{\phantom{\ddag}}\pm0.006$ & 0.000$^{\phantom{\ddag}}\pm0.000$ & \textbf{1.000} & \textbf{1.000} \\

 \vspace{0.01cm} \\ 
\multirow{7}{*}{COMPAS} &CC(LR) & 0.328$^{\phantom{\ddag}}\pm0.091$ & 0.116$^{\phantom{\ddag}}\pm0.056$ & 0.022 & 0.081 \\
 &PCC(LR) & \textbf{0.026}$^{\phantom{\ddag}}\pm0.019$ & \textbf{0.001}$^{\phantom{\ddag}}\pm0.001$ & 1.000 & \textbf{1.000} \\
 &ACC(LR) & 0.349$^{\phantom{\ddag}}\pm0.211$ & 0.166$^{\phantom{\ddag}}\pm0.192$ & 0.130 & 0.252 \\
 &PACC(LR) & 0.194$^{\phantom{\ddag}}\pm0.164$ & 0.065$^{\phantom{\ddag}}\pm0.115$ & 0.345 & 0.607 \\
 &SLD(LR) & 0.114$^{\phantom{\ddag}}\pm0.083$ & 0.020$^{\phantom{\ddag}}\pm0.027$ & 0.512 & 0.849 \\
 &HDy(LR) & 0.096$^{\phantom{\ddag}}\pm0.076$ & 0.015$^{\phantom{\ddag}}\pm0.023$ & 0.605 & 0.897 \\
 &MLPE & 0.027$^{\ddag}\pm0.019$ & 0.001$^{\ddag}\pm0.001$ & \textbf{1.000} & \textbf{1.000} \\

 \vspace{0.01cm} \\ 
\multirow{7}{*}{CreditCard} &CC(LR) & 0.152$^{\phantom{\ddag}}\pm0.095$ & 0.032$^{\phantom{\ddag}}\pm0.036$ & 0.338 & 0.711 \\
 &PCC(LR) & \textbf{0.010}$^{\phantom{\ddag}}\pm0.007$ & \textbf{0.000}$^{\phantom{\ddag}}\pm0.000$ & \textbf{1.000} & \textbf{1.000} \\
 &ACC(LR) & 0.187$^{\phantom{\ddag}}\pm0.152$ & 0.058$^{\phantom{\ddag}}\pm0.094$ & 0.347 & 0.626 \\
 &PACC(LR) & 0.130$^{\phantom{\ddag}}\pm0.106$ & 0.028$^{\phantom{\ddag}}\pm0.046$ & 0.487 & 0.777 \\
 &SLD(LR) & 0.047$^{\phantom{\ddag}}\pm0.037$ & 0.004$^{\phantom{\ddag}}\pm0.005$ & 0.902 & 0.998 \\
 &HDy(LR) & 0.061$^{\phantom{\ddag}}\pm0.047$ & 0.006$^{\phantom{\ddag}}\pm0.009$ & 0.814 & 0.989 \\
 &MLPE & 0.011$^{\phantom{\ddag}}\pm0.008$ & 0.000$^{\phantom{\ddag}}\pm0.000$ & \textbf{1.000} & \textbf{1.000} \\
\bottomrule
\end{tabular}%

  % }%
\end{table}

Results for the COMPAS and CreditCard datasets are reported in Table
\ref{tab:sample_prev_d1}. Confirming the results of previous
protocols, these datasets provide a harder setting for the estimate of
demographic disparity, as shown by higher MAE and MSE, which, for
instance, increase by one order of magnitude for SLD and PACC moving
from Adult to COMPAS. PCC is the best performer, for the same reasons
discussed in Section~\ref{sec:sample_prev_d3}, i.e., the absence of
shift between $\mathcal{D}_2$ and $\mathcal{D}_3$.

\subsection{Distribution Shift Affecting $\mathcal{D}_1$: Protocol
\texttt{flip-prev-$\mathcal{D}_1$}}
\label{sec:flip_prev_d1}

% -------------------------------------------------------------------

\subsubsection{Motivation and Setup}

\noindent %We just noted that changing $\Pr(Y=S)$ via sampling in the training set $\mathcal{D}_1$ may leave the classifier $h$ and its demographic parity unaffected.

\noindent Certain biases in the training set resulting from
domain-specific practices, such as the use of arrest as a substitute
for the offence, can be modelled as either a selection bias
\citep{Fogliato:2021mb} or a label bias distorting the ground truth
variable $Y$ \citep{fogliato2020:fairness}. With this experimental
protocol, we impose the latter bias by actively flipping some ground
truth labels $Y$ in $\mathcal{D}_1$ based on their sensitive
attribute. Similarly to \texttt{sample-prev-$\mathcal{D}_1$}, this
protocol achieves a given association
% $\Pr(Y=S)$
between the target $Y$ and the sensitive variable $S$ in the training
set $\mathcal{D}_1$. However, instead of sampling, it does so by
flipping the $Y$ label of some data points. More specifically, we
impose
% $\Pr(Y=S|S=0)=\Pr(Y=S|S=1)
$\Pr(Y=\ominus|S=0)=\Pr(Y=\oplus|S=1) = p$ and let
$p$ take values across eleven evenly spaced values between 0 and
1. For every value of
$p$, we first sample a random subset
$\breve{\mathcal{D}}_1$ of the training set with cardinality
500. Next, we actively flip some
$Y$ labels in both demographic groups, until both
$\Pr(Y=\ominus|S=0)$ and
$\Pr(Y=\oplus|S=1)$ reach the desired value of $p \in \left \lbrace
  0.0, 0.1, \dots, 0.9, 1.0 \right
\rbrace$. Finally, we train a classifier $h$ on the attributes
$X$ and modified ground truth $Y$ of $\breve{\mathcal{D}}_1$.

This experimental protocol is compatible with settings where the
training data capture a distorted ground truth due to systematic
biases and group-dependent annotation accuracy
\citep{wang2021:fair}. As an example, the quality of medical diagnoses
can depend on race, sex, and socioeconomic status
\citep{gianfrancesco2018potential}. In addition, health care
expenditures have been used as a proxy to train an algorithm deployed
nationwide in the US to estimate patients' health care needs,
resulting in a systematic underestimation of the needs of
African-American patients \citep{obermeyer2019dissecting}. In the
hiring domain, employer response rates to resumes have been found to
vary with the perceived ethnic origin of an applicant’s name
\citep{bertrand2004emily}.
% Finally, the gender gap in mathematical performance, while
% negligible in elementary school, has been found to increase with age
% \citep{lindberg2010new}, possibly due to gender stereotypes arising
% in this domain from an early age \citep{cvencek2011:mg} and to the
% prescriptive nature of these stereotypes
% \citep{nosek2002:mm,ellemers2018:gs}.
These are all examples where the ``ground truth'' associated with a
dataset is distorted to the disadvantage of a sensitive demographic
group.

Similarly to Section~\ref{sec:sample_prev_d1}, we expect that this
experimental protocol will cause significant variations in the
demographic disparity of the classifier $h$ due to the strong
correlation we impose between $S$ and $Y$ by label flipping. The
pseudocode that describes this protocol is essentially the same as in
Pseudocode~\ref{pseudo:sample-prev-d1}, simply replacing the sampling
in line~\ref{line:sample:p1} with the label flipping procedure
described above; therefore, we omit it.

% Indeed, if the condition $Y \perp S | X$ held true on the original
% dataset $\mathcal{D}_1$, after label flipping we are guaranteed to
% disrupt the conditional independence of ground truth and sensitive
% attribute.

% -------------------------------------------------------------------

\subsubsection{Results}

\begin{figure}[tb]
  \centering
  \includegraphics[width=\mysize\textwidth]{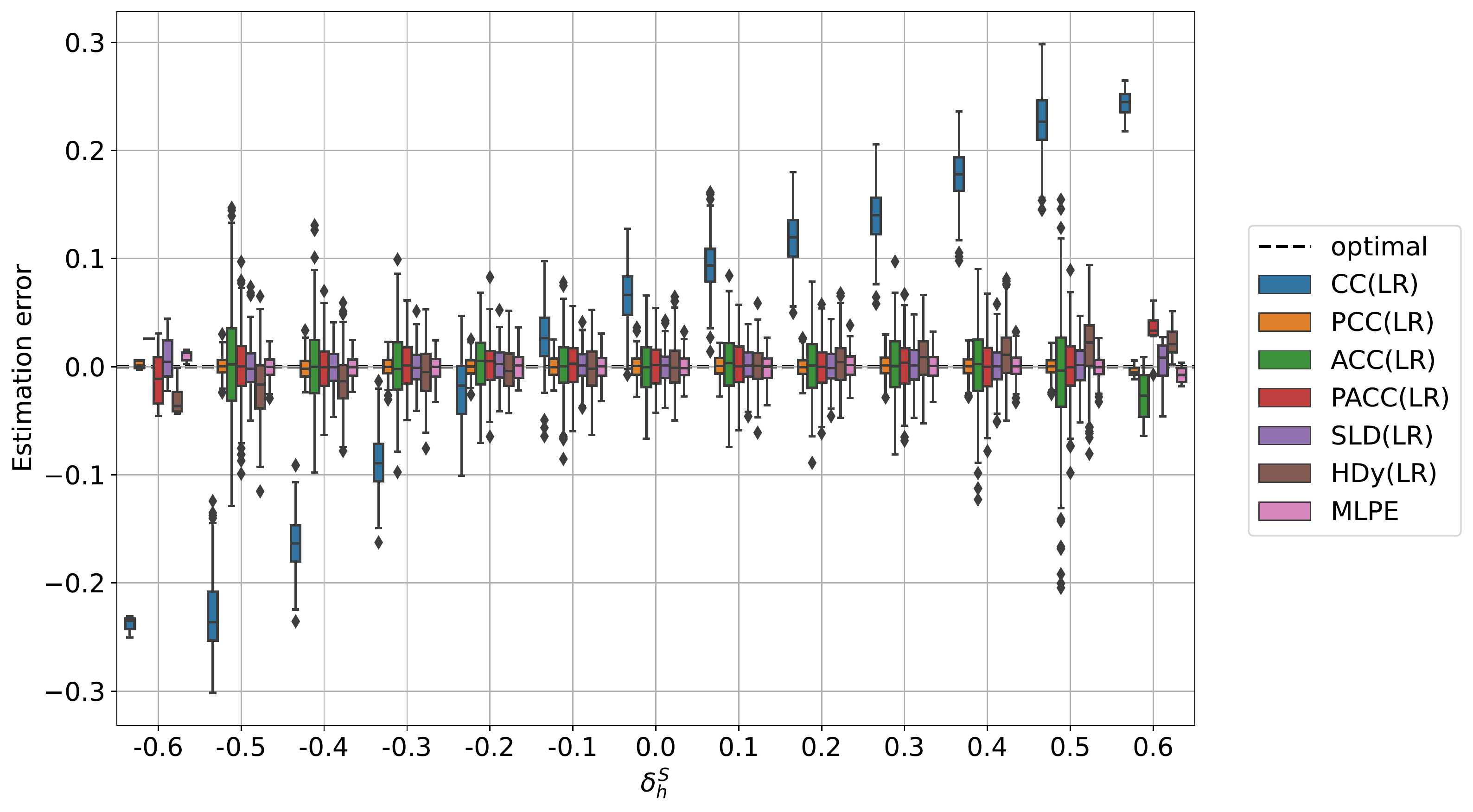}
  \caption{Protocol \texttt{flip-prev-$\mathcal{D}_1$} on the Adult
  dataset. Distribution of the estimation error ($y$ axis) as
  $\delta_{h}^{S}$ varies ($x$ axis).}
  \label{fig:flip_prev_d1}
\end{figure}

\noindent Figure~\ref{fig:flip_prev_d1} illustrates the key trends
caused by this experimental protocol on the Adult dataset. A clear
trend is visible along the $x$ axis, which reports the true
demographic disparity $\delta_{h}^{S}$ for the classifier $h$
(Equation~\ref{eq:dd}), quantized with a step of 0.1. We choose to
depict the true demographic disparity on the $x$ axis as it is the
estimand, hence a quantity of interest by definition. The error
incurred by CC displays a linear trend that goes from severe
underestimation (for low values of the $x$ axis) to severe
overestimation (for large values of the $x$ axis).
% CC displays a tendency to over-estimation (or under-estimation
% \alex{forse potremmo dire che passa dall'over-estimation
% all'under-estimation}) with a clear dependence on the true value of
% the estimand.
In other words, the (signed) estimation error increases with the true
demographic disparity of the classifier $h$, a phenomenon also noticed
by \citet{chen2019:fu}. All remaining approaches compensate for this
weakness and display a good estimation error: PCC, ACC, PACC, SLD, \replace{and HDy}{HDy, and MLPE} have low variance and a median estimation close to zero across
different values of the estimand. Table~\ref{tab:flip_prev_d1}
summarizes similar results on COMPASS and CreditCard; PCC remains
well-calibrated and very effective, while SLD and HDy also have good
performance.

\begin{table}[tb]
  \small \centering
  \caption{Results obtained in the experiments run according to
  protocol \texttt{flip-prev-$\mathcal{D}_1$}.}
  \label{tab:flip_prev_d1}
  % \resizebox{\textwidth}{!}{%
  
\begin{tabular}{llcccc} \toprule
 & & $\downarrow$ MAE & $\downarrow$ MSE & $\uparrow$ $P(\mathrm{AE}<0.1)$ & $\uparrow$ $P(\mathrm{AE}<0.2)$ \\ \midrule
\multirow{7}{*}{Adult} &CC(LR) & 0.151$^{\phantom{\ddag}}\pm0.072$ & 0.028$^{\phantom{\ddag}}\pm0.021$ & 0.274 & 0.706 \\
 &PCC(LR) & \textbf{0.008}$^{\phantom{\ddag}}\pm0.006$ & \textbf{0.000}$^{\phantom{\ddag}}\pm0.000$ & \textbf{1.000} & \textbf{1.000} \\
 &ACC(LR) & 0.030$^{\phantom{\ddag}}\pm0.025$ & 0.002$^{\phantom{\ddag}}\pm0.003$ & 0.982 & 0.999 \\
 &PACC(LR) & 0.020$^{\phantom{\ddag}}\pm0.015$ & 0.001$^{\phantom{\ddag}}\pm0.001$ & \textbf{1.000} & \textbf{1.000} \\
 &SLD(LR) & 0.014$^{\phantom{\ddag}}\pm0.011$ & 0.000$^{\phantom{\ddag}}\pm0.000$ & \textbf{1.000} & \textbf{1.000} \\
 &HDy(LR) & 0.022$^{\phantom{\ddag}}\pm0.017$ & 0.001$^{\phantom{\ddag}}\pm0.001$ & 1.000 & \textbf{1.000} \\
 &MLPE & 0.009$^{\phantom{\ddag}}\pm0.006$ & 0.000$^{\phantom{\ddag}}\pm0.000$ & \textbf{1.000} & \textbf{1.000} \\

 \vspace{0.01cm} \\ 
\multirow{7}{*}{COMPAS} &CC(LR) & 0.388$^{\phantom{\ddag}}\pm0.116$ & 0.164$^{\phantom{\ddag}}\pm0.083$ & 0.027 & 0.068 \\
 &PCC(LR) & \textbf{0.027}$^{\phantom{\ddag}}\pm0.020$ & \textbf{0.001}$^{\phantom{\ddag}}\pm0.001$ & 0.998 & \textbf{1.000} \\
 &ACC(LR) & 0.392$^{\phantom{\ddag}}\pm0.211$ & 0.198$^{\phantom{\ddag}}\pm0.199$ & 0.105 & 0.194 \\
 &PACC(LR) & 0.195$^{\phantom{\ddag}}\pm0.160$ & 0.063$^{\phantom{\ddag}}\pm0.106$ & 0.337 & 0.611 \\
 &SLD(LR) & 0.115$^{\phantom{\ddag}}\pm0.084$ & 0.020$^{\phantom{\ddag}}\pm0.027$ & 0.513 & 0.836 \\
 &HDy(LR) & 0.094$^{\phantom{\ddag}}\pm0.075$ & 0.015$^{\phantom{\ddag}}\pm0.023$ & 0.612 & 0.906 \\
 &MLPE & 0.028$^{\ddag}\pm0.019$ & 0.001$^{\ddag}\pm0.001$ & \textbf{0.999} & \textbf{1.000} \\

 \vspace{0.01cm} \\ 
\multirow{7}{*}{CreditCard} &CC(LR) & 0.159$^{\phantom{\ddag}}\pm0.101$ & 0.036$^{\phantom{\ddag}}\pm0.037$ & 0.345 & 0.640 \\
 &PCC(LR) & \textbf{0.011}$^{\phantom{\ddag}}\pm0.009$ & \textbf{0.000}$^{\phantom{\ddag}}\pm0.000$ & \textbf{1.000} & \textbf{1.000} \\
 &ACC(LR) & 0.223$^{\phantom{\ddag}}\pm0.185$ & 0.084$^{\phantom{\ddag}}\pm0.130$ & 0.307 & 0.565 \\
 &PACC(LR) & 0.147$^{\phantom{\ddag}}\pm0.117$ & 0.035$^{\phantom{\ddag}}\pm0.056$ & 0.420 & 0.725 \\
 &SLD(LR) & 0.056$^{\phantom{\ddag}}\pm0.043$ & 0.005$^{\phantom{\ddag}}\pm0.007$ & 0.843 & 0.995 \\
 &HDy(LR) & 0.071$^{\phantom{\ddag}}\pm0.055$ & 0.008$^{\phantom{\ddag}}\pm0.012$ & 0.732 & 0.973 \\
 &MLPE & 0.012$^{\dag}\pm0.009$ & 0.000$^{\dag}\pm0.000$ & \textbf{1.000} & \textbf{1.000} \\
\bottomrule
\end{tabular}%

  % }%
\end{table}

% PCC stands out as particularly effective, which, just like for
% protocol \texttt{sample-prev-$\mathcal{D}_1$}, can be explained with
% a complete absence of shift between $\mathcal{D}_2$ and
% $\mathcal{D}_3$, which makes the posteriors of the underlying
% classifier $k$ reliable during inference.

% As summarised in Table~\ref{tab:flip_prev_d1}, similar trends obtain
% for the COMPAS dataset, albeit with worse average performance,
% already noticed in Section~\ref{sec:sample_prev_d1}. A notable
% result on CreditCard is the solid performance of
% Classify-and-Count. We notice that, on CreditCard, the entire range
% of biased training sets $\breve{\mathcal{D}}_1$ imposed by protocol
% \texttt{flip-prev-$\mathcal{D}_1$} brings about low variability in
% the true demographic disparity of classifier $h$ on $\mathcal{D}_3$,
% such that the maximum and minimum values of $\delta_{h}^{S}$ across
% the entire protocol differ by 0.2. This variability is one order of
% magnitude lower than the one measured on the Adult dataset, where
% the maximum and minimum values taken by $\delta_{h}^{S}$ differ by
% 1, as per Figure~\ref{fig:flip_prev_d1}, making estimation via
% Classify-and-Count more challenging on Adult rather than on
% CreditCard.

% -------------------------------------------------------------------
\blue{
\subsection{Estimating Fairness for Discrimination-Aware Classifiers} \label{sec:fair_classifier}
\subsubsection{Motivation and Setup}
So far, we have considered classifiers $h(\mathbf{x})$ which only maximize accuracy. In practice, it can be especially interesting to monitor fairness for methods that target this quantity, explicitly optimizing fairness during training. In fact, sensitive attributes may be available during training, allowing for a direct optimization of equity, but unavailable after deployment, complicating fairness evaluation of live systems. In this section, we replace the vanilla LR classifier from the previous experiments with a fairness-aware method. We train a decision tree $h_{\mathrm{T}}$, jointly optimizing accuracy and demographic parity, with the cost-sensitive method of \citet{agarwal2018reductions}. This method makes use of $s$ during training to adjust the cost of positive and negative predictions according to group membership. This learning scheme leads to a classifier $h_{\mathrm{T}}(\mathbf{x})$ which is fairness-aware but does not require access to sensitive attributes to issue predictions on $\mathcal{D}_3$.
\subsubsection{Results}
}

\begin{figure}[t]%[h!]
  \centering
  \begin{subfigure}{\mysize\textwidth}
    \centering
    \includegraphics[width=\mysize\textwidth]{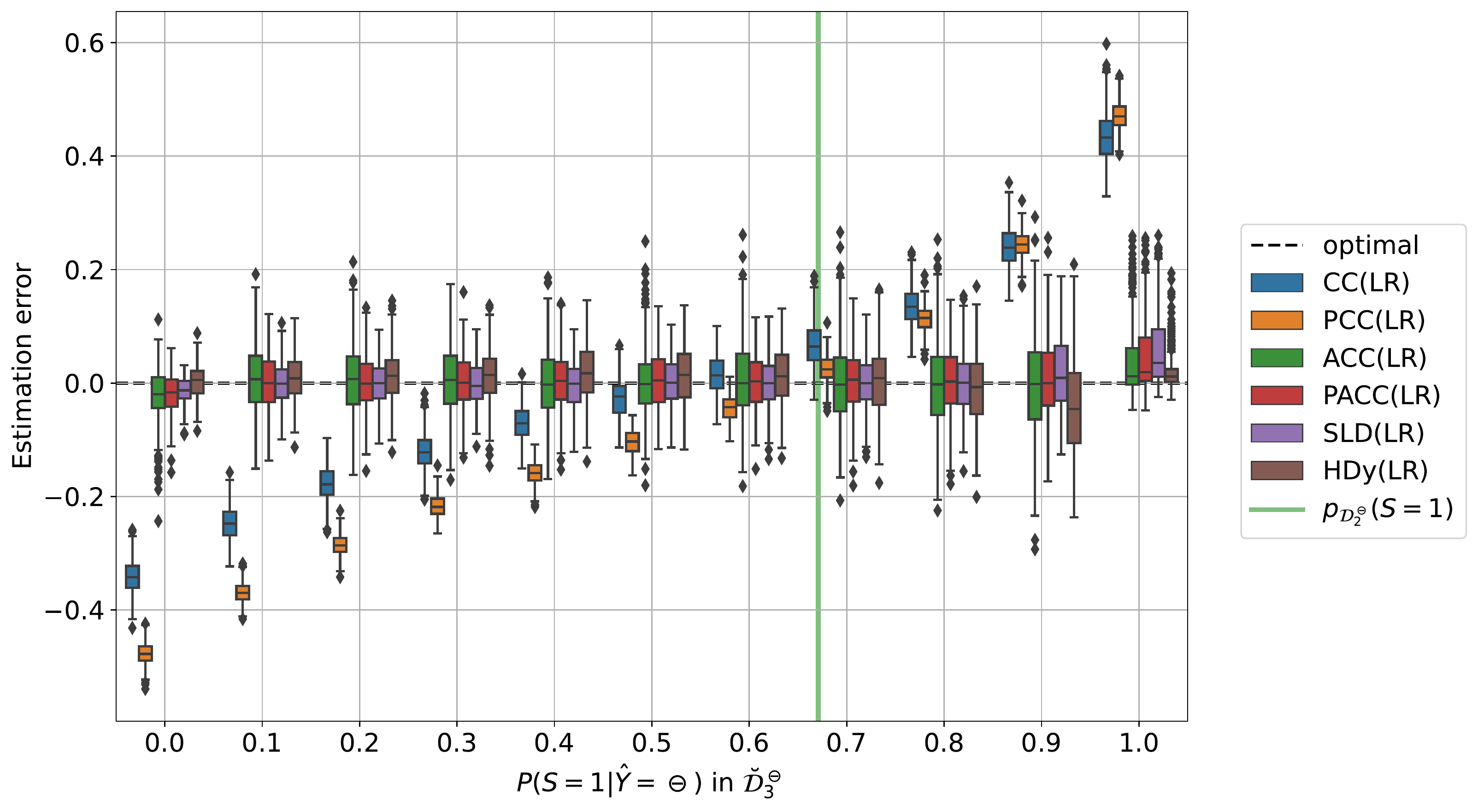}
    \caption{Protocol
    \texttt{sample-prev-$\mathcal{D}_{3}^{\ominus}$}}
    \label{fig:sample_prev_d30_fair}
  \end{subfigure} \\
  \begin{subfigure}{\mysize\textwidth}
    \centering
    \includegraphics[width=\mysize\textwidth]{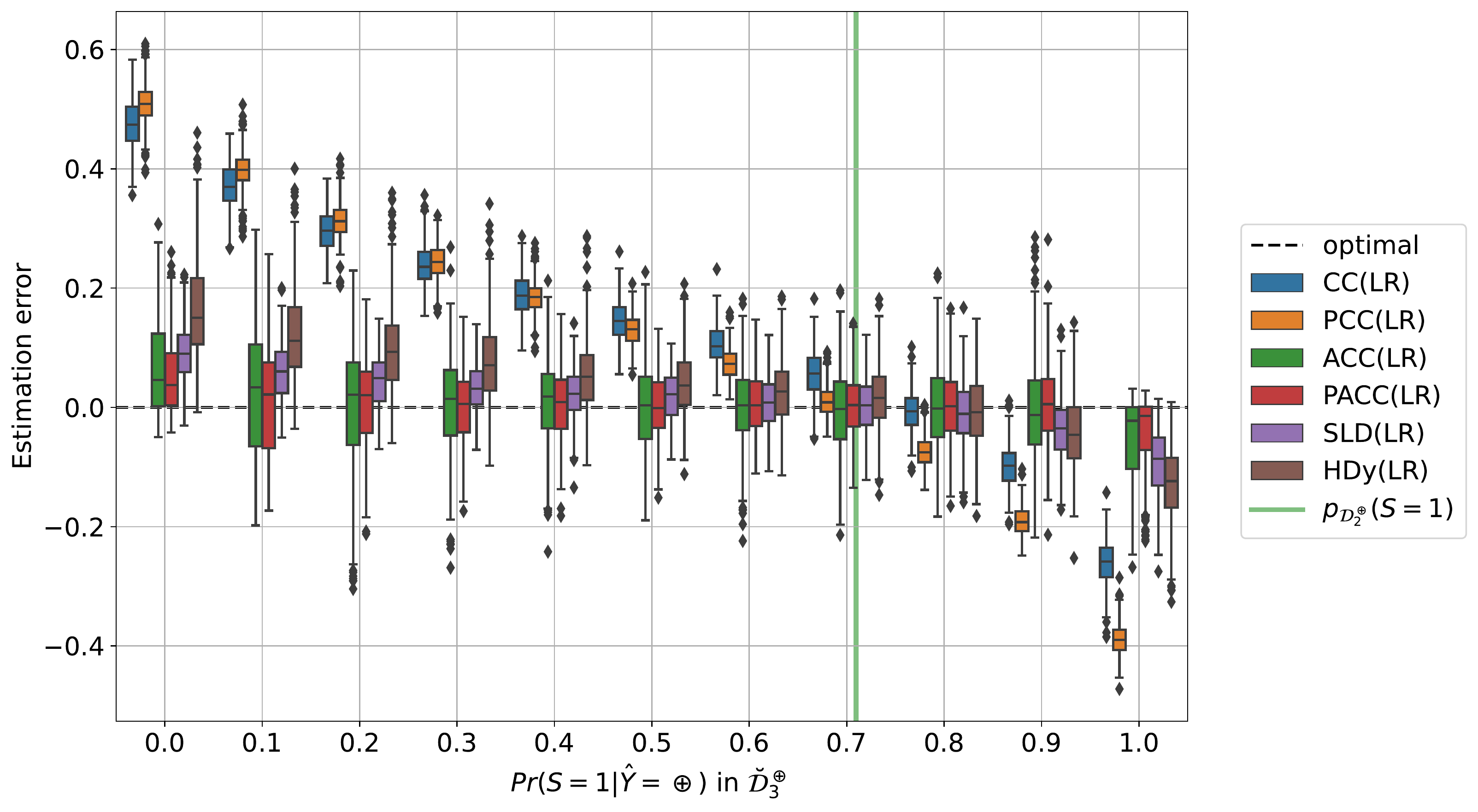}
    \caption{Protocol \texttt{sample-prev-$\mathcal{D}_{3}^{\oplus}$}}
    \label{fig:sample_prev_d31_fair}
  \end{subfigure}
  \caption{\blue{Experiments conducted according to protocol
  \texttt{sample-prev-$\mathcal{D}_3$} on a fairness-aware classifier.}}
  \label{fig:sample_prev_d3_fair}
\end{figure}

\blue{
We focus our exposition on protocol \texttt{sample-prev-$\mathcal{D}_3$}; analogous results are obtained on the remaining protocols. The fairness-aware decision tree improves DD by one order of magnitude, with an average $\delta_{h_{\mathrm{T}}}^{S}=0.017$, down from $\delta_{h}^{S}=0.158$ for LR. Figure \ref{fig:sample_prev_d3_fair}, reporting the estimation error from different quantifiers, shows the same patterns as its counterpart from Figure \ref{fig:sample_prev_d3}. CC and PCC have a sizeable bias, while ACC, PACC, SLD, and HDy display low estimation error for all the tested prevalence values. This experiment confirms the suitability of our method in measuring fairness under unawareness, also for fairness-aware classifiers. 
}

\subsection{Quantifying Without Classifying}
% \subsection{Decoupling quantification and classification}
\label{sec:q_not_c}

% -------------------------------------------------------------------

\subsubsection{Motivation and Setup}

\noindent The motivating use case for this work are internal audits of
group fairness, to characterize a model and its potential to harm
sensitive categories of users. Following \citet{aswasthi2021:ef}, we
envision this as an important first step in empowering practitioners
to argue for resources and, more broadly, to advocate for a deeper
understanding and careful evaluation of models. Unfortunately,
developing a tool to infer demographic information, even if motivated
by careful intentions and good faith, leaves open the possibility for
misuse, especially at an individual level. Once a predictive tool,
also capable of instance-level classification, is available, it will
be tempting for some actors to exploit it precisely for this purpose.

For example, the \textit{Bayesian Improved Surname Geocoding} (BISG)
method was designed to estimate population-level disparities in health
care \citep{elliott2009:uc}, but later used to identify individuals
potentially eligible for settlements related to discriminatory
practices of auto lenders
\citep{andriotis2015:US,koren2016:fu}. Automatic inference of sensitive
attributes of individuals is problematic for several reasons. Such
procedure exploits the co-occurrence of membership in a group and
display of a given trait, running the risk of learning, encoding, and
reinforcing stereotypical associations. Although also true of
group-level estimates, this practice is particularly troublesome at
the individual level, where it is likely to cause harms for people who
do not fit the norm, resulting, for instance, in misgendering and the
associated negative effects \citep{mclemore2015experiences}. Even when
``accurate'', the mere act of externally assigning sensitive labels
can be problematic. For example, gender assignment can be forceful and
cause psychological harm for individuals
\citep{keyes2018:misgendering}.

In this section, we aim to demonstrate that it is possible to decouple
the objective of (group-level) quantification of sensitive attributes
from that of (individual-level) classification. For each protocol in
the previous sections, we compute the accuracy and $F_1$ score
(defined below) of the sensitive attribute classifier $k$ underlying
the tested quantifiers, comparing it against their estimation error
for class prevalence of the same sensitive attribute
(Equation~\ref{eq:estim_err1}). Accuracy is the proportion of
correctly classified instances over the total
(Equation~\ref{eq:accuracy}) while $F_1$ is the harmonic mean of
precision and recall (Equation~\ref{eq:F1}). Both measures can be
computed from the counters of true positives (TP), true negatives
(TN), false positives (FP), and false negatives (FN), as follows:
%
% \begin{equation}
%   \label{eq:accuracy}
%   \mathrm{accuracy} = 
%   \dfrac{\TP+\TN}{\TP+\TN + \FP + \FN}
% \end{equation}
% %
% \begin{align}
%   \label{eq:F1}
%   F_{1} = & \ \left\{
%             \begin{array}{cl}
%               \dfrac{2\TP}{2\TP + \FP + \FN} & \mathrm{if} \ \TP + \FP + \FN>0 \rule[-3ex]{0mm}{7ex} \\
%               1 & \mathrm{if} \ \TP=\FP=\FN=0 \\
%             \end{array}
%   \right.
% \end{align}

\begin{align}
  \label{eq:accuracy}
  \mathrm{accuracy} = & \ 
                        \dfrac{\TP+\TN}{\TP+\TN + \FP + \FN} \\
  \label{eq:F1}
  F_{1} = & \ \left\{
            \begin{array}{cl}
              \dfrac{2\TP}{2\TP + \FP + \FN} & \mathrm{if} \ \TP + \FP + \FN>0 \rule[-3ex]{0mm}{7ex} \\
              1 & \mathrm{if} \ \TP=\FP=\FN=0 \\
            \end{array}
  \right.
\end{align}

% -------------------------------------------------------------------

\subsubsection{Results}

\def \clfplotlength {\mysize\textwidth}
\begin{figure}[!t]%[h!]
  \centering
  \begin{subfigure}{\clfplotlength}
    \centering
    \includegraphics[width=\mysize\textwidth]{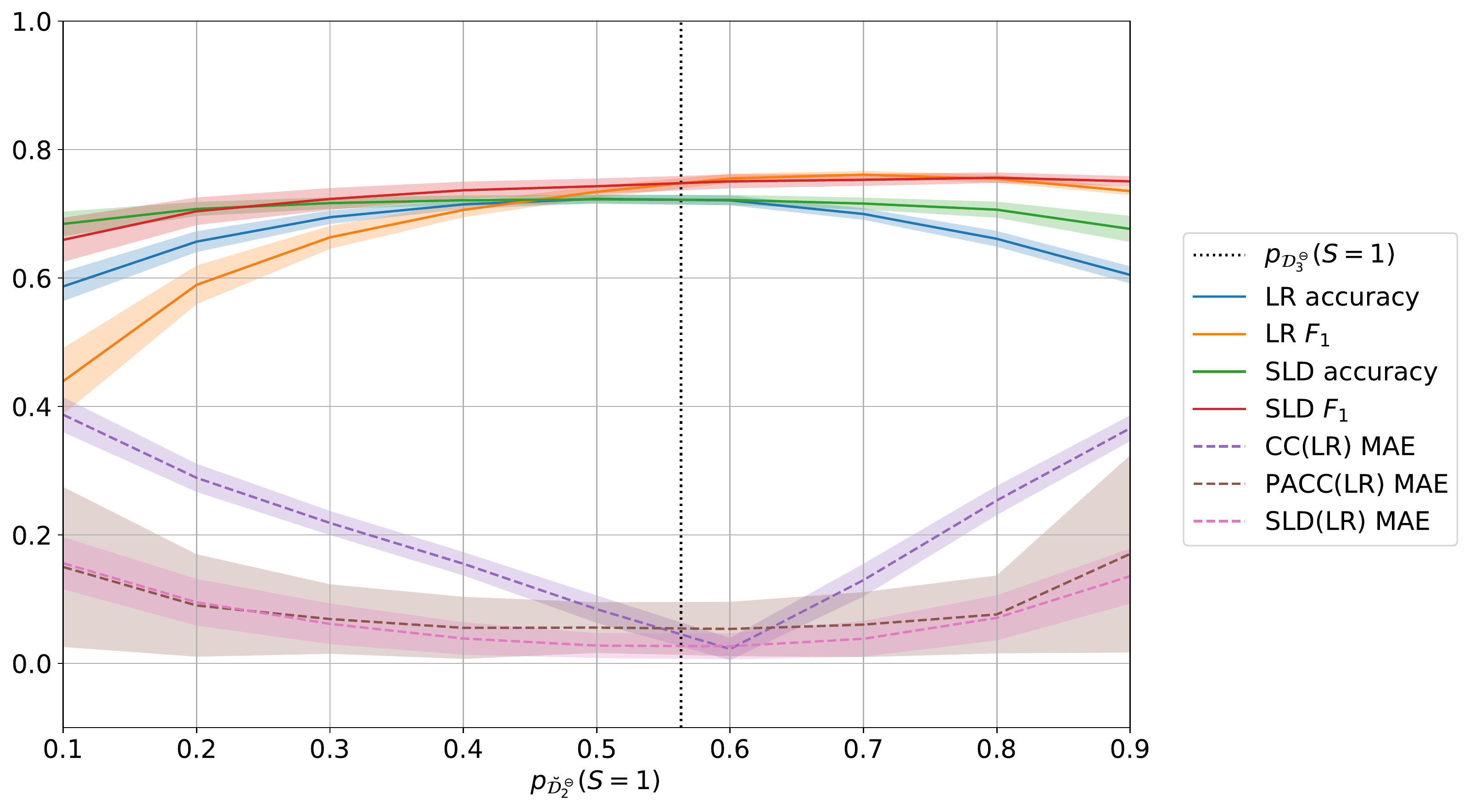}
    \caption{Protocol
    \texttt{sample-prev-$\mathcal{D}_{2}^{\ominus}$}}
    \label{fig:sample_prev_d20_q_wo_c}
  \end{subfigure} \\
  \begin{subfigure}{\clfplotlength}
    \centering
    \includegraphics[width=\mysize\textwidth]{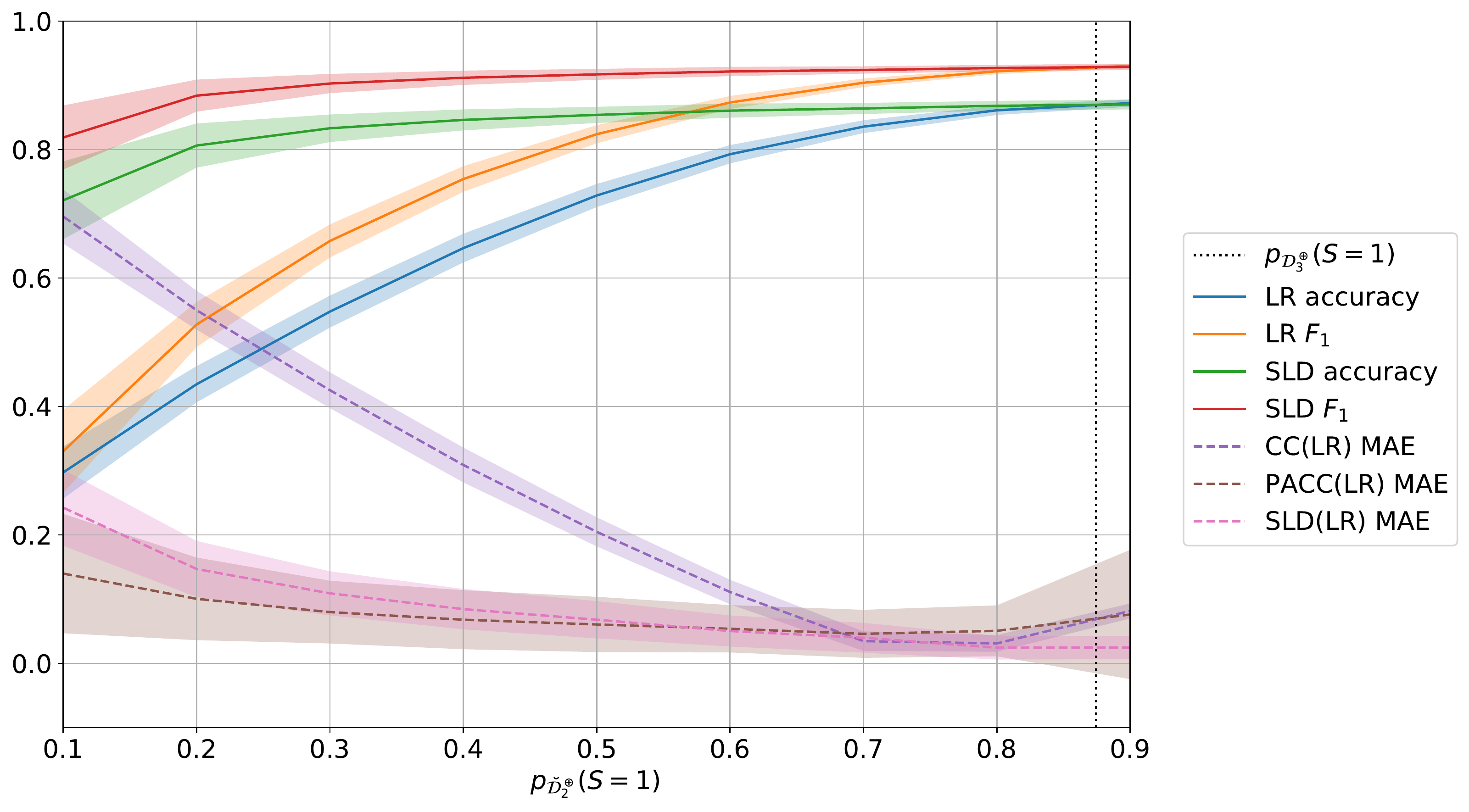}
    \caption{Protocol \texttt{sample-prev-$\mathcal{D}_{2}^{\oplus}$}}
    \label{fig:sample_prev_d21_q_wo_c}
  \end{subfigure} \\
  \caption{Performance of CC, SLD and PACC on the Adult dataset when
  used for quantification (MAE -- lower is better, dashed) and
  classification ($F_1$, accuracy -- higher is better, solid) under
  protocol \texttt{sample-prev-$\mathcal{D}_2$}. The classification
  performance of PACC is equivalent to that of CC (both equal to the
  performance of the underlying LR), and we thus omit it for
  readability.}
  \label{fig:q_wo_c_d2}
\end{figure}

\begin{figure}[!t]%[h!]
  \centering
  \begin{subfigure}{\clfplotlength}
    \centering
    \includegraphics[width=\mysize\textwidth]{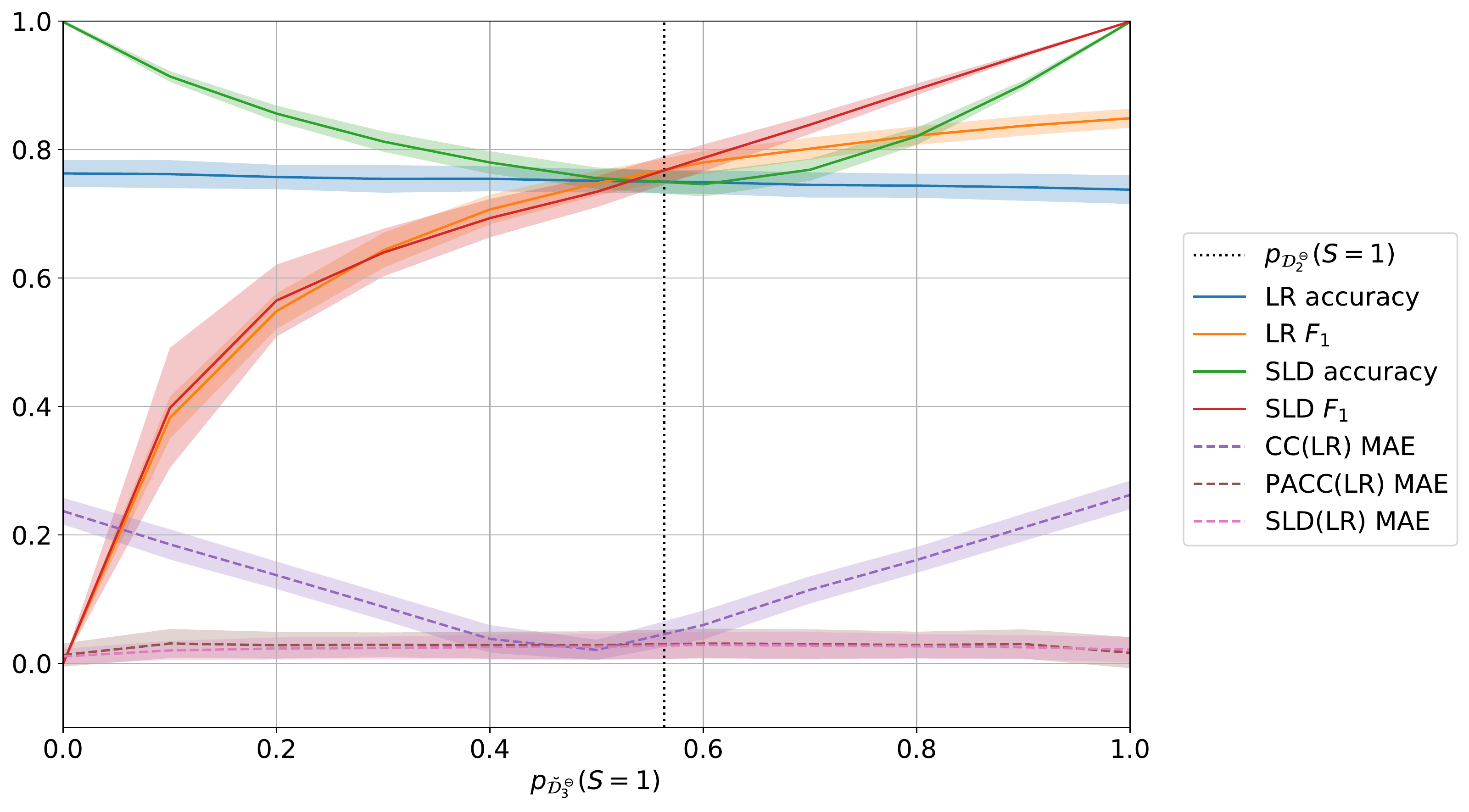}
    \caption{Protocol
    \texttt{sample-prev-$\mathcal{D}_{3}^{\ominus}$}}
    \label{fig:sample_prev_d30_q_wo_c}
  \end{subfigure} \\
  \begin{subfigure}{\clfplotlength}
    \centering
    \includegraphics[width=\mysize\textwidth]{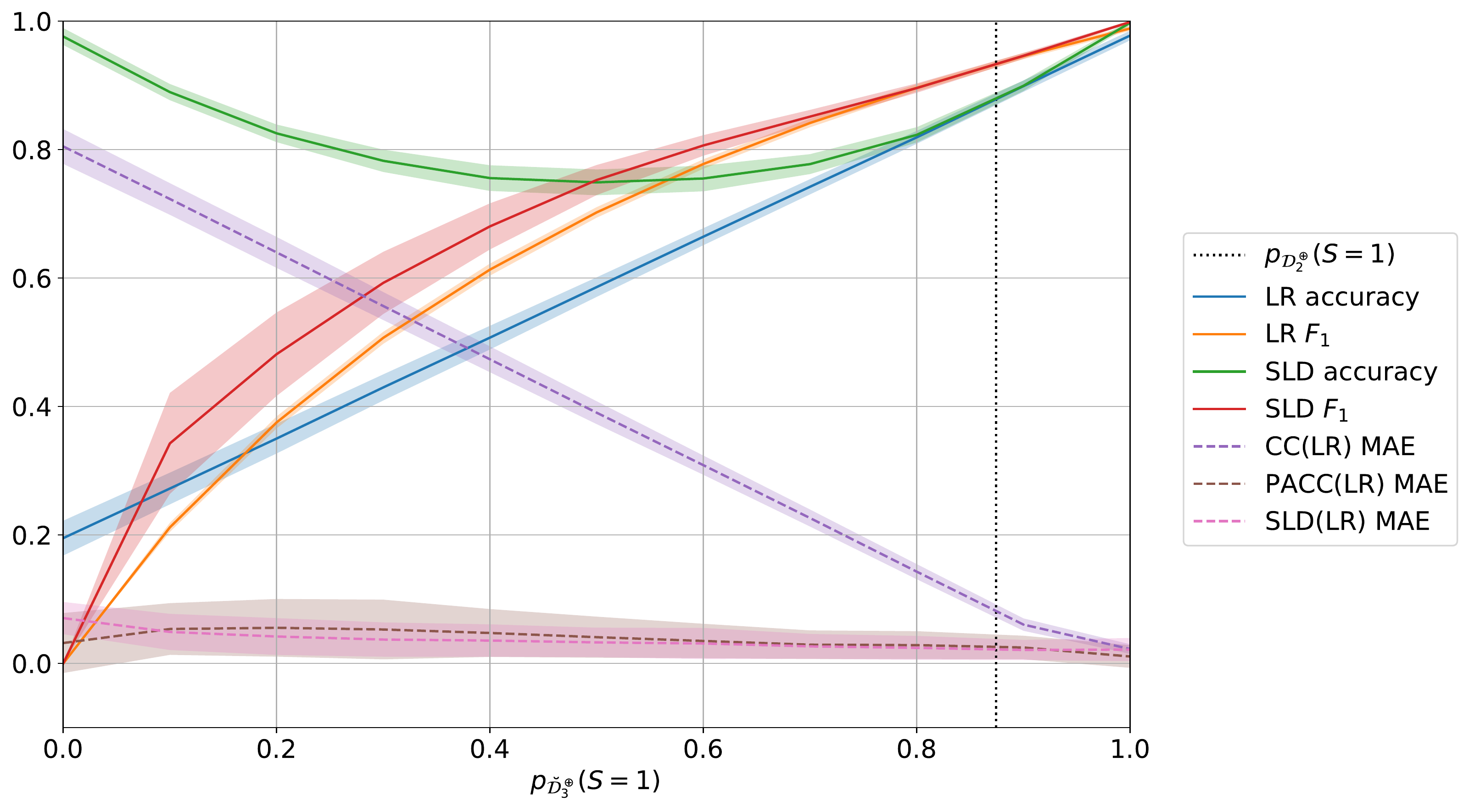}
    \caption{Protocol \texttt{sample-prev-$\mathcal{D}_{3}^{\oplus}$}}
    \label{fig:sample_prev_d31_q_wo_c}
  \end{subfigure}%
  \caption{Performance of CC, SLD and PACC on the Adult dataset when
  used for quantification (MAE -- lower is better, dashed) and
  classification ($F_1$, accuracy -- higher is better, solid) under
  protocol \texttt{sample-prev-$\mathcal{D}_3$}. The classification
  performance of PACC is equivalent to that of CC (both equal to the
  performance of the underlying LR), and we thus omit it for
  readability.}
  \label{fig:q_wo_c_d3}
\end{figure}

\noindent
Figures~\ref{fig:q_wo_c_d2} and \ref{fig:q_wo_c_d3} displays the
quantification performance (MAE -- dashed) and classification
performance ($F_1$, accuracy -- solid) of CC, SLD and PACC on the
Adult dataset under protocols \texttt{sample-prev-$\mathcal{D}_{2}$}
and \texttt{sample-prev-$\mathcal{D}_{3}$}, respectively. As usual, we
describe the results for LR-based learners and report their SVM-based
duals in the appendix (Figures~\ref{fig:q_wo_c_svm_d2} and
\ref{fig:q_wo_c_svm_d3}). To evaluate the quantification performance
of each approach, we simply report their MAE in estimating the
prevalence $p_{\mathcal{D}_{3}^{\ominus}}(S=1)$,
$p_{\mathcal{D}_{3}^{\oplus}}(S=1)$ in either test subset, depending
on the protocol at hand. To assess the performance of the sensitive
attribute classifier $k$ underlying each quantifier, we proceed as
follows. For CC and PACC, we simply run $k$ (LR) on either
$\mathcal{D}_{3}^{\ominus}$ or $\mathcal{D}_{3}^{\oplus}$, reporting
its accuracy and $F_1$ score in inferring the sensitive attribute of
individual instances. The classification performance scores of the
classifiers underlying CC and PACC are equivalent, so we omit the
latter from Figures~\ref{fig:q_wo_c_d2} and \ref{fig:q_wo_c_d3} for
readability. For SLD, we take the novel posteriors obtained by
applying the EM algorithm to either test subset, and use them for
classification with a threshold of 0.5.

Clearly, SLD improves both the quantification and classification
performance of the classifier $k$. In terms of quantification, its MAE
is consistently below that of CC, and in terms of classification, it
displays better $F_1$ and accuracy. However, under large prevalence
shifts across the auxiliary set $\mathcal{D}_2$ and the test set
$\mathcal{D}_3$, its classification performance becomes unreliable. In
particular, under protocol
\texttt{sample-prev-$\mathcal{D}_{3}^{\ominus}$}
(resp. \texttt{sample-prev-$\mathcal{D}_{3}^{\oplus}$}) in Figure
\ref{fig:sample_prev_d30_q_wo_c} (resp. Figure
\ref{fig:sample_prev_d31_q_wo_c}), for low values of the $x$ axis,
i.e., when the true prevalence values
$p_{\mathcal{D}_{3}^{\ominus}}(S=1)$
(resp. $p_{\mathcal{D}_{3}^{\oplus}}(S=1)$) becomes small, the
SLD-based classifier starts acting as a trivial rejector with low
recall, and hence low $F_1$ score. On the other hand, the
quantification performance of SLD does not degrade in the same way,
since its MAE is low and flat across the entire $x$ axis in
Figures~\ref{fig:sample_prev_d30_q_wo_c} and
\ref{fig:sample_prev_d31_q_wo_c}. This is a first hint of the fact
that classification and quantification performance may be decoupled.

PACC is another method that significantly outperforms CC in estimating
the prevalence of sensitive attributes in both test subsets
$\mathcal{D}_{3}^{\ominus}$, $\mathcal{D}_{3}^{\oplus}$. Indeed, its
MAE is well aligned with that of SLD, displaying low quantification
error under all protocols (Figures
\ref{fig:q_wo_c_d2}--\ref{fig:q_wo_c_d3}). On the other hand, its
classification performance is aligned with the accuracy and $F_1$
score of CC, which is unsatisfactory and can even become worse than
random. This fact shows that it is possible to build models which
yield good prevalence estimates for the sensitive attribute within a
sample, without providing reliable demographic estimates for single
instances. Indeed, quantification methods of type \textit{aggregative}
(that is, based on the output of a classifier -- like all methods we
use in this study) are suited to repair the initial prevalence
estimate (computed by classifying and counting) without precise
knowledge of which specific data points have been misclassified. In
the context of models to measure fairness under unawareness of
sensitive attributes, we highlight this as a positive result,
decoupling a desirable ability to estimate group-level disparities
from the potential for undesirable misuse at the individual level.

% -------------------------------------------------------------------

\subsection{Ablation Study}
\label{sec:ablation}

% -------------------------------------------------------------------

\subsubsection{Motivation and Setup}

\noindent In the previous sections, we tested six approaches to
estimate demographic disparity. For each approach, we used multiple
quantifiers for the sensitive attribute $S$, namely one for each class
in the codomain of the classifier $h$\wasblue{, as described in Step
\ref{item:trainquantifiers} of the method for quantification-based
estimate of demographic disparity}. In the binary setting adopted in
this work, where $\mathcal{Y}=\{\ominus,\oplus\}$, we trained two
quantifiers. A quantifier was trained on the set of
positively-classified instances of the auxiliary set
$\mathcal{D}_{2}^{\oplus}=\{ (\mathbf{x}_i,s_i) \in \mathcal{D}_2
\;|\; h(\mathbf{x})=\oplus \}$ and deployed to quantify the prevalence
of sensitive instances (such that $S=s$) within the test subset
$\mathcal{D}_{3}^{\oplus}$. The remaining quantifier was trained on
$\mathcal{D}_{2}^{\ominus}$ and deployed on
$\mathcal{D}_{3}^{\ominus}$.

% Although probably negligible in comparison with the computation and
% storage requirements of the classifier $h$, t
Training and maintaining multiple quantifiers is more expensive and
cumbersome than having a single one. Firstly, quantifiers that depend
on the classification outcome $\hat{y}=h(\mathbf{x})$ require
retraining every time $h$ is modified, e.g., due to a model update
being rolled out. Second, the maintenance cost is multiplied by the
number of classes $|\mathcal{Y}|$ that are possible for the outcome
variable. To ensure that these downsides are compensated by
performance improvements, we perform an ablation study and evaluate
the performance of different estimators of demographic disparity
supported by a single quantifier.

In this section we concentrate on three estimation approaches, namely
PCC, SLD, and PACC. SLD and PACC are among the best overall
performers, displaying low bias or variance across all protocols. PCC
shows great performance in situations where its posteriors are
well-calibrated on $\mathcal{D}_3$. We compare their performance in
two settings. In the first setting, adopted so far, two separate
quantifiers $q_{\ominus}$ and $q_{\oplus}$ are trained on
$\mathcal{D}_{2}^{\ominus}$, $\mathcal{D}_{2}^{\oplus}$ and deployed
on $\mathcal{D}_{3}^{\ominus}$, $\mathcal{D}_{3}^{\oplus}$,
respectively. In the second setting, we train a single quantifier $q$
on $\mathcal{D}_{2}$ and deploy it separately on
$\mathcal{D}_{3}^{\ominus}$ and $\mathcal{D}_{3}^{\oplus}$ to estimate
$\hat{\delta_{h}^{S}}$ using Equations~\eqref{eq:dd_estim} and
\eqref{eq:mu_ql}, specialized so that $q_{\ominus}$ and $q_{\oplus}$
are the same quantifier.

% -------------------------------------------------------------------

\subsubsection{Results}

\begin{figure}[t]%[h!]
  \centering
  \begin{subfigure}{\mysize\textwidth}
    \centering
    \includegraphics[width=\mysize\textwidth]{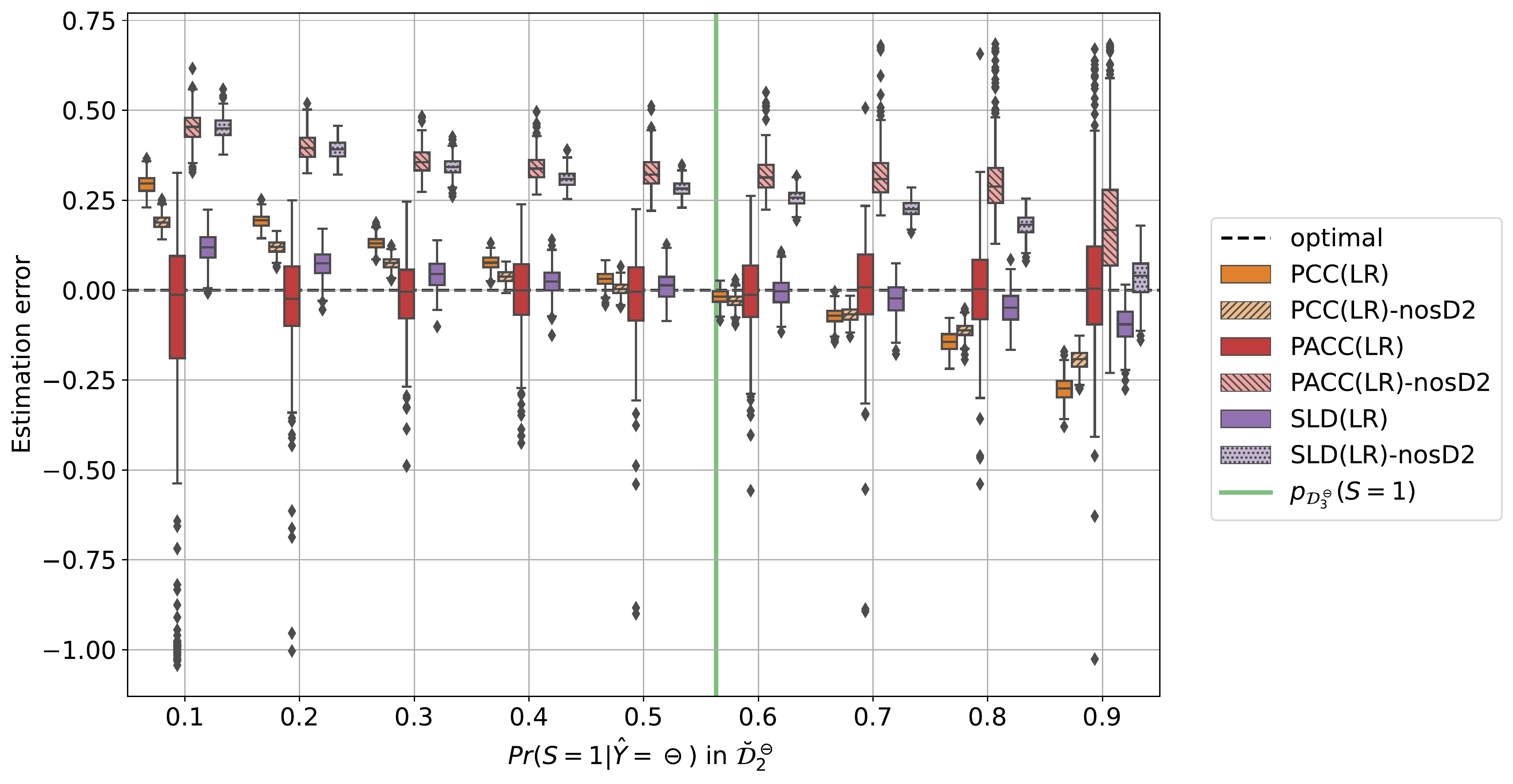}
    \caption{Protocol
    \texttt{sample-prev-$\mathcal{D}_{2}^{\ominus}$}}
    \label{fig:sample_prev_d2_ablation}
  \end{subfigure} \\
  \begin{subfigure}{\mysize\textwidth}
    \centering
    \includegraphics[width=\mysize\textwidth]{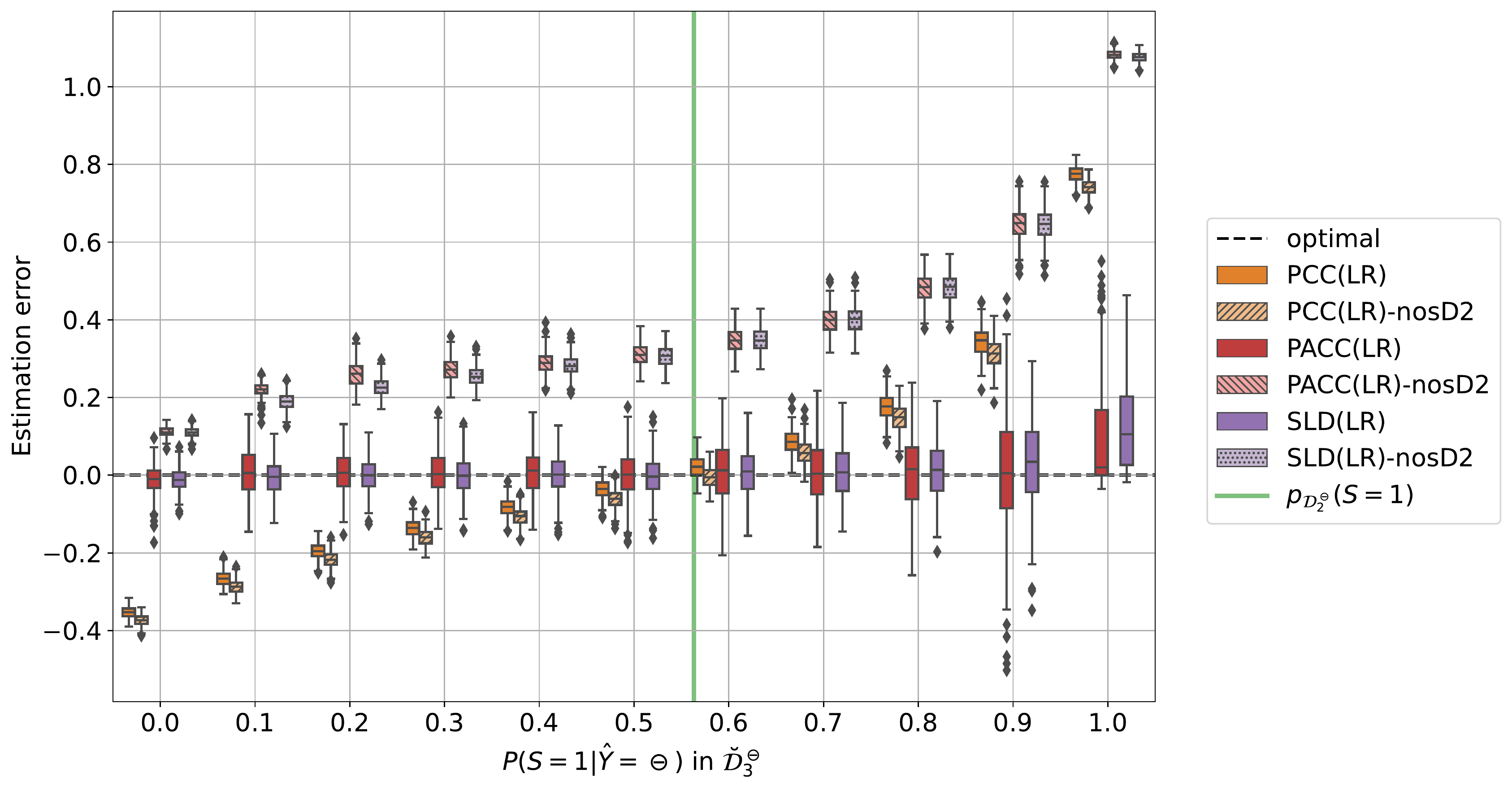}
    \caption{Protocol
    \texttt{sample-prev-$\mathcal{D}_{3}^{\ominus}$}}
    \label{fig:sample_prev_d3_ablation}
  \end{subfigure}
  \caption{Ablation study on the Adult dataset. Distribution of the
  estimation error ($y$ axis) for CC, PACC, SLD, and their
  single-quantifier counterparts, as $\Pr(S=1|\hat{Y}=\ominus)$ vary
  in $\mathcal{D}_2$, plot (a), and $\mathcal{D}_3$, plot (b).}
  \label{fig:ablation}
\end{figure}

\noindent Figure~\ref{fig:ablation} summarizes results for the Adult
dataset under two protocols that are representative of the overall
trends, namely \texttt{sample-prev-$\mathcal{D}_{2}$} (Figure
\ref{fig:sample_prev_d2_ablation}) and
\texttt{sample-prev-$\mathcal{D}_{3}$} (Figure
\ref{fig:sample_prev_d3_ablation}).\footnote{In the interest of
brevity, the figures in this section refer to LR-based quantification
on the Adult dataset under two protocols. Results for SVM-based
quantifiers under every protocol are depicted in the Appendix (Figures
\ref{fig:sample_prev_d2_ablation_svm} and
\ref{fig:sample_prev_d3_ablation_svm}). Analogous results hold on
CreditCard and COMPAS.} The $y$ axis depicts the estimation error of
PCC, SLD, PACC, and their single-quantifier counterparts, denoted by
the suffix ``nosD2'' to indicate that the auxiliary set
$\mathcal{D}_2$ is \underline{no}t \underline{s}plit into
$\mathcal{D}_{2}^{\ominus}$, $\mathcal{D}_{2}^{\oplus}$ during
training. The $x$ axis depicts the quantity of interest varied under
each protocol.

Interestingly, PCC appears to be rather insensitive to the ablation
study, so that the estimation errors of PCC and PCC-nosD2 are
well-aligned. PCC-nosD2 performs slightly better under the protocol
\texttt{sample-prev-$\mathcal{D}_{2}$}, where the auxiliary set is
small, and splitting it to learn separate quantifiers may result in
poor performance. The opposite is true for PACC-nosD2, showing a clear
decline in performance in the single-quantifier setting. This is due
to the fact that the estimates of $\mathrm{tpr}$ (and $\mathrm{fpr}$)
in $\mathcal{D}_3^{\oplus}$ and $\mathcal{D}_3^{\ominus}$ for the
adjustment (Equation~\ref{eq:pacc}) are more precise when issued by
dedicated estimators rather than a single one computed without
splitting $\mathcal{D}_2$. SLD-nosD2 also shows a sizeable performance
decay.

Under all protocols, the performance of SLD and PACC is compromised in
the absence of class-specific quantifiers $q_{\ominus}$ and
$q_{\oplus}$. If a single quantifier is trained on the full auxiliary
set $\mathcal{D}_2$, the corrections brought about by SLD and PACC can
end up worsening, rather than improving, the prevalence estimates of
vanilla CC. PCC is less sensitive to the ablation, showing small
performance differences in both directions under the single quantifier
setting. In general, it seems beneficial to partition the auxiliary
set into subsets $\mathcal{D}_{2}^{\ominus}$ and
$\mathcal{D}_{2}^{\oplus}$ according to the method in
Section~\ref{sec:adaptation}.

% -------------------------------------------------------------------

% \FloatBarrier

% -------------------------------------------------------------------

\section{\blue{Summary and Takeaway Message}}
\label{sec:discussion}

\noindent Overall, our work shows that quantification approaches are
suited to measure demographic parity under unawareness of sensitive
attributes if a small auxiliary dataset, containing sensitive and
non-sensitive attributes, is available. This is a common setting in
real-world scenarios, where such datasets may originate from targeted
efforts or voluntary disclosure. Despite an inevitable selection bias,
these datasets still represent a valuable asset for fairness audits,
if coupled with robust estimation approaches. Indeed, several
quantification methods tested in this work provide precise estimates
of demographic disparity despite the distribution shift across
training and testing caused by selection bias, and other distribution
shifts that arise in the context of human processes. This is an
important improvement over CC and PCC, previously studied in the
algorithmic fairness literature \wasblue{as the \emph{threshold
estimator} and \emph{weighted estimator}} \citep{chen2019:fu}. SLD
strikes the best balance in performance across all protocols; we
suggest its adoption, especially when the distribution shift between
development and deployment conditions has not been carefully
characterized. Moreover, while the development of proxy methods
typically comes with a potential for misuse on individuals (e.g.,
profiling), quantification approaches demonstrate the potential to
circumvent this issue. More in detail, from the above experimental
section, we summarize the following trends concerning different
approaches to measure demographic parity under unawareness.

\wasblue{\textbf{Fairness under unawareness can be measured using
quantification}, for both vanilla and fairness-aware classifiers. Group fairness under unawareness can be cast as a
prevalence estimation problem and effectively solved by methods of
proven consistency from the quantification literature. We demonstrate
several estimators that outperform the previously proposed methods
\citep{chen2019:fu}, corresponding to CC and PCC, i.e., two weak
baselines in the quantification literature.}

\textbf{CC is suboptimal}. Naïve Classify-and-Count represents the
default approach for practitioners unaware of quantification. Ad hoc
quantification methods outperform CC in most combinations of 5 protocols, 3 datasets,
and 2 underlying learners.

\textbf{PCC suffers under distribution shift}. As long as the
underlying posteriors are well-calibrated, PCC is a strong
performer. However, when its training set and test set have different
prevalence values for the sensitive attribute $S$, \wasblue{a common
situation in practice, }PCC displays a systematic estimation bias,
which increases sharply with the prior probability shift between
training and test.

\textbf{HDy, ACC and PACC deteriorate in the small data regime}. These
methods require splitting their training set (that is, the auxiliary
set $\mathcal{D}_2$), so their performance drops faster when its
cardinality is small. PACC and ACC display good median performance but
a large variance; the former method always outperforms the latter.

\textbf{SLD strikes a good balance}. This method was shown to be the
best performer under (the inevitable) distribution shift between the
auxiliary set $\mathcal{D}_2$ and the test set $\mathcal{D}_3$, with a
moderate performance decrease when $|\mathcal{D}_2|$ becomes
small. However, in situations where it is not possible to maintain
separate quantifiers for positively and negatively predicted
instances, its performance may drop substantially.

\textbf{Decoupling is possible}. Methods such as SLD and PACC fare
much better than CC in estimating group-level quantities (such as
demographic parity), while if misused for individual classification of
sensitive attributes, the improvement is minor (SLD) or zero (PACC).

% -------------------------------------------------------------------

\section{Conclusion}
\label{sec:conclusion}

\noindent Measuring the differential impact of models on groups of
individuals is important to understand their effects in the real world
and their tendency to encode and reinforce divisions and privilege
across sensitive attributes. Unfortunately, in practice, demographic
attributes are often not available. In this work, we have taken the
perspective of responsible practitioners, interested in internal
fairness audits of production models. We have proposed a novel
approach to measure group fairness under unawareness of sensitive
attributes, utilizing methods from the quantification
literature. These methods are specifically designed for group-level
\wasblue{prevalence} estimation rather than individual-level
classification. Since practitioners who try to measure fairness under
unawareness are precisely interested in
prevalence estimates \wasblue{of sensitive
attributes (Proposition \ref{prop:ql4facct})}, it is useful for the
fairness and quantification communities to exchange lessons.

We have studied the problem of estimating a classifier's
fairness under unawareness of
sensitive attributes, with access to a disjoint auxiliary set of data
for which demographic information is available. \wasblue{We have shown
how this can be cast as a quantification problem, and solved with
established approaches of proven consistency. We have conducted a
detailed empirical evaluation of different methods and their
properties focused on demographic parity.} Drawing from the
algorithmic fairness literature, we have identified five important
factors for this problem, associating each of them with a formal
evaluation protocol. We have tested several quantification-based approaches,
which, under realistic assumptions for an internal fairness audit,
outperform previously proposed estimators in the fairness
literature. We have discussed their benefits and limitations,
including the unbiasedness guarantees of some methods, and the
potential for misuse at an individual level.

% These factors mirror challenges in real-world applications,
% including dataset shift and variable cardinality for auxiliary
% datasets comprising demographic information. We have tested five
% quantification methods under every protocol, comparing them against
% the naïve Classify-and-Count (CC) method, which represents the
% default approach for practitioners unaware of quantification. Each
% quantification approach was shown to outperform CC under most
% combinations of 5 protocols, 3 datasets, and 2 underlying
% learners. Moreover, we have shown a simple approach to integrate
% quantification methods into existing machine learning pipelines with
% little orchestration effort, and demonstrated the importance of each
% component through an ablation study.

% Finally, we have considered the problem of model misuse to infer
% demographic characteristics at an individual level, which represents
% a concern when developing models to measure group fairness via proxy
% attributes. Through a dedicated set of experiments, we have shown
% that it is possible to obtain precise estimates of demographic
% disparity from methods that have poor classification
% performance. This is a positive result for decoupling these two
% objectives, which should help deter from the exploitation of these
% models for individual-level inference.

Future work may require a deeper study of the relation between
classification and quantification performance and the extent to which
these two objectives can be decoupled. It would be interesting to
explicitly target decoupling through learners aimed at maximizing
quantification performance subject to a low classification performance
constraint. \wasblue{Ideally, decoupling should provide precise privacy
guarantees to individuals while allowing for precise group-level
estimates.} Another important avenue for future work is the study of
confidence intervals for fairness estimates provided by quantification
methods. A reliable indication of confidence for estimates of group
fairness may be invaluable for a practitioner arguing for resources
and attention to the disparate effects of a model on different
populations. \wasblue{Finally, the estimators presented in this work may
be plugged into optimization procedures aimed at improving, rather
than measuring, algorithmic fairness. \blue{Mixed loss functions, jointly optimizing accuracy and fairness can be optimized, even under unawareness of sensitive attributes, with our methods providing fairness estimates at each iteration.} It will be interesting to
evaluate fairness estimators in this broader context and extend them,
e.g., to ranking problems and counterfactual settings.}

\acks{The work by Alessandro Fabris was supported by MIUR (Italian
Ministry for University and Research) under the ``Departments of
Excellence'' initiative (Law 232/2016). The work by Andrea Esuli,
Alejandro Moreo, and Fabrizio Sebastiani has been supported by the
\textsc{SoBigData++} project, funded by the European Commission (Grant
871042) under the H2020 Programme INFRAIA-2019-1, by the
\textsc{AI4Media} project, funded by the European Commission (Grant
951911) under the H2020 Programme ICT-48-2020, and by the \textsc{SoBigData.it}, \textsc{FAIR} and \textsc{ITSERR}
projects funded by the Italian Ministry of University and Research under the NextGenerationEU program.
These authors' opinions
do not necessarily reflect those of the funding agencies.}

\clearpage

\appendix
\section{The SLD Method}
\label{app:sld}

\noindent SLD \citep{Saerens:2002uq} produces prevalence estimates
$\hat{p}_{\sigma}^{\mathrm{SLD}}(s)$ iteratively, using EM
algorithms. In detail, given two sets, $L$ and $U$, where the former
represents the \emph{labelled} one (training set) and the latter
represents the \emph{unlabelled} one (test set). The method iterates
until convergence (i.e., the difference between the prevalence
estimated across two consecutive iterations is less than a tolerance
factor $\epsilon$ --we use $\epsilon=1e-4$) or until a maximum number
of iterations is reached. The pseudocode describing SLD is as follows:
% \afabcomment{A couple words on stopping condition?}

\IncMargin{1em}
\begin{algorithm}
  \SetKwInOut{Input}{Input}\SetKwInOut{Output}{Output} \Input{Class
  prevalence values $p_{L}(s)$ on $L$; \\
  \hspace{0em} Posterior probabilities $\pi_s(\mathbf{x}_i)$, for all
  $\mathbf{x}_i\in U$;} \Output{Estimates $\hat{p}_{U}(s)$ of class
  prevalence values on $U$;} \BlankLine \tcc{Initialisation}
  $t\leftarrow 0$\; \For{$s\in S$}{
  $\hat{p}^{(t)}_{U}(s)\leftarrow p_{L}(s)$\;
  \For{$\mathbf{x}_i\in U$}{
  $\Pr^{(t)}(s|\textbf{x}_i)\leftarrow \pi_s(\mathbf{x}_i)$\;}}
  \BlankLine \tcc{Main Iteration Cycle} \While{stopping condition =
  false}{ $t\leftarrow t+1$\;
  \For{$s\in S$\nllabel{line:posteriorupdate}}{
  \For{$\mathbf{x}_i\in U$}{
  $\Pr^{(t)}(s|\textbf{x}_i)\leftarrow
  \displaystyle\frac{\displaystyle\frac{\hat{p}^{(t-1)}_{U}(s)
  }{\hat{p}^{(0)}_{U}(s)}\cdot
  \Pr^{(0)}(s|\textbf{x}_i)}{\displaystyle\sum_{s\in
  S}\displaystyle\frac{\hat{p}^{(t-1)}_{U}(s)}{\hat{p}^{(0)}_{U}(s)}\cdot
  \Pr^{(0)}(s|\textbf{x}_i)}$}
  $\hat{p}^{(t)}_{U}(s) \leftarrow
  \displaystyle\frac{1}{|U|}\displaystyle\sum_{\textbf{x}_i \in
  U}\Pr^{(t)}(s|\textbf{x}_i)$ \nllabel{line:classpriorupdate}}}
  \BlankLine \tcc{Generate output}
  \For{$s\in
  S$}{$\hat{p}_{U}^{\mathrm{SLD}}(s)\leftarrow \hat{p}^{(t)}_{U}(s)$}
  \caption{The SLD algorithm
  \citep{Saerens:2002uq}.} \label{alg:Saerens}
\end{algorithm}\DecMargin{1em}

% -------------------------------------------------------------------

\clearpage
\newpage

\section{The HDy Method}
\label{app:hdy}

\noindent HDy \citep{Gonzalez-Castro:2013fk} measures the divergence
between two distributions of posterior probabilities (i.e., as
returned by a calibrated classifier) $v$ and $u$ in terms of the
Hellinger Distance (HD), defined as
\begin{equation*}
  \mathrm{HD}(v,u)=\sqrt{\int \left(\sqrt{v(x)} - \sqrt{u(x)} \right)^2 dx}
\end{equation*}
The HD between two continuous distributions $v$ and $u$ is typically
approximated by discretizing $v$ and $u$ across bins and then
integrating
\begin{equation*}
  \hat{\mathrm{HD}}(V,U)= \sqrt{\sum_{i=1}^b \left( \sqrt{\frac{|V_{i}|}{|V|}} - \sqrt{\frac{|U_{i}|}{|U|}} \right)^2} 
\end{equation*}
\noindent with $V$ and $U$ the discrete distributions, $b$ the number
of bins and $V_i$, $U_i$ representing the frequency in the $i$th bin
for each distribution, respectively.

The method seeks the $\alpha$ parameter that yields the smallest
distance between the validation distribution $V$ (typically, a
held-out split of the training set that has not been used to train the
classifier) and the unlabelled distribution $U$, i.e.,
\begin{equation*}
  \alpha^* = \arg\min_{\alpha\in [0,1]} \hat{\mathrm{HD}}(V^\alpha,U)
\end{equation*}
\noindent where $V^\alpha$ is the mixture of the positive distribution
($V^{S=1}$) and the negative distribution ($V^{S=0}$) defined by
\begin{equation*}
  V^{\alpha}(x) = (1-\alpha)\cdot V^{S=0}(x) + \alpha\cdot V^{S=1}(x)
\end{equation*}
\noindent HDy returns $\alpha^*$ as the sought positive class
prevalence
\begin{align*}
  \hat{p}_{\sigma}^{\mathrm{HDy}}(1) = \alpha^*
\end{align*}

\noindent Since the number of bins $b$ could have a significant impact
on the calculation, one typically returns the median of the
distribution of the best $\alpha$'s found for a range of $b$'s (in our
case, we explore $b\in[10,20,30,\ldots,110]$).

% -------------------------------------------------------------------

\clearpage
\newpage

\section{Proof of Proposition \ref{prop:we_pcc}}
\label{sec:proof}

\noindent We show that Equation~\eqref{eq:mu_we} and
Equation~\eqref{eq:mu_ql} are equivalent when the latter is
instantiated by prevalence estimates given by PCC:
\begin{align*}
  \hat{\mu}^{\mathrm{PCC}}(s) &= \hat{p}_{\mathcal{D}_{3}^{\oplus}}^{\mathrm{PCC}}(s) \frac{p_{\mathcal{D}_{3}}(\oplus)}{\hat{p}_{\mathcal{D}_{3}^{\oplus}}^{\mathrm{PCC}}(s) p_{\mathcal{D}_{3}}(\oplus) + \hat{p}_{\mathcal{D}_{3}^{\ominus}}^{\mathrm{PCC}}(s) p_{\mathcal{D}_{3}}(\ominus)}
\end{align*}

\noindent The terms in the denominator can be written as
\begin{align*}
  \hat{p}_{\mathcal{D}_{3}^{\oplus}}^{\mathrm{PCC}}(s) &= \frac{\sum_{\mathbf{x}_i \in \mathcal{D}_{3}^{\oplus}}\pi_{s}(\mathbf{x}_i)}{|\mathcal{D}_{3}^{\oplus}|} \\
                                                       &= \frac{\sum_{\mathbf{x}_i} \pi_s(\mathbf{x}_i) h_{\oplus}(\mathbf{x}_i)}{\sum_{\mathbf{x}_i} h_{\oplus}(\mathbf{x}_i)}
\end{align*}
\begin{align*}
  \hat{p}_{\mathcal{D}_{3}^{\ominus}}^{\mathrm{PCC}}(s) = \frac{\sum_{\mathbf{x}_i} \pi_s(\mathbf{x}_i) (1-h_{\oplus}(\mathbf{x}_i))}{\sum_{\mathbf{x}_i} (1-h_{\oplus}(\mathbf{x}_i))}
\end{align*}
\begin{align*}
  p_{\mathcal{D}_{3}}(\oplus) = \frac{\sum_{\mathbf{x}_i} h_{\oplus}(\mathbf{x}_i)}{|\mathcal{D}_3|}
\end{align*}
\begin{align*}
  p_{\mathcal{D}_{3}}(\ominus) = \frac{\sum_{\mathbf{x}_i} (1-h_{\oplus}(\mathbf{x}_i))}{|\mathcal{D}_3|}
\end{align*}

\noindent Plugging them into the denominator yields
\begin{align*}
  \hat{\mu}^{\mathrm{PCC}}(s) &= \hat{p}_{\mathcal{D}_{3}^{\oplus}}^{\mathrm{PCC}}(s) \frac{p_{\mathcal{D}_{3}}(\oplus)}{\frac{\sum_{\mathbf{x}_i} \pi_s(\mathbf{x}_i)}{|\mathcal{D}_3|}} \\
                              &= \frac{\sum_{\mathbf{x}_i} \pi_s(\mathbf{x}_i) h_{\oplus}(\mathbf{x}_i)}{\sum_{\mathbf{x}_i} h_{\oplus}(\mathbf{x}_i)} \cdot \frac{\sum_{\mathbf{x}_i} h_{\oplus}(\mathbf{x}_i)}{|\mathcal{D}_3|} \cdot \frac{|\mathcal{D}_3|}{\sum_{\mathbf{x}_i} \pi_s(\mathbf{x}_i)} \\
                              &= \hat{\mu}^{\mathrm{WE}}(s)
\end{align*}
\noindent The equivalence between CC and TE is straightforward.
\hfill\qed

% -------------------------------------------------------------------

\clearpage
\newpage

\section{SVM-based Quantification}
\label{app:svm}

\noindent In this appendix we report the results of experiments,
analogous to the ones in Sections
\ref{sec:sample_prev_d1}-\ref{sec:q_not_c}, where quantifiers are
wrapped around an SVM classifier rather than an LR classifier. The
experimental protocols are summarized in Tables
\ref{tab:sample_prev_d3_svm}-\ref{tab:flip_prev_d1_svm}. The ablation
study is depicted in Figures
%\ref{fig:sample_prev_d1_ablation_svm}-\ref{fig:sample_prev_d3_ablation_svm}. 
\ref{fig:sample_prev_d2_ablation_svm} and \ref{fig:sample_prev_d3_ablation_svm}.
Experiments
on decoupling the quantification performance of a model from its
classification performance are reported in Figures
\ref{fig:q_wo_c_svm_d2} and \ref{fig:q_wo_c_svm_d3}.

\begin{table}[h!]
  \caption{Results obtained in the experiments run according to
  protocol \texttt{sample-prev-$\mathcal{D}_3$} with the SVM-based
  classifier.}
  \small \centering
  % \resizebox{\textwidth}{!}{%
  
\begin{tabular}{llcccc} \toprule
 & & $\downarrow$ MAE & $\downarrow$ MSE & $\uparrow$ $P(\mathrm{AE}<0.1)$ & $\uparrow$ $P(\mathrm{AE}<0.2)$ \\ \midrule
\multirow{7}{*}{Adult} &CC(SVM) & 0.410$^{\phantom{\ddag}}\pm0.323$ & 0.273$^{\phantom{\ddag}}\pm0.341$ & 0.193 & 0.365 \\
 &PCC(SVM) & 0.308$^{\phantom{\ddag}}\pm0.244$ & 0.154$^{\phantom{\ddag}}\pm0.210$ & 0.230 & 0.412 \\
 &ACC(SVM) & 0.107$^{\phantom{\ddag}}\pm0.105$ & 0.022$^{\phantom{\ddag}}\pm0.053$ & 0.606 & 0.857 \\
 &PACC(SVM) & 0.059$^{\phantom{\ddag}}\pm0.057$ & 0.007$^{\phantom{\ddag}}\pm0.016$ & 0.824 & 0.971 \\
 &SLD(SVM) & \textbf{0.056}$^{\phantom{\ddag}}\pm0.050$ & \textbf{0.006}$^{\phantom{\ddag}}\pm0.011$ & \textbf{0.836} & \textbf{0.983} \\
 &HDy(SVM) & 0.104$^{\phantom{\ddag}}\pm0.078$ & 0.017$^{\phantom{\ddag}}\pm0.028$ & 0.546 & 0.895 \\
 &MLPE & 0.397$^{\phantom{\ddag}}\pm0.298$ & 0.246$^{\phantom{\ddag}}\pm0.316$ & 0.162 & 0.294 \\

 \vspace{0.01cm} \\ 
\multirow{7}{*}{COMPAS} &CC(SVM) & 0.543$^{\phantom{\ddag}}\pm0.370$ & 0.432$^{\phantom{\ddag}}\pm0.474$ & 0.115 & 0.235 \\
 &PCC(SVM) & 0.339$^{\phantom{\ddag}}\pm0.243$ & 0.174$^{\phantom{\ddag}}\pm0.216$ & 0.179 & 0.343 \\
 &ACC(SVM) & 0.497$^{\phantom{\ddag}}\pm0.346$ & 0.367$^{\phantom{\ddag}}\pm0.448$ & 0.127 & 0.224 \\
 &PACC(SVM) & 0.269$^{\phantom{\ddag}}\pm0.207$ & 0.115$^{\phantom{\ddag}}\pm0.165$ & 0.250 & 0.445 \\
 &SLD(SVM) & \textbf{0.227}$^{\phantom{\ddag}}\pm0.202$ & \textbf{0.092}$^{\phantom{\ddag}}\pm0.154$ & \textbf{0.335} & \textbf{0.566} \\
 &HDy(SVM) & 0.265$^{\phantom{\ddag}}\pm0.204$ & 0.112$^{\phantom{\ddag}}\pm0.162$ & 0.238 & 0.459 \\
 &MLPE & 0.349$^{\phantom{\ddag}}\pm0.249$ & 0.184$^{\phantom{\ddag}}\pm0.227$ & 0.175 & 0.332 \\

 \vspace{0.01cm} \\ 
\multirow{7}{*}{CreditCard} &CC(SVM) & 0.346$^{\phantom{\ddag}}\pm0.241$ & 0.178$^{\phantom{\ddag}}\pm0.213$ & 0.171 & 0.335 \\
 &PCC(SVM) & 0.329$^{\phantom{\ddag}}\pm0.215$ & 0.155$^{\phantom{\ddag}}\pm0.161$ & 0.173 & 0.335 \\
 &ACC(SVM) & 0.358$^{\phantom{\ddag}}\pm0.270$ & 0.201$^{\phantom{\ddag}}\pm0.276$ & 0.175 & 0.348 \\
 &PACC(SVM) & 0.267$^{\phantom{\ddag}}\pm0.215$ & 0.118$^{\phantom{\ddag}}\pm0.180$ & 0.252 & 0.473 \\
 &SLD(SVM) & 0.243$^{\ddag}\pm0.191$ & 0.096$^{\dag}\pm0.143$ & 0.268 & 0.496 \\
 &HDy(SVM) & \textbf{0.237}$^{\phantom{\ddag}}\pm0.186$ & \textbf{0.090}$^{\phantom{\ddag}}\pm0.137$ & \textbf{0.271} & \textbf{0.507} \\
 &MLPE & 0.334$^{\phantom{\ddag}}\pm0.218$ & 0.159$^{\phantom{\ddag}}\pm0.165$ & 0.172 & 0.330 \\
\bottomrule
\end{tabular}%

  % }%
  \label{tab:sample_prev_d3_svm}
\end{table}

\begin{table}[h!]
  \caption{Results obtained in the experiments run according to
  protocol \texttt{sample-prev-$\mathcal{D}_2$} with the SVM-based
  classifier.}
  \small \centering
  % \resizebox{\textwidth}{!}{%
  
\begin{tabular}{llcccc} \toprule
 & & $\downarrow$ MAE & $\downarrow$ MSE & $\uparrow$ $P(\mathrm{AE}<0.1)$ & $\uparrow$ $P(\mathrm{AE}<0.2)$ \\ \midrule
\multirow{7}{*}{Adult} &CC(SVM) & 0.217$^{\phantom{\ddag}}\pm0.168$ & 0.075$^{\phantom{\ddag}}\pm0.106$ & 0.286 & 0.554 \\
 &PCC(SVM) & 0.242$^{\phantom{\ddag}}\pm0.190$ & 0.095$^{\phantom{\ddag}}\pm0.129$ & 0.303 & 0.507 \\
 &ACC(SVM) & 0.150$^{\phantom{\ddag}}\pm0.169$ & 0.051$^{\phantom{\ddag}}\pm0.147$ & 0.458 & 0.799 \\
 &PACC(SVM) & 0.111$^{\phantom{\ddag}}\pm0.114$ & 0.025$^{\phantom{\ddag}}\pm0.085$ & 0.555 & 0.888 \\
 &SLD(SVM) & \textbf{0.095}$^{\phantom{\ddag}}\pm0.100$ & \textbf{0.019}$^{\phantom{\ddag}}\pm0.067$ & \textbf{0.634} & \textbf{0.929} \\
 &HDy(SVM) & 0.182$^{\phantom{\ddag}}\pm0.151$ & 0.056$^{\phantom{\ddag}}\pm0.084$ & 0.381 & 0.634 \\
 &MLPE & 0.295$^{\phantom{\ddag}}\pm0.218$ & 0.134$^{\phantom{\ddag}}\pm0.165$ & 0.240 & 0.415 \\

 \vspace{0.01cm} \\ 
\multirow{7}{*}{COMPAS} &CC(SVM) & 0.506$^{\phantom{\ddag}}\pm0.255$ & 0.321$^{\phantom{\ddag}}\pm0.267$ & 0.036 & 0.116 \\
 &PCC(SVM) & 0.266$^{\phantom{\ddag}}\pm0.187$ & 0.106$^{\ddag}\pm0.128$ & 0.226 & 0.430 \\
 &ACC(SVM) & 0.479$^{\phantom{\ddag}}\pm0.276$ & 0.306$^{\phantom{\ddag}}\pm0.305$ & 0.076 & 0.175 \\
 &PACC(SVM) & 0.356$^{\phantom{\ddag}}\pm0.260$ & 0.194$^{\phantom{\ddag}}\pm0.261$ & 0.167 & 0.324 \\
 &SLD(SVM) & 0.297$^{\phantom{\ddag}}\pm0.244$ & 0.148$^{\phantom{\ddag}}\pm0.228$ & 0.231 & 0.424 \\
 &HDy(SVM) & \textbf{0.255}$^{\phantom{\ddag}}\pm0.192$ & \textbf{0.102}$^{\phantom{\ddag}}\pm0.141$ & \textbf{0.240} & \textbf{0.479} \\
 &MLPE & 0.275$^{\phantom{\ddag}}\pm0.192$ & 0.112$^{\phantom{\ddag}}\pm0.134$ & 0.220 & 0.410 \\

 \vspace{0.01cm} \\ 
\multirow{7}{*}{CreditCard} &CC(SVM) & 0.428$^{\phantom{\ddag}}\pm0.253$ & 0.247$^{\phantom{\ddag}}\pm0.237$ & 0.106 & 0.229 \\
 &PCC(SVM) & \textbf{0.209}$^{\phantom{\ddag}}\pm0.142$ & \textbf{0.064}$^{\phantom{\ddag}}\pm0.076$ & 0.285 & \textbf{0.537} \\
 &ACC(SVM) & 0.531$^{\phantom{\ddag}}\pm0.316$ & 0.382$^{\phantom{\ddag}}\pm0.352$ & 0.090 & 0.164 \\
 &PACC(SVM) & 0.542$^{\phantom{\ddag}}\pm0.313$ & 0.391$^{\phantom{\ddag}}\pm0.352$ & 0.078 & 0.145 \\
 &SLD(SVM) & 0.445$^{\phantom{\ddag}}\pm0.284$ & 0.279$^{\phantom{\ddag}}\pm0.288$ & 0.118 & 0.237 \\
 &HDy(SVM) & 0.246$^{\phantom{\ddag}}\pm0.193$ & 0.098$^{\phantom{\ddag}}\pm0.150$ & 0.256 & 0.487 \\
 &MLPE & 0.210$^{\ddag}\pm0.143$ & 0.065$^{\ddag}\pm0.077$ & \textbf{0.287} & 0.530 \\
\bottomrule
\end{tabular}%

  % }%
\end{table}

\begin{table}[h!]
  \caption{Results obtained in the experiments run according to
  protocol \texttt{sample-size-$\mathcal{D}_2$} with the SVM-based
  classifier}
  \small \centering
  % \resizebox{\textwidth}{!}{%
  
\begin{tabular}{llcccc} \toprule
 & & $\downarrow$ MAE & $\downarrow$ MSE & $\uparrow$ $P(\mathrm{AE}<0.1)$ & $\uparrow$ $P(\mathrm{AE}<0.2)$ \\ \midrule
\multirow{7}{*}{Adult} &CC(SVM) & 0.131$^{\phantom{\ddag}}\pm0.028$ & 0.018$^{\phantom{\ddag}}\pm0.007$ & 0.133 & \textbf{1.000} \\
 &PCC(SVM) & \textbf{0.012}$^{\phantom{\ddag}}\pm0.011$ & \textbf{0.000}$^{\phantom{\ddag}}\pm0.000$ & \textbf{1.000} & \textbf{1.000} \\
 &ACC(SVM) & 0.081$^{\phantom{\ddag}}\pm0.107$ & 0.018$^{\phantom{\ddag}}\pm0.076$ & 0.759 & 0.935 \\
 &PACC(SVM) & 0.051$^{\phantom{\ddag}}\pm0.066$ & 0.007$^{\phantom{\ddag}}\pm0.036$ & 0.873 & 0.977 \\
 &SLD(SVM) & 0.043$^{\phantom{\ddag}}\pm0.062$ & 0.006$^{\phantom{\ddag}}\pm0.030$ & 0.907 & 0.971 \\
 &HDy(SVM) & 0.045$^{\phantom{\ddag}}\pm0.034$ & 0.003$^{\phantom{\ddag}}\pm0.005$ & 0.918 & 0.999 \\
 &MLPE & 0.013$^{\ddag}\pm0.011$ & 0.000$^{\ddag}\pm0.001$ & \textbf{1.000} & \textbf{1.000} \\

 \vspace{0.01cm} \\ 
\multirow{7}{*}{COMPAS} &CC(SVM) & 0.355$^{\phantom{\ddag}}\pm0.044$ & 0.128$^{\phantom{\ddag}}\pm0.031$ & 0.000 & 0.003 \\
 &PCC(SVM) & \textbf{0.029}$^{\phantom{\ddag}}\pm0.019$ & \textbf{0.001}$^{\phantom{\ddag}}\pm0.001$ & \textbf{0.999} & \textbf{1.000} \\
 &ACC(SVM) & 0.389$^{\phantom{\ddag}}\pm0.212$ & 0.196$^{\phantom{\ddag}}\pm0.212$ & 0.090 & 0.171 \\
 &PACC(SVM) & 0.284$^{\phantom{\ddag}}\pm0.231$ & 0.134$^{\phantom{\ddag}}\pm0.210$ & 0.233 & 0.444 \\
 &SLD(SVM) & 0.228$^{\phantom{\ddag}}\pm0.199$ & 0.092$^{\phantom{\ddag}}\pm0.158$ & 0.305 & 0.555 \\
 &HDy(SVM) & 0.130$^{\phantom{\ddag}}\pm0.102$ & 0.027$^{\phantom{\ddag}}\pm0.041$ & 0.461 & 0.778 \\
 &MLPE & 0.029$^{\ddag}\pm0.021$ & 0.001$^{\ddag}\pm0.002$ & \textbf{0.999} & \textbf{1.000} \\

 \vspace{0.01cm} \\ 
\multirow{7}{*}{CreditCard} &CC(SVM) & 0.189$^{\phantom{\ddag}}\pm0.079$ & 0.042$^{\phantom{\ddag}}\pm0.032$ & 0.132 & 0.552 \\
 &PCC(SVM) & \textbf{0.016}$^{\phantom{\ddag}}\pm0.013$ & \textbf{0.000}$^{\phantom{\ddag}}\pm0.001$ & \textbf{1.000} & \textbf{1.000} \\
 &ACC(SVM) & 0.358$^{\phantom{\ddag}}\pm0.269$ & 0.201$^{\phantom{\ddag}}\pm0.267$ & 0.181 & 0.328 \\
 &PACC(SVM) & 0.322$^{\phantom{\ddag}}\pm0.257$ & 0.169$^{\phantom{\ddag}}\pm0.248$ & 0.213 & 0.385 \\
 &SLD(SVM) & 0.243$^{\phantom{\ddag}}\pm0.192$ & 0.096$^{\phantom{\ddag}}\pm0.142$ & 0.284 & 0.497 \\
 &HDy(SVM) & 0.105$^{\phantom{\ddag}}\pm0.096$ & 0.020$^{\phantom{\ddag}}\pm0.051$ & 0.583 & 0.876 \\
 &MLPE & 0.017$^{\ddag}\pm0.013$ & 0.000$^{\ddag}\pm0.001$ & \textbf{1.000} & \textbf{1.000} \\
\bottomrule
\end{tabular}%

  % }%
\end{table}

\begin{table}[h!]
  \caption{Results obtained in the experiments run according to
  protocol \texttt{sample-prev-$\mathcal{D}_1$} with the SVM-based
  classifier.}
  \label{tab:sample_prev_d1_svm}
  \small \centering
  % \resizebox{\textwidth}{!}{%
  
\begin{tabular}{llcccc} \toprule
 & & $\downarrow$ MAE & $\downarrow$ MSE & $\uparrow$ $P(\mathrm{AE}<0.1)$ & $\uparrow$ $P(\mathrm{AE}<0.2)$ \\ \midrule
\multirow{7}{*}{Adult} &CC(SVM) & 0.126$^{\phantom{\ddag}}\pm0.047$ & 0.018$^{\phantom{\ddag}}\pm0.011$ & 0.268 & 0.959 \\
 &PCC(SVM) & \textbf{0.007}$^{\phantom{\ddag}}\pm0.005$ & \textbf{0.000}$^{\phantom{\ddag}}\pm0.000$ & \textbf{1.000} & \textbf{1.000} \\
 &ACC(SVM) & 0.032$^{\phantom{\ddag}}\pm0.032$ & 0.002$^{\phantom{\ddag}}\pm0.014$ & 0.968 & 0.998 \\
 &PACC(SVM) & 0.018$^{\phantom{\ddag}}\pm0.014$ & 0.001$^{\phantom{\ddag}}\pm0.001$ & \textbf{1.000} & \textbf{1.000} \\
 &SLD(SVM) & 0.013$^{\phantom{\ddag}}\pm0.010$ & 0.000$^{\phantom{\ddag}}\pm0.000$ & \textbf{1.000} & \textbf{1.000} \\
 &HDy(SVM) & 0.022$^{\phantom{\ddag}}\pm0.016$ & 0.001$^{\phantom{\ddag}}\pm0.001$ & \textbf{1.000} & \textbf{1.000} \\
 &MLPE & 0.008$^{\phantom{\ddag}}\pm0.006$ & 0.000$^{\phantom{\ddag}}\pm0.000$ & \textbf{1.000} & \textbf{1.000} \\

 \vspace{0.01cm} \\ 
\multirow{7}{*}{COMPAS} &CC(SVM) & 0.334$^{\phantom{\ddag}}\pm0.087$ & 0.119$^{\phantom{\ddag}}\pm0.055$ & 0.018 & 0.063 \\
 &PCC(SVM) & \textbf{0.026}$^{\phantom{\ddag}}\pm0.018$ & \textbf{0.001}$^{\phantom{\ddag}}\pm0.001$ & \textbf{1.000} & \textbf{1.000} \\
 &ACC(SVM) & 0.349$^{\phantom{\ddag}}\pm0.196$ & 0.160$^{\phantom{\ddag}}\pm0.174$ & 0.123 & 0.221 \\
 &PACC(SVM) & 0.208$^{\phantom{\ddag}}\pm0.179$ & 0.075$^{\phantom{\ddag}}\pm0.134$ & 0.332 & 0.578 \\
 &SLD(SVM) & 0.170$^{\phantom{\ddag}}\pm0.166$ & 0.057$^{\phantom{\ddag}}\pm0.117$ & 0.422 & 0.707 \\
 &HDy(SVM) & 0.113$^{\phantom{\ddag}}\pm0.089$ & 0.021$^{\phantom{\ddag}}\pm0.031$ & 0.528 & 0.839 \\
 &MLPE & 0.027$^{\ddag}\pm0.019$ & 0.001$^{\ddag}\pm0.001$ & \textbf{1.000} & \textbf{1.000} \\

 \vspace{0.01cm} \\ 
\multirow{7}{*}{CreditCard} &CC(SVM) & 0.152$^{\phantom{\ddag}}\pm0.100$ & 0.033$^{\phantom{\ddag}}\pm0.038$ & 0.360 & 0.708 \\
 &PCC(SVM) & \textbf{0.010}$^{\phantom{\ddag}}\pm0.007$ & \textbf{0.000}$^{\phantom{\ddag}}\pm0.000$ & \textbf{1.000} & \textbf{1.000} \\
 &ACC(SVM) & 0.194$^{\phantom{\ddag}}\pm0.160$ & 0.063$^{\phantom{\ddag}}\pm0.104$ & 0.342 & 0.618 \\
 &PACC(SVM) & 0.132$^{\phantom{\ddag}}\pm0.108$ & 0.029$^{\phantom{\ddag}}\pm0.046$ & 0.482 & 0.778 \\
 &SLD(SVM) & 0.110$^{\phantom{\ddag}}\pm0.091$ & 0.020$^{\phantom{\ddag}}\pm0.032$ & 0.560 & 0.845 \\
 &HDy(SVM) & 0.080$^{\phantom{\ddag}}\pm0.061$ & 0.010$^{\phantom{\ddag}}\pm0.014$ & 0.683 & 0.953 \\
 &MLPE & 0.011$^{\dag}\pm0.008$ & 0.000$^{\phantom{\ddag}}\pm0.000$ & \textbf{1.000} & \textbf{1.000} \\
\bottomrule
\end{tabular}%

  % }%
 
\end{table}

\begin{table}[h!]
  \caption{Results obtained in the experiments run according to
  protocol \texttt{flip-prev-$\mathcal{D}_1$} with the SVM-based
  classifier.}
  \label{tab:flip_prev_d1_svm}
  \small \centering
  % \resizebox{\textwidth}{!}{%
  
\begin{tabular}{llcccc} \toprule
 & & $\downarrow$ MAE & $\downarrow$ MSE & $\uparrow$ $P(\mathrm{AE}<0.1)$ & $\uparrow$ $P(\mathrm{AE}<0.2)$ \\ \midrule
\multirow{7}{*}{Adult} &CC(SVM) & 0.175$^{\phantom{\ddag}}\pm0.085$ & 0.038$^{\phantom{\ddag}}\pm0.028$ & 0.231 & 0.544 \\
 &PCC(SVM) & \textbf{0.007}$^{\phantom{\ddag}}\pm0.006$ & \textbf{0.000}$^{\phantom{\ddag}}\pm0.000$ & \textbf{1.000} & \textbf{1.000} \\
 &ACC(SVM) & 0.032$^{\phantom{\ddag}}\pm0.028$ & 0.002$^{\phantom{\ddag}}\pm0.004$ & 0.969 & 0.999 \\
 &PACC(SVM) & 0.020$^{\phantom{\ddag}}\pm0.015$ & 0.001$^{\phantom{\ddag}}\pm0.001$ & \textbf{1.000} & \textbf{1.000} \\
 &SLD(SVM) & 0.015$^{\phantom{\ddag}}\pm0.012$ & 0.000$^{\phantom{\ddag}}\pm0.001$ & \textbf{1.000} & \textbf{1.000} \\
 &HDy(SVM) & 0.022$^{\phantom{\ddag}}\pm0.018$ & 0.001$^{\phantom{\ddag}}\pm0.001$ & 1.000 & \textbf{1.000} \\
 &MLPE & 0.009$^{\phantom{\ddag}}\pm0.006$ & 0.000$^{\phantom{\ddag}}\pm0.000$ & \textbf{1.000} & \textbf{1.000} \\

 \vspace{0.01cm} \\ 
\multirow{7}{*}{COMPAS} &CC(SVM) & 0.395$^{\phantom{\ddag}}\pm0.113$ & 0.169$^{\phantom{\ddag}}\pm0.083$ & 0.021 & 0.055 \\
 &PCC(SVM) & \textbf{0.027}$^{\phantom{\ddag}}\pm0.019$ & \textbf{0.001}$^{\phantom{\ddag}}\pm0.001$ & 0.998 & \textbf{1.000} \\
 &ACC(SVM) & 0.399$^{\phantom{\ddag}}\pm0.204$ & 0.201$^{\phantom{\ddag}}\pm0.193$ & 0.094 & 0.174 \\
 &PACC(SVM) & 0.207$^{\phantom{\ddag}}\pm0.176$ & 0.074$^{\phantom{\ddag}}\pm0.131$ & 0.325 & 0.587 \\
 &SLD(SVM) & 0.160$^{\phantom{\ddag}}\pm0.146$ & 0.047$^{\phantom{\ddag}}\pm0.095$ & 0.418 & 0.722 \\
 &HDy(SVM) & 0.112$^{\phantom{\ddag}}\pm0.084$ & 0.020$^{\phantom{\ddag}}\pm0.028$ & 0.528 & 0.842 \\
 &MLPE & 0.027$^{\ddag}\pm0.019$ & 0.001$^{\ddag}\pm0.001$ & \textbf{0.999} & \textbf{1.000} \\

 \vspace{0.01cm} \\ 
\multirow{7}{*}{CreditCard} &CC(SVM) & 0.165$^{\phantom{\ddag}}\pm0.105$ & 0.038$^{\phantom{\ddag}}\pm0.039$ & 0.328 & 0.627 \\
 &PCC(SVM) & 0.012$^{\ddag}\pm0.009$ & 0.000$^{\ddag}\pm0.000$ & \textbf{1.000} & \textbf{1.000} \\
 &ACC(SVM) & 0.227$^{\phantom{\ddag}}\pm0.186$ & 0.086$^{\phantom{\ddag}}\pm0.130$ & 0.303 & 0.542 \\
 &PACC(SVM) & 0.144$^{\phantom{\ddag}}\pm0.120$ & 0.035$^{\phantom{\ddag}}\pm0.059$ & 0.442 & 0.742 \\
 &SLD(SVM) & 0.118$^{\phantom{\ddag}}\pm0.095$ & 0.023$^{\phantom{\ddag}}\pm0.036$ & 0.512 & 0.819 \\
 &HDy(SVM) & 0.092$^{\phantom{\ddag}}\pm0.070$ & 0.013$^{\phantom{\ddag}}\pm0.019$ & 0.621 & 0.913 \\
 &MLPE & \textbf{0.012}$^{\phantom{\ddag}}\pm0.009$ & \textbf{0.000}$^{\phantom{\ddag}}\pm0.000$ & \textbf{1.000} & \textbf{1.000} \\
\bottomrule
\end{tabular}%

  % }%

\end{table}

\clearpage

% ABLATION
% \begin{figure}[h!]%[h!]
%   \centering
%   \includegraphics[width=\mysize\textwidth]{./plots/ablation/SVM/protD1_vary_prev_d1}
%   \caption{Results obtained in the ablation study on the Adult dataset
%   with SVM-based quantification for protocol
%   \texttt{sample-prev-$\mathcal{D}_{1}$}.}
%   \label{fig:sample_prev_d1_ablation_svm}
% \end{figure}

% \begin{figure}[h!]%[h!]
%   \centering
%   \includegraphics[width=\mysize\textwidth]{./plots/ablation/SVM/protD1_vary_prev_d1_flip}
%   \caption{Results obtained in the ablation study on the Adult dataset
%   with SVM-based quantification for protocol
%   \texttt{flip-prev-$\mathcal{D}_{1}$}.}
%   \label{fig:flip_prev_d1_ablation_svm}
% \end{figure}

% \begin{figure}[h!]%[h!]
%   \centering
%   \includegraphics[width=\mysize\textwidth]{./plots/ablation/SVM/protD2_vary_size}
%   \caption{Results obtained in the ablation study on the Adult dataset
%   with SVM-based quantification for protocol
%   \texttt{sample-size-$\mathcal{D}_{2}$}.}
%   \label{fig:sample_size_d2_ablation_svm}
% \end{figure}

\begin{figure}[h!]%[h!]
  \begin{subfigure}{\mysize\textwidth}
    \centering
    \includegraphics[width=\mysize\textwidth]{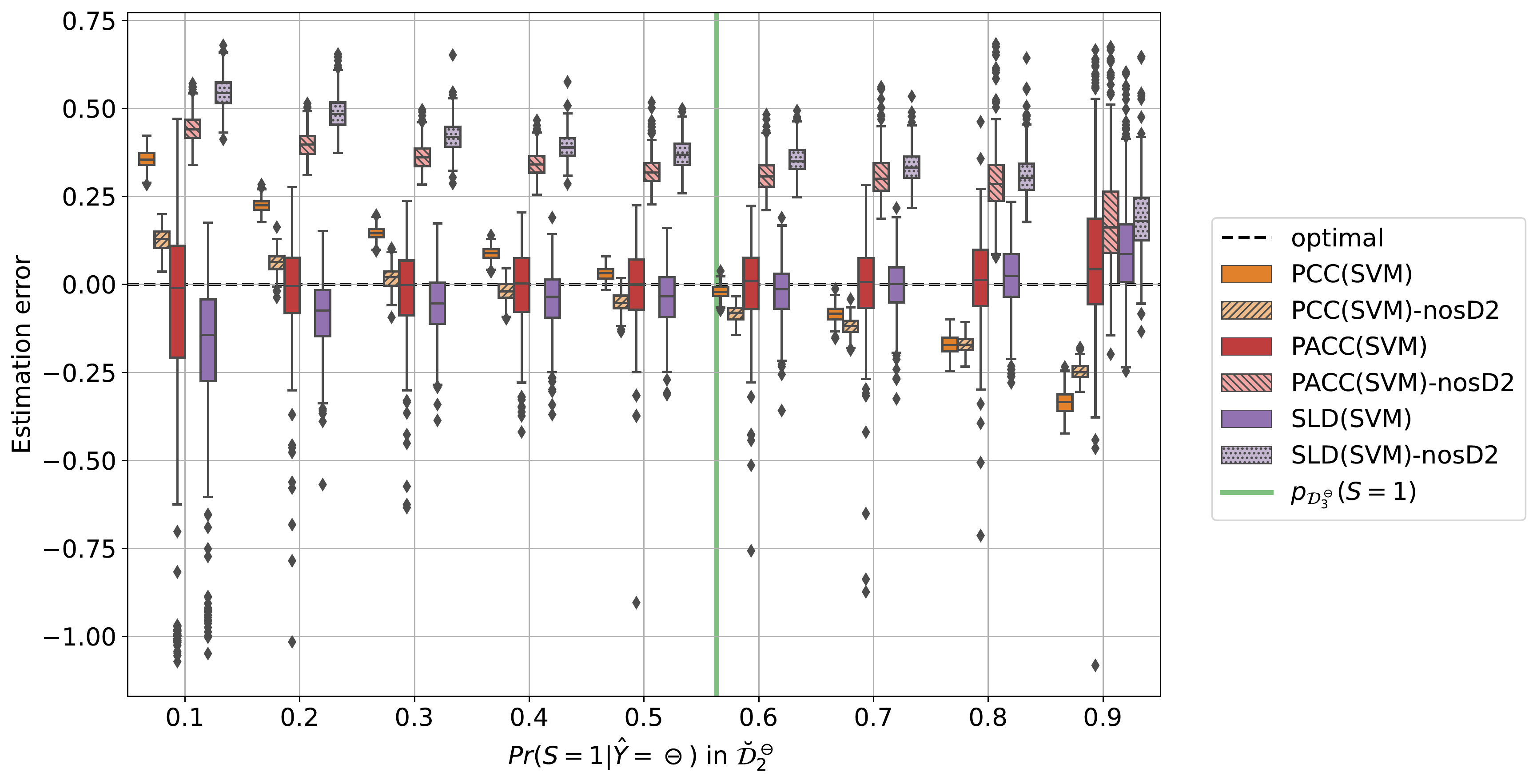}
    \caption{Protocol
    \texttt{sample-prev-$\mathcal{D}_{2}^{\ominus}$}}
    \label{fig:sample_prev_d20_ablation_svm}
  \end{subfigure} \\
  \begin{subfigure}{\mysize\textwidth}
    \centering
    \includegraphics[width=\mysize\textwidth]{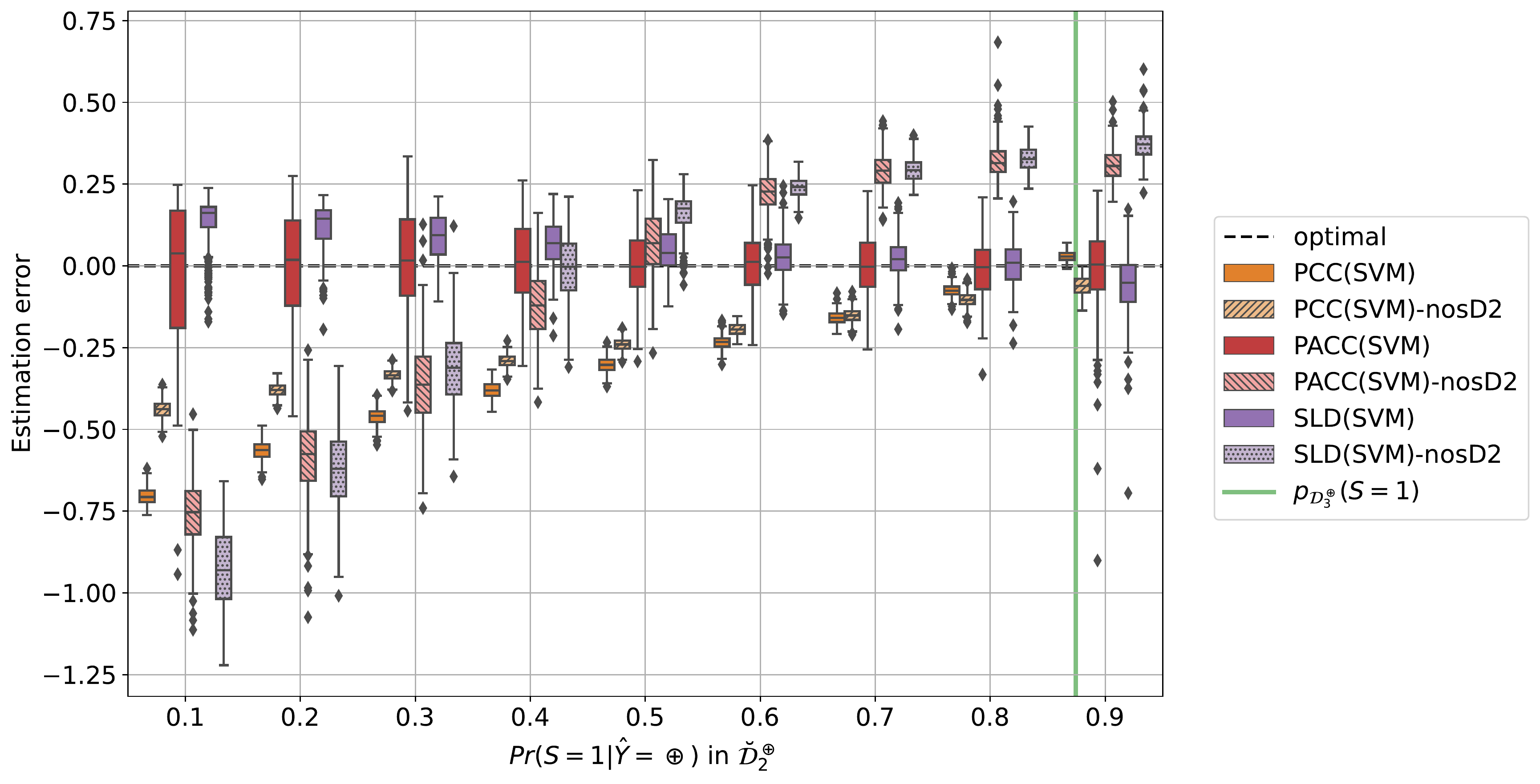}
    \caption{Protocol \texttt{sample-prev-$\mathcal{D}_{2}^{\oplus}$}}
    \label{fig:sample_prev_d21_ablation_svm}
  \end{subfigure}%
  \caption{Results obtained in the ablation study on the Adult dataset
  with SVM-based quantification for protocol
  \texttt{sample-prev-$\mathcal{D}_{2}$}.}
  \label{fig:sample_prev_d2_ablation_svm}
\end{figure}

\begin{figure}[h!]%[h!]
  \begin{subfigure}{\mysize\textwidth}
    \centering
    \includegraphics[width=\mysize\textwidth]{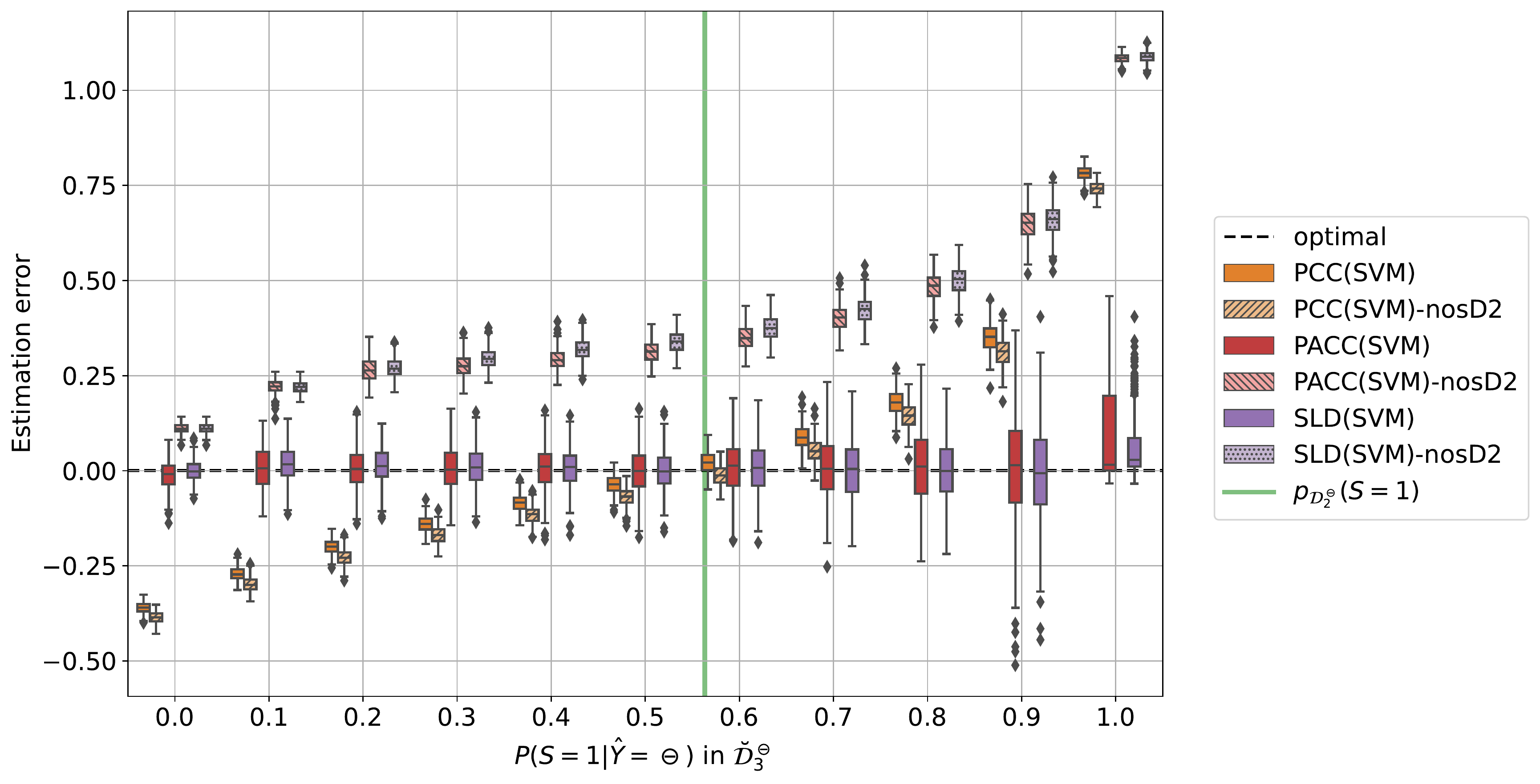}
    \caption{Protocol
    \texttt{sample-prev-$\mathcal{D}_{3}^{\ominus}$}}
    \label{fig:sample_prev_d30_ablation_svm}
  \end{subfigure} \\
  \begin{subfigure}{\mysize\textwidth}
    \centering
    \includegraphics[width=\mysize\textwidth]{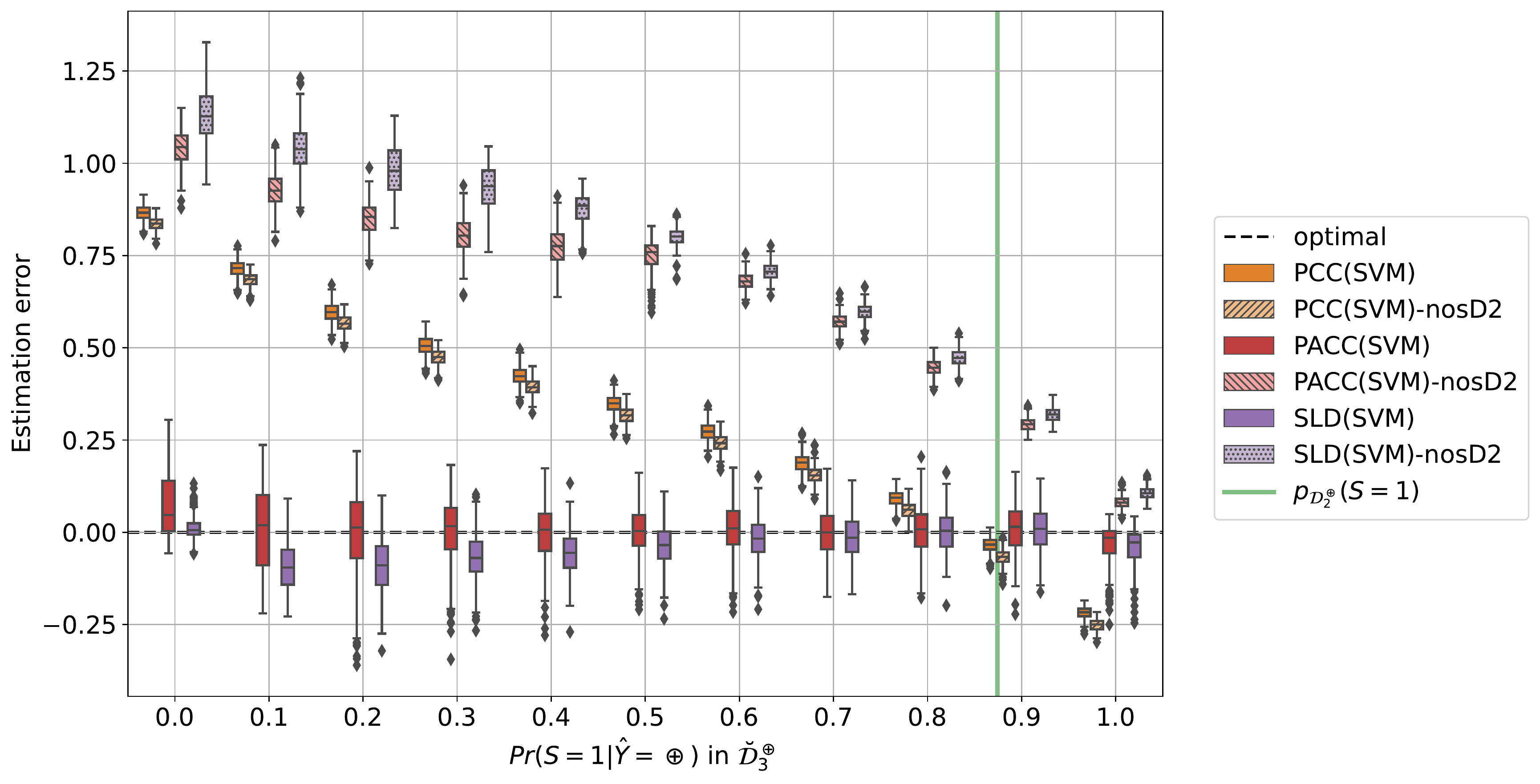}
    \caption{Protocol \texttt{sample-prev-$\mathcal{D}_{3}^{\oplus}$}}
    \label{fig:sample_prev_d31_ablation_svm}
  \end{subfigure}%
  \caption{Results obtained in the ablation study on the Adult dataset
  with SVM-based quantification for protocol
  \texttt{sample-prev-$\mathcal{D}_{3}$}.}
  \label{fig:sample_prev_d3_ablation_svm}
\end{figure}

% DECOUPLING
\begin{figure}[h!]%[h!]
  \centering
  \begin{subfigure}{\clfplotlength}
    \centering
    \includegraphics[width=\mysize\textwidth]{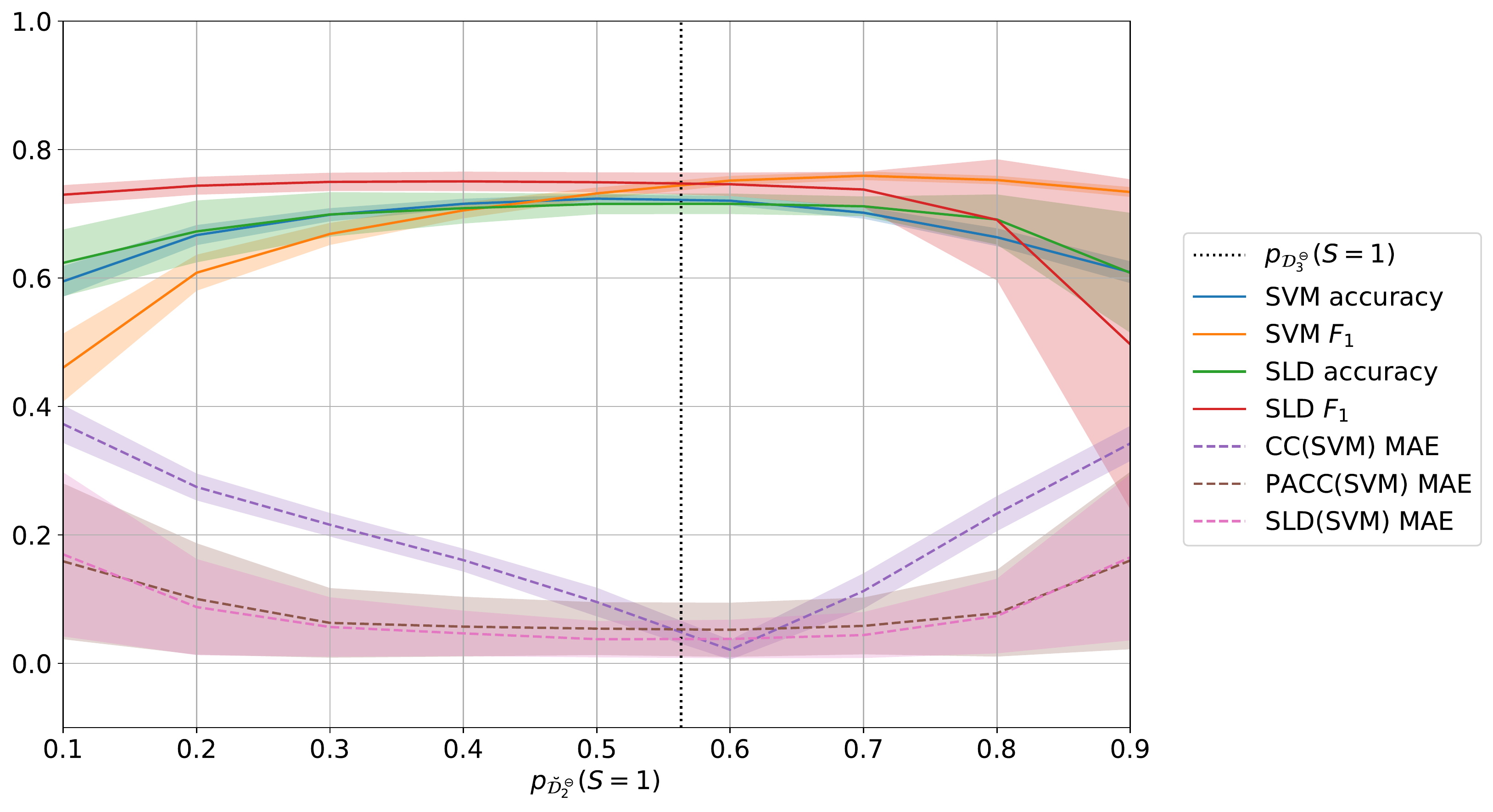}
    \caption{Protocol
    \texttt{sample-prev-$\mathcal{D}_{2}^{\ominus}$}}
    \label{fig:sample_prev_d20_q_wo_c_svm}
  \end{subfigure} \\
  \begin{subfigure}{\clfplotlength}
    \centering
    \includegraphics[width=\mysize\textwidth]{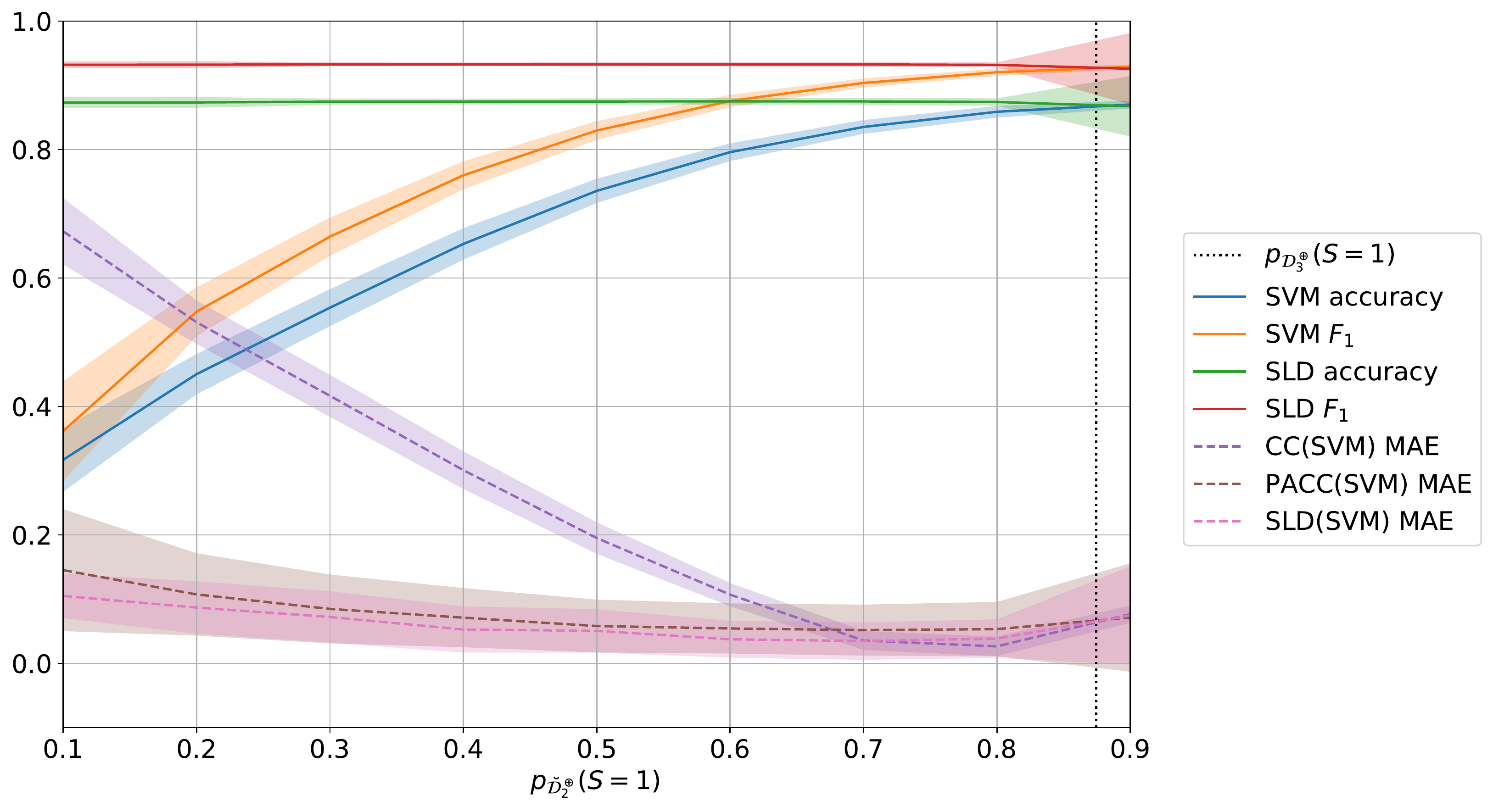}
    \caption{Protocol \texttt{sample-prev-$\mathcal{D}_{2}^{\oplus}$}}
    \label{fig:sample_prev_d21_q_wo_c_svm}
  \end{subfigure} \\
  \caption{Performance of SVM-based methods CC, SLD and PACC on the
  Adult dataset when used for quantification (MAE -- lower is better)
  and classification ($F_1$, accuracy -- higher is better) under
  protocol \texttt{sample-prev-$\mathcal{D}_2$}. The classification
  performance of PACC is equivalent to that of CC (both equal to the
  performance of the underlying SVM), and we thus omit it for
  readability.}
  \label{fig:q_wo_c_svm_d2}
\end{figure}

\begin{figure}[h!]%[h!]
  \begin{center}
    \begin{subfigure}{\clfplotlength}
      \centering
      \includegraphics[width=\mysize\textwidth]{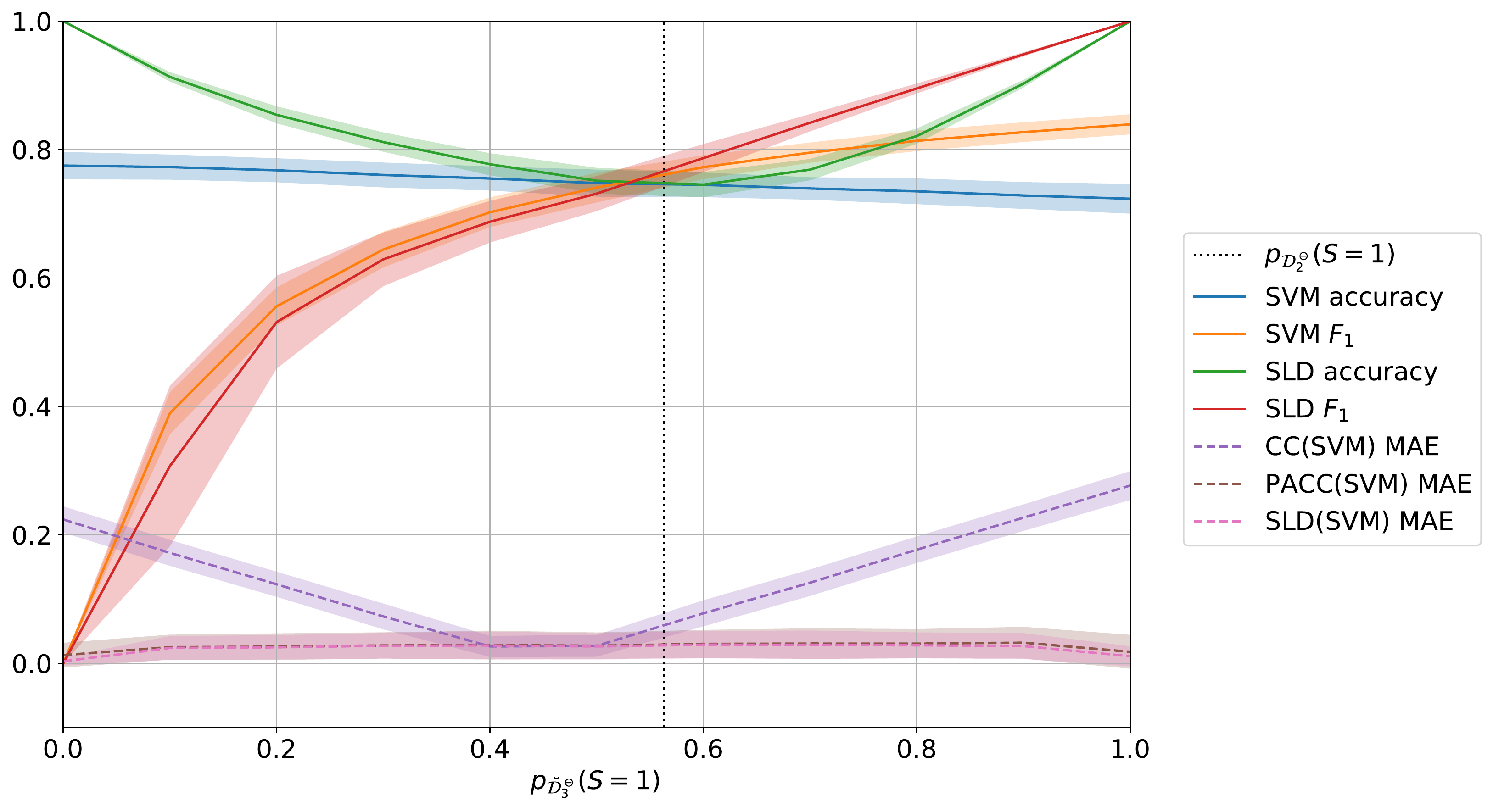}
      \caption{Protocol
      \texttt{sample-prev-$\mathcal{D}_{3}^{\ominus}$}}
      \label{fig:sample_prev_d30_q_wo_c_svm}
    \end{subfigure} \\
    \begin{subfigure}{\clfplotlength}
      \centering
      \includegraphics[width=\mysize\textwidth]{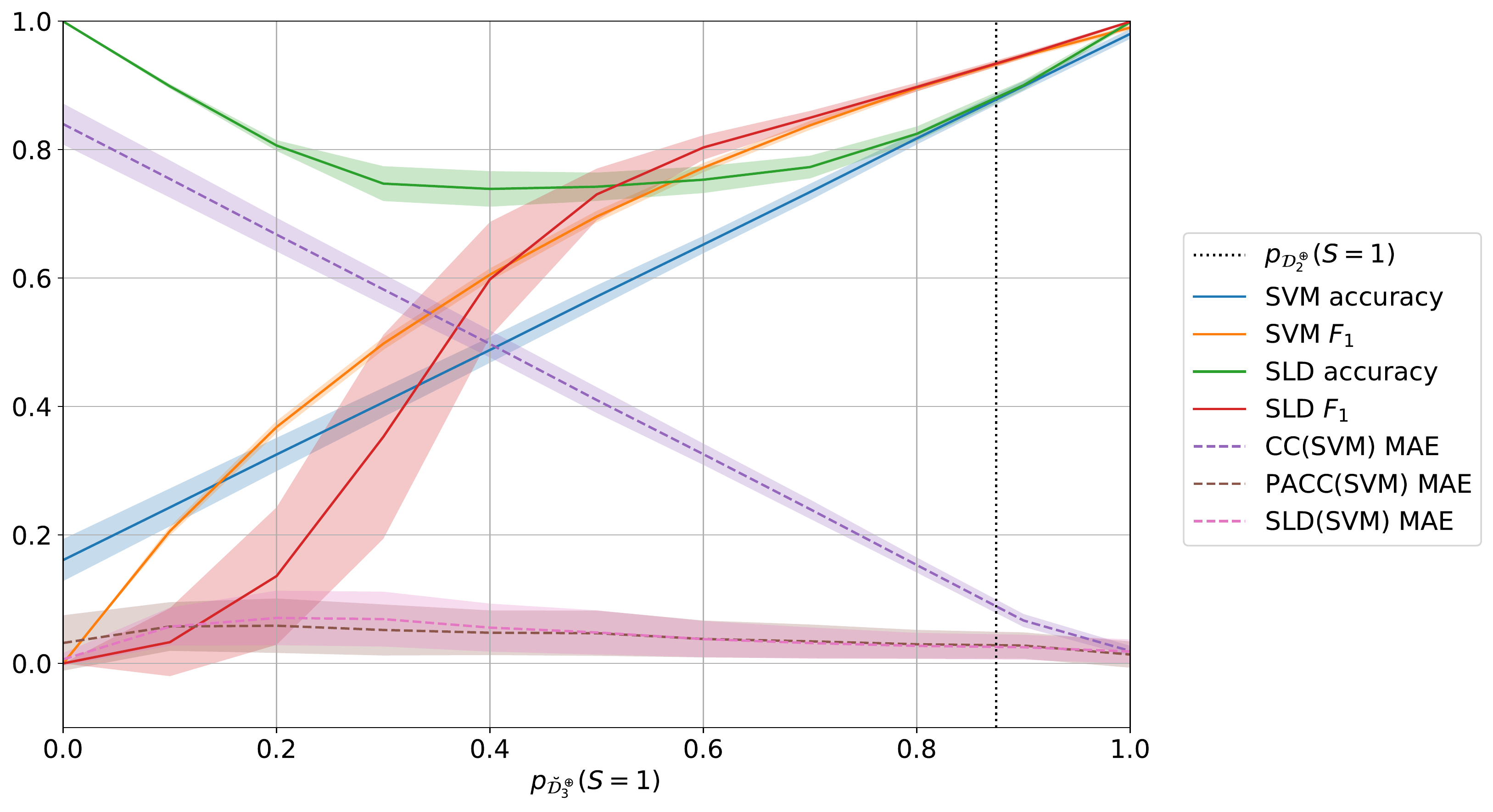}
      \caption{Protocol
      \texttt{sample-prev-$\mathcal{D}_{3}^{\oplus}$}}
      \label{fig:sample_prev_d31_q_wo_c_svm}
    \end{subfigure}%
  \end{center}
  \caption{Performance of SVM-based methods CC, SLD and PACC on the
  Adult dataset when used for quantification (MAE -- lower is better)
  and classification ($F_1$, accuracy -- higher is better) under
  protocol \texttt{sample-prev-$\mathcal{D}_3$}. The classification
  performance of PACC is equivalent to that of CC (both equal to the
  performance of the underlying SVM), and we thus omit it for
  readability.}
  \label{fig:q_wo_c_svm_d3}
\end{figure}

% -------------------------------------------------------------------

\clearpage
\newpage

\section{Pseudocode}

\noindent This section reports pseudocode for protocols
\texttt{sample-prev-$\mathcal{D}_2$} (Pseudocode
\ref{pseudo:sample-prev-d2}), \texttt{sample-size-$\mathcal{D}_2$}
(Pseudocode \ref{pseudo:sample-size-d2}), and
\texttt{sample-prev-$\mathcal{D}_1$} (Pseudocode
\ref{pseudo:sample-prev-d1}).

\begin{algorithm}[h!]
  \LinesNumbered \SetNoFillComment
 
 \begin{footnotesize}
   \SetKwInOut{Input}{Input} \SetKwInOut{Output}{Output} \Input{
   \textbullet\ Dataset $\mathcal{D}$ ; \\
   \hspace{0.35em}\textbullet\ Classifier learner CLS; \\
   \hspace{0.35em}\textbullet\ Quantification method Q; \\
   } \Output{
   \textbullet\ MAE of the demographic disparity estimates ; \\
   \hspace{0.35em}\textbullet\ MSE of the demographic disparity
   estimates ;}
 
   \BlankLine

   $E\leftarrow\emptyset$ ;
 
   \For{5 random splits}{
   $\mathcal{D}_A,\mathcal{D}_B,\mathcal{D}_C \leftarrow
   \mathrm{split\_stratify}(\mathcal{D})$ ;
 
   \For{$\mathcal{D}_1,\mathcal{D}_2,\mathcal{D}_3 \in
   \mathrm{permutations}(\mathcal{D}_A,\mathcal{D}_B,\mathcal{D}_C)$}{
 
   \tcc{Learn a classifier $h : \mathcal{X} \rightarrow \mathcal{Y}$}
   $h\leftarrow$ CLS.fit($\mathcal{D}_1$) ;
 
   $\mathcal{D}_{2}^{\ominus}\leftarrow\{(\mathbf{x}_i,s_i)\in\mathcal{D}_2
   \;|\; h(\mathbf{x}_i)=\ominus\}$ ;
 
   $\mathcal{D}_{2}^{\oplus}\leftarrow\{(\mathbf{x}_i,s_i)\in\mathcal{D}_2
   \;|\; h(\mathbf{x}_i)=\oplus\}$ ;
 
   \For{10 repeats}{

   \tikzmk{A} \For{$p \in \{0.1,0.2,\ldots,0.9\}$}{
 
   \tcc{Generate samples from $\mathcal{D}_{2}^{\ominus}$ at desired
   prevalence and size, and uniform samples from
   $\mathcal{D}_{2}^{\oplus}$ at desired size}
 
   $\breve{\mathcal{D}}_{2}^{\ominus}\sim\mathcal{D}_{2}^{\ominus}$
   with $p_{\breve{\mathcal{D}}_{2}^{\ominus}}(s)=p$ and
   $|\breve{\mathcal{D}}_{2}^{\ominus}|=500$ \label{line:randomsample:d20}
   ;
 
   $\breve{\mathcal{D}}_{2}^{\oplus}\sim\mathcal{D}_{2}^{\oplus}$ with
   $|\breve{\mathcal{D}}_{2}^{\oplus}|=500$ \label{line:randomsample:d21}
   ;
 
   \tcc{Learn quantifiers $q_y : 2^\mathcal{X} \rightarrow [0,1]$}

   $q_{\ominus} \leftarrow $
   Q.fit($\breve{\mathcal{D}}_{2}^{\ominus})$ ;
 
   $q_{\oplus} \leftarrow $ Q.fit($\breve{\mathcal{D}}_{2}^{\oplus})$
   ;
 
   \tikzmk{B} \boxit{pink} \tcc{Use quantifiers to estimate
   demographic prevalence}
   $\mathcal{D}_{3}^{\ominus}\leftarrow\{\mathbf{x}_i\in\mathcal{D}_3
   \;|\; h(\mathbf{x}_i)=\ominus\}$ ;
 
   $\mathcal{D}_{3}^{\oplus}\leftarrow\{\mathbf{x}_i\in\mathcal{D}_3
   \;|\; h(\mathbf{x}_i)=\oplus\}$ ;
 
   $\hat{p}^{q_{\ominus}}_{\mathcal{D}_{3}^{\ominus}}(s) \leftarrow
   q_{\ominus}(\mathcal{D}_{3}^{\ominus})$ ;
 
   $\hat{p}^{q_{\oplus}}_{\mathcal{D}_{3}^{\oplus}}(s) \leftarrow
   q_{\oplus}(\mathcal{D}_{3}^{\oplus})$ ;
 
   \tcc{Compute the signed error of the demographic disparity
   estimate} $e \leftarrow $ compute error using
   $\hat{p}^{q_{\ominus}}_{\mathcal{D}_{3}^{\ominus}}(s)$,
   $\hat{p}^{q_{\oplus}}_{\mathcal{D}_{3}^{\oplus}}(s)$ and Equation
   \ref{eq:estim_err1}
   % $e=\left [ \hat{p}_{\mathcal{D}_{3}^{\ominus}}^{q_{\ominus}}(s)
   %   -
   %   \hat{p}_{\mathcal{D}_{3}^{\oplus}}^{q_{\oplus}}(s) \right ] -
   % \left [ p_{\mathcal{D}_{3}^{\ominus}}(s) -
   %   p_{\mathcal{D}_{3}^{\oplus}}(s) \right ]$ ;
 
   $E \leftarrow E \cup \{e\}$}}}}
 
   $\mathrm{mae}\leftarrow\mathrm{MAE}(E)$ ;
 
   $\mathrm{mse}\leftarrow\mathrm{MSE}(E)$ ;

   \Return{$\mathrm{mae}$, $\mathrm{mse}$} \BlankLine

   \caption{Protocol \texttt{sample-prev-$\mathcal{D}_2$}, shown for
   variations of prevalence values in class $y=\ominus$.}
   \label{pseudo:sample-prev-d2}
 \end{footnotesize}
\end{algorithm}

\begin{algorithm}[t]
  \LinesNumbered \SetNoFillComment
 
  % \SetAlgorithmName{Protocol}{List of protocols}
 
 \begin{footnotesize}
   \SetKwInOut{Input}{Input} \SetKwInOut{Output}{Output} \Input{
   \textbullet\ Dataset $\mathcal{D}$ ; \\
   \hspace{0.35em}\textbullet\ Classifier learner CLS; \\
   \hspace{0.35em}\textbullet\ Quantification method Q; \\
   } \Output{
   \textbullet\ MAE of the demographic disparity estimates ; \\
   \hspace{0.35em}\textbullet\ MSE of the demographic disparity
   estimates ;}
 
   \BlankLine

   $E\leftarrow\emptyset$ ;
 
   \For{5 random splits\label{line:randomsplit}}{
   $\mathcal{D}_A,\mathcal{D}_B,\mathcal{D}_C \leftarrow
   \mathrm{split\_stratify}(\mathcal{D})$ ;
 
   \For{$\mathcal{D}_1,\mathcal{D}_2,\mathcal{D}_3 \in
   \mathrm{permutations}(\mathcal{D}_A,\mathcal{D}_B,\mathcal{D}_C)$}{
 
   \tcc{Learn a classifier $h : \mathcal{X} \rightarrow \mathcal{Y}$}
   $h\leftarrow$ CLS.fit($\mathcal{D}_1$) ;
 
   \For{10 repeats\label{line:10repeats}}{
 
   \tikzmk{A} \For{size $s \in $ logspace(\textbf{from:} 1000,
   \textbf{to:} $|\mathcal{D}_2|$, \textbf{steps:} 5)}{
 
   \tcc{Generate samples from $\mathcal{D}_2$ at desired size}
 
   $\breve{\mathcal{D}}_2\sim\mathcal{D}_2$ with
   $|\breve{\mathcal{D}}_2|=s$ \label{line:sample:p2};
 
   \tcc{Learn quantifiers $q_y : 2^\mathcal{X} \rightarrow [0,1]$}
   $\breve{\mathcal{D}}_{2}^{\ominus}\leftarrow\{(\mathbf{x}_i,s_i)\in\breve{\mathcal{D}}_2
   \;|\; h(\mathbf{x}_i)=\ominus\}$ ;
 
   $\breve{\mathcal{D}}_{2}^{\oplus}\leftarrow\{(\mathbf{x}_i,s_i)\in\breve{\mathcal{D}}_2
   \;|\; h(\mathbf{x}_i)=\oplus\}$ ;
 
   $q_{\ominus} \leftarrow $
   Q.fit($\breve{\mathcal{D}}_{2}^{\ominus})$ ;
 
   $q_{\oplus} \leftarrow $ Q.fit($\breve{\mathcal{D}}_{2}^{\oplus})$
   ;
 
   \tikzmk{B} \boxit{pink} \tcc{Use quantifiers to estimate
   demographic prevalence}
   $\mathcal{D}_{3}^{\ominus}\leftarrow\{\mathbf{x}_i\in\mathcal{D}_3
   \;|\; h(\mathbf{x}_i)=\ominus\}$ ;
 
   $\mathcal{D}_{3}^{\oplus}\leftarrow\{\mathbf{x}_i\in\mathcal{D}_3
   \;|\; h(\mathbf{x}_i)=\oplus\}$ ;
 
   $\hat{p}^{q_{\ominus}}_{\mathcal{D}_{3}^{\ominus}}(s) \leftarrow
   q_{\ominus}(\mathcal{D}_{3}^{\ominus})$ ;
 
   $\hat{p}^{q_{\oplus}}_{\mathcal{D}_{3}^{\oplus}}(s) \leftarrow
   q_{\oplus}(\mathcal{D}_{3}^{\oplus})$ ;
 
   \tcc{Compute the signed error of the demographic disparity
   estimate} $e \leftarrow $ compute error using
   $\hat{p}^{q_{\ominus}}_{\mathcal{D}_{3}^{\ominus}}(s)$,
   $\hat{p}^{q_{\oplus}}_{\mathcal{D}_{3}^{\oplus}}(s)$ and Equation
   \ref{eq:estim_err1}
   % $e=\left [ \hat{p}_{\mathcal{D}_{3}^{\ominus}}^{q_{\ominus}}(s)
   %   -
   %   \hat{p}_{\mathcal{D}_{3}^{\oplus}}^{q_{\oplus}}(s) \right ] -
   % \left [ p_{\mathcal{D}_{3}^{\ominus}}(s) -
   %   p_{\mathcal{D}_{3}^{\oplus}}(s) \right ]$ ;
 
   $E \leftarrow E \cup \{e\}$}}}}
 
   $\mathrm{mae}\leftarrow\mathrm{MAE}(E)$ ;
 
   $\mathrm{mse}\leftarrow\mathrm{MSE}(E)$ ;

   \Return{$\mathrm{mae}$, $\mathrm{mse}$} \BlankLine

   \caption{Protocol \texttt{sample-size-$\mathcal{D}_2$}.}
   \label{pseudo:sample-size-d2}
 \end{footnotesize}
\end{algorithm}

\begin{algorithm}[t]
  \LinesNumbered \SetNoFillComment
 
  % \SetAlgorithmName{Protocol}{List of protocols}
 
 \begin{footnotesize}
   \SetKwInOut{Input}{Input} \SetKwInOut{Output}{Output} \Input{
   \textbullet\ Dataset $\mathcal{D}$ ; \\
   \hspace{0.35em}\textbullet\ Classifier learner CLS; \\
   \hspace{0.35em}\textbullet\ Quantification method Q; \\
   } \Output{
   \textbullet\ MAE of the demographic disparity estimates ; \\
   \hspace{0.35em}\textbullet\ MSE of the demographic disparity
   estimates ;}
 
   \BlankLine

   $E\leftarrow\emptyset$ ;
 
   \For{5 random splits}{
   $\mathcal{D}_A,\mathcal{D}_B,\mathcal{D}_C \leftarrow
   \mathrm{split\_stratify}(\mathcal{D})$ ;
 
   \For{$\mathcal{D}_1,\mathcal{D}_2,\mathcal{D}_3 \in
   \mathrm{permutations}(\mathcal{D}_A,\mathcal{D}_B,\mathcal{D}_C)$}{
 
   \For{10 repeats\label{line:10repeatsB}}{
   \For{$p\in\{0.0, 0.1, \ldots, 0.9, 1.0\}$}{
 
   \tikzmk{A} \tcc{Generate samples from $\mathcal{D}_1$ at desired
   prevalence}
 
   $\breve{\mathcal{D}}_1\sim\mathcal{D}_1$ with $P(Y=S)=p$ and
   $|\breve{\mathcal{D}}_1|=500$ \label{line:sample:p1};
 
   \tcc{Learn a classifier $h : \mathcal{X} \rightarrow \mathcal{Y}$}
   $h\leftarrow$ CLS.fit($\breve{\mathcal{D}}_1$) ;
 
   \tikzmk{B} \boxit{pink} \tcc{Learn quantifiers
   $q_y : 2^\mathcal{X} \rightarrow [0,1]$}
   $\mathcal{D}_{2}^{\ominus}\leftarrow\{(\mathbf{x}_i,s_i)\in\mathcal{D}_2
   \;|\; h(\mathbf{x}_i)=\ominus\}$ ;
 
   $\mathcal{D}_{2}^{\oplus}\leftarrow\{(\mathbf{x}_i,s_i)\in\mathcal{D}_2
   \;|\; h(\mathbf{x}_i)=\oplus\}$ ;
 
   $q_{\ominus} \leftarrow $ Q.fit($\mathcal{D}_{2}^{\ominus})$ ;
 
   $q_{\oplus} \leftarrow $ Q.fit($\mathcal{D}_{2}^{\oplus})$ ;
 
   \tcc{Use quantifiers to estimate demographic prevalence}
   $\mathcal{D}_{3}^{\ominus}\leftarrow\{\mathbf{x}_i\in\mathcal{D}_3
   \;|\; h(\mathbf{x}_i)=\ominus\}$ ;
 
   $\mathcal{D}_{3}^{\oplus}\leftarrow\{\mathbf{x}_i\in\mathcal{D}_3
   \;|\; h(\mathbf{x}_i)=\oplus\}$ ;
 
   $\hat{p}^{q_{\ominus}}_{\mathcal{D}_{3}^{\ominus}}(s) \leftarrow
   q_{\ominus}(\mathcal{D}_{3}^{\ominus})$ ;
 
   $\hat{p}^{q_{\oplus}}_{\mathcal{D}_{3}^{\oplus}}(s) \leftarrow
   q_{\oplus}(\mathcal{D}_{3}^{\oplus})$ ;
 
   \tcc{Compute the signed error of the demographic disparity
   estimate} $e \leftarrow $ compute error using
   $\hat{p}^{q_{\ominus}}_{\mathcal{D}_{3}^{\ominus}}(s)$,
   $\hat{p}^{q_{\oplus}}_{\mathcal{D}_{3}^{\oplus}}(s)$ and Equation
   \ref{eq:estim_err1}
   % $e=\left [ \hat{p}_{\mathcal{D}_{3}^{\ominus}}^{q_{\ominus}}(s)
   %   -
   %   \hat{p}_{\mathcal{D}_{3}^{\oplus}}^{q_{\oplus}}(s) \right ] -
   % \left [ p_{\mathcal{D}_{3}^{\ominus}}(s) -
   %   p_{\mathcal{D}_{3}^{\oplus}}(s) \right ]$ ;
 
   $E \leftarrow E \cup \{e\}$}}}}
 
   $\mathrm{mae}\leftarrow\mathrm{MAE}(E)$ ;
 
   $\mathrm{mse}\leftarrow\mathrm{MSE}(E)$ ;

   \Return{$\mathrm{mae}$, $\mathrm{mse}$} \BlankLine

   \caption{Protocol \texttt{sample-prev-$\mathcal{D}_1$}.}
   \label{pseudo:sample-prev-d1}
 \end{footnotesize}
\end{algorithm}

% \begin{appendices}
%   \section{Protocol \texttt{sample-prev-$\mathcal{D}_1$} on COMPAS}

%   \begin{figure}[tb]
%     \centering
%     \includegraphics[width=\textwidth]{./plots/compas/protD1_vary_prev_d1.pdf}
%     \caption{Protocol \texttt{sample-prev-$\mathcal{D}_1$} on the
%     Adult dataset.}
%     \label{fig:sample_prev_d1_compas}
%   \end{figure}

%   -------------------------------------------------------------------

%\end{document}

% --------------------------------------------------------------------

% -------------------------------------------------------------------------

\clearpage
\newpage
\vskip 0.2in \bibliography{biblionew} \bibliographystyle{plainnat}

\end{document}